\documentclass[%
 reprint,
superscriptaddress,
 amsmath,amssymb,
 aps,
 pra,
longbibliography
]{revtex4-2}

\usepackage{graphicx}
\usepackage{dcolumn}
\usepackage{bm}

\usepackage[utf8]{inputenc}
\usepackage[english]{babel}
\usepackage[T1]{fontenc}

\usepackage{tikz}
\usepackage{lipsum}
\usepackage[colorlinks,linkcolor=blue,citecolor=blue]{hyperref}

\usepackage{amsmath,amsfonts,amssymb,amsthm,ascmac}
\usepackage{mathtools}
\usepackage{mathrsfs}
\usepackage{siunitx}
\usepackage{bm}
\usepackage{cancel}

\usepackage{enumerate}
\usepackage{siunitx}
\usepackage{physics}
\usepackage{pxrubrica}
\usepackage[version=3]{mhchem}
\usepackage{array}
\usepackage{float}

\usepackage{booktabs}
\usepackage{multirow}
\usepackage{hhline}
\usepackage{subcaption}
\usepackage{graphicx}
\usepackage{tikz}
\usetikzlibrary{intersections,calc,patterns,through,positioning,arrows,backgrounds,cd}
\usepackage{pgfplots}
\usepgfplotslibrary{patchplots}
\pgfplotsset{compat=1.15}
\usepackage{titlesec}
\usepackage{picture}
\usepackage{fancybox}
\usepackage{boites}
\usepackage{tcolorbox}
\usepackage{dsfont}
\usepackage{quantikz}
\usepackage{ulem}

\usepackage{algorithm}
\usepackage{algpseudocode}
\usepackage{listings}
\usepackage{mathtools}
\usepackage{amsthm}
\usepackage{xcolor}
\usepackage{fancybox}
\usepackage{stmaryrd}
\usepackage{ytableau}
\usepackage{dsfont}
\usepackage{physics}
\usepackage[justification=raggedright, singlelinecheck=true]{caption}
\usepackage{tabularx}
\usepackage{ragged2e}
\usepackage[english]{babel}

\usepackage{todonotes}

\ytableausetup{boxsize=1em}

\newtheorem{Def}{Definition}

\newtheorem{theorem}{Theorem}

\newtheorem{lemma}{Lemma}
\newtheorem{problem}{Problem}
\theoremstyle{remark}
\newtheorem{Remark}{Remark}

\newcommand{\1}{\mathds{1}}

\newcommand{\mac}[1]{\mathcal{#1}}
\newcommand{\mtr}[1]{\mathrm{#1}}

\newcommand{\ot}{\otimes}

\DeclareRobustCommand{\erase}{\bgroup\markoverwith{\textcolor{red}{\rule[.5ex]{2pt}{0.4pt}}}\ULon}
\newcommand{\coloneq}{:=}

\begin{document}

\renewcommand{\thetable}{\arabic{table}}
\renewcommand{\thefigure}{\arabic{figure}}
\setcounter{equation}{0}
\setcounter{table}{0}
\setcounter{figure}{0}
\setcounter{Def}{0}
\setcounter{page}{1}

\onecolumngrid

\title{ \mbox{Comprehensive Study on Heisenberg-limited Quantum Algorithms} \\
for Multiple Observables Estimation}

\author{Yuki Koizumi}
\email{koizumiyuki903@gmail.com}
 \affiliation{Department of Applied Physics, University of \mbox{Tokyo, 7-3-1 Hongo, Bunkyo-ku, Tokyo 113-8656, Japan}}
\author{Kaito Wada}%
\affiliation{\mbox{Graduate School of Science and Technology, Keio University, 3-14-1 Hiyoshi, Kohoku, Yokohama, Kanagawa, 223-8522, Japan}}

\author{Wataru Mizukami}
\affiliation{\mbox{Center for Quantum Information and Quantum Biology, Osaka
University, 1-2 Machikaneyama, Toyonaka, Osaka 560-0043, Japan}}
\affiliation{Graduate School of Engineering Science, \mbox{Osaka University, 1-3 Machikaneyama, Toyonaka, Osaka 560-8531, Japan}}

\author{Nobuyuki Yoshioka}
\email{ny.nobuyoshioka@gmail.com}
\affiliation{\mbox{International Center for Elementary Particle Physics, The University of Tokyo, 7-3-1 Hongo, Bunkyo-ku, Tokyo 113-0033, Japan}}

\begin{abstract}
In the accompanying paper of arXiv:2505.00697, we have presented a generalized scheme of adaptive quantum gradient estimation (QGE) algorithm, and further proposed two practical variants which not only achieve doubly quantum enhancement in query complexity regarding estimation precision and number of observables, but also enable minimal cost to estimate $k$-RDMs in fermionic systems among existing quantum algorithms. Here, we provide full descriptions on the algorithm, and provide theoretical guarantee for the estimation precision in terms of the root mean squared error. 
Furthermore, we analyze the performance of the quantum amplitude estimation algorithm, another variant of the Heisenberg-limited scaling algorithm, and show how the estimation error is minimized under the circuit structure that resembles the phase estimation algorithm.
We finally describe the details for the numerical evaluation of the query complexity of the Heisenberg-limited algorithms and sampling-based methods to make a thorough comparison in the task of estimating fermionic $k$-RDMs.
\end{abstract}

\maketitle
\onecolumngrid

\section{Introduction} \label{sec:introduction}
Collecting information of observables from quantum system is one of the core subroutines that undermine various subfields of quantum information theory such as quantum simulation~\cite{lloyd1996universal, aspuru2005simulated},   quantum machine learning~\cite{biamonte2017quantum}, and quantum benchmarking~\cite{knill2007optimal, Erhard2019Characterizing, Eisert2020Quantum}. 
Especially when we aim for high accuracy,  it is known that there is a quadratic gap between standard sampling and quantum-enhanced 
estimation, and correspondingly,
the minimum level of achievable uncertainty in two schemes are discriminated as the standard quantum limit (SQL) and the Heisenberg limit (HL). It has been of keen interest for researchers to answer what the best strategy to reach such scaling is, and what resource is required for such methods, not least because of its importance in that it is the ultimate limitation due to quantum mechanical principle, but also from practical interest. 


When the quantum computers are not capable of running deep circuits, information is extracted via sampling in the Pauli basis by projective measurement. Since this approach is hardware friendly, numerous ideas have been proposed with the scope of implementation in existing devices, including mapping into randomized measurement~\cite{enk2012measuring, elben2018renyi,elben2019statistical, Huang:2020tih, elben2023randomized}, partial tomography using local basis transformation~\cite{cotler2020quantum, bonet2020nearly}, and adaptive scheme~\cite{garcia2021learning}. A notable algorithm is the classical shadow tomography~\cite{Huang:2020tih} which achieves $O(\log M/\epsilon^2)$ scaling for specific set of observables, where $M$ is the number of observables and $\varepsilon$ is the target precision.
By varying the measurement basis, it is known to achieve near-optimal sample complexity not only for local Pauli observables but also for fermionic~\cite{Zhao:2020vxp, low2022classical} and bosonic~\cite{gu2023efficient} systems.
By allowing use of ancila qubits, one may also achieve $O(\log M/\epsilon^4)$ scaling for various  sets of observables 
with fixed magnitude, e.g., Pauli operators, using the gentle measurement~\cite{huang2021information}.

One of the most renown examples achieving the HL scaling in observable estimation is the quantum amplitude estimation (QAE) algorithm~\cite{brassard2002amplitude,knill2007optimal,rall2020estimating}.
The QAE algorithm was originally proposed by Brassard {\it et al.}~\cite{brassard2002amplitude} as a method to estimate $\|\Pi |\psi\rangle\|^2$ 
for a projector $\Pi$ and a state $|\psi\rangle$ by performing  phase estimation algorithm to a Grover operator consisting of reflections.
Stimulated by proposals to apply the QAE algorithm to Monte Carlo sampling~\cite{montanaro2015quantum}, further works presented advancement in the scaling with respect to the root mean square error (MSE)
and the reduction of constant factors~\cite{suzuki2020amplitude,nakaji2020faster,grinko2021iterative,zhao2022adaptive,rall2023amplitude,fukuzawa2023modified}, although there is no known method that provably achieves the lower bound determined by the genuine HL including its prefactor.

Another thread of exploration for the HL scaling is based on the quantum gradient estimation (QGE) algorithm, which is an extension of the quantum phase estimation (QPE) algorithm to multi-parameter setup~\cite{jordan2005fast, gilyen2019gradient}.
By encoding the information of all the observables into the phase, the QGE algorithm can achieve quadratic improvement in both accuracy and number of non-commuting target observables.
Initial proposals have shown that QGE indeed achieves almost quadratic speedup 
in terms of the additive error~\cite{apeldoorn2023quantum, huggins2022nearly}, while it was later shown that this yields unwanted logarithmic correction 
in a more strict error metric---the root MSE
~\cite{wada2024Heisenberg}.
Such a challenge was overcome by Ref.~\cite{wada2024Heisenberg} which proposed to employ an adaptive scheme; the HL scaling was achieved in terms of the root MSE, with the bonus of reducing the use of ancillary qubits and classical computation to determine the phase factor for quantum signal processing.

In this work, we aim to complement the details of accompanying paper~\cite{koizumi2025letter}, which aim to
answer a crucial question to be answered for practitioners of quantum computing; {\it how does the state-of-the-art estimation algorithms compare with each other in practical problems?} Our contributions are three-fold. First, we prove that the QAE algorithm using sinusoidal amplitude state as the probe achieves near-optimal query complexity in MSE with a significantly small prefactor, under the circuit structure resembling the phase estimation algorithm. Second, we present a generalized framework of the QGE algorithm that leads us to propose two novel variants, later referred to as Method I which incorporates the symmetry inherent in the target state and Method II which further utilizes the parallel scheme in a single-shot manner to accelerate the estimation. Finally, we show results of numerical comparison between state-of-the-art algorithms in challenging problems, namely partial tomography in strongly correlated fermionic systems such as the FeMo cofactor and Fermi-Hubbard models.

The remainder of this paper is organized as follows.
In Sec.~\ref{sec:Prob_setting}, we summarize the problem setting considered in this work.
In Sec.~\ref{sec:est_with_QAE}, we first revisit the quantum amplitude estimation algorithm, and then show how to nearly saturate the lower bound of MSE under a circuit structure similar to the phase estimation algorithm. In Sec.~\ref{sec:qge-general}, we present a generalized framework of adaptive QGE algorithm, together with a brief review on prior works.
Then we proceed to Sec.~\ref{sec:all-symemtry-diven-QGE} where we propose two practical variants: Method I which incorporates symmetry in target system, and Method II which inherits the benefit of the fully-parallel scheme and adaptive scheme.
We provide numerical evaluations of state-of-the-art estimation algorithm in Sec.~\ref{sec:Application}, and finally present the discussion in Sec.~\ref{sec:discussion}.

\section{Problem setting}
\label{sec:Prob_setting}

In this section, we briefly outline the problem setting considered in this work. Our study focuses on estimating the expectation values of $2^N$-dimensional $M$ observables $\{O_j\}_{j=1}^M $ with spectral norm $\|O_j\| \leq 1$ regarding an $N$-qubit pure state prepared as $\ket{\psi} = U_{\psi}|0\rangle^{\otimes N} $.
The Hilbert space dimension is alternatively denoted as $d$.  
Here, we consider quantum algorithms that yield estimators $\hat u_j$ of the expectation values $
\langle O_j \rangle \coloneq \langle \psi \lvert O_j \rvert \psi \rangle,
$
ensuring that the overall mean squared error (MSE) remains below a specified threshold $\varepsilon^2$, i.e., for any $j$, $\hat u_j$ satisfies $\mtr{MSE}[\hat u_j] \coloneq \mathbb{E}[(\hat u_j -  \langle O_j \rangle)^2 ] \leq \varepsilon^2 $.
We choose the root MSE $\varepsilon$ as the target precision measure because it successfully bounds other
measures of uncertainty such as additive error or confidence intervals, while the converse claim does not hold in general
\cite{2009NJPh...11g3023H, PhysRevA.63.053804}.

To facilitate the implementation of the target state and target observables on quantum circuits, we employ a state-preparation oracle access and a block-encoding framework.
The quantum state of interest, $\lvert \psi \rangle$, is assumed to be generated via a state preparation oracle $U_\psi$, which transforms the initial computational basis state $\lvert 0 \rangle^{\otimes N}$ into $\lvert \psi \rangle$, formally described by $
U_\psi : \ket{0}^{\otimes N} \mapsto \lvert \psi \rangle.
$
We further assume oracle access to both $U_\psi$ and its inverse $U_\psi^\dagger$.
Additionally, for each observable $O_j$, we can call to a block-encoding $B_j$ of $O_j$. In other words, we have each corresponding unitary operator $B_j$ acting on the combined system of the $d$-dimensional Hilbert space and an $a$-qubit ancillary space, such that the top-left block of $B_j$ is exactly $O_j$. Formally, we define a block-encoding as follows:
\begin{Def}[Block-encoding]
For positive real numbers \(\alpha, \varepsilon\), a non-negative integer \(a\), and an \(n\)-qubit operator \(A\), an \((n+a)\)-qubit operator \(U_A\) is said to be a \((\alpha, a, \varepsilon)\)-block-encoding of \(A\) if
\begin{eqnarray}
    \norm{A - \alpha \Bigl[(\bra{0}^{\otimes a} \otimes \1_n) U_A (\ket{0}^{\otimes a} \otimes \1_n)\Bigr]} \leq \varepsilon,
\end{eqnarray}
where $\norm{\cdot}$ denotes the spectral norm and $\1_n$ is the identity operator acting on an 
$n$-qubit space.
\end{Def}

Under such a condition, we evaluate the performance of various methods in terms of the query complexity of $U_\psi$ and its inverse. Since $U_\psi$ typically scales with the system size, the efficiency of the quantum algorithm is predominantly determined by the total number of queries to $U_\psi$ (with the cost of $U_\psi^\dagger$ assumed to be comparable). Therefore, it is imperative to design an estimation algorithm that minimizes statistical uncertainty while keeping the query complexity of $U_\psi$ and its inverse as low as possible. Summarizing setups and assumptions, the formal problem is expressed as
\begin{problem} \label{prob:query_estimation}
Let \(N\) and \(M\) be positive integers. Suppose we have oracle access to a state-preparation unitary 
\(U_\psi: \ket{0}^{\otimes N} \mapsto \ket{\psi}\) and its inverse, and assume that for each \(j \in \{1,\ldots,M\}\) we have access to a block-encoding \(B_j\) of a \(2^N\)-dimensional observable \(O_j\). Given a target precision \(\varepsilon \in (0,1)\), the goal is to produce estimators \(\{\hat{u}_j\}_{j=1}^M\) satisfying
\begin{align}
    \forall j \in \{1,\ldots,M\}, \quad \mtr{MSE}[\hat{u}_j] \coloneq \mathbb{E}\Bigl[\Bigl(\hat{u}_j - \bra{\psi} O_j \ket{\psi}\Bigr)^2\Bigr] \le \varepsilon^2,
\end{align}
while minimizing the number of queries to \(U_\psi\) and \(U_\psi^\dagger\).
Here, the expectation value is taken over all possible outcomes of $\{\hat{u}_j\}$.
\end{problem}

It is important to remark that we assume oracular access to the state preparation unitary $U_\psi$ and its inverse, which is a stronger assumption than the case where one can only make a call to the state $|\psi\rangle$ itself. The sampling-based methods such as the classical shadow tomography can be applied under both cases, while numerous algorithms with HL can only be applied when $U_\psi$ and its inverse are available.


We also note that although it is not standard practice, we often write $\mathcal{O}(x/\varepsilon)$ in the form $\mathcal{O}(x)/\varepsilon$ for asymptotic analysis. This notation serves to distinguish between $\tilde{\mathcal{O}}(x/\varepsilon)$—which carries additional logarithmic factors in $\varepsilon^{-1}$—and $\tilde{\mathcal{O}}(x)/\varepsilon$, which has no such terms. 


\section{Quantum amplitude estimation algorithm}

In Sec.~\ref{sec:HL_amplitude}, we revisit the quantum amplitude estimation algorithm,
and then prove that the MSE is nearly minimized by employing a {\it sine state}, which is a probe state with its amplitude expressed by a sinusoidal function.%
Then, in Sec.~\ref{sec:Expectation_base_QAE}, 
we study the complexity to extract information of a given physical observable based on the nearly-optimal QAE algorithm. 

\label{sec:est_with_QAE}

\subsection{Heisenberg-limited scaling quantum amplitude estimation algorithm}
\label{sec:HL_amplitude}

Quantum Amplitude Estimation (QAE) algorithm \cite{brassard2002amplitude} stands as a cornerstone quantum algorithm that offers a quadratic speedup in query complexity over classical methods in estimating an amplitude of a quantum state. 
Concretely, the goal of the QAE algorithm is to estimate an amplitude $a \in [0,1]$ when given access to a quantum oracle $\mac A $ defined by
\begin{align}
    \mathcal{A} \ket{0}^{\otimes N} \ket{0} = \sqrt{1-a}\,\ket{\Psi_0} \ket{0} + \sqrt{a}\,\ket{\Psi_1} \ket{1}
\end{align}
 where $N$ is a positive integer, 
 $\ket{\Psi_1}$ and $\ket{\Psi_0}$ are arbitrary normalized quantum states. 
The original QAE paper \cite{brassard2002amplitude} employed additive error as the error metric and constructed a quantum circuit based on the quantum phase estimation (QPE) algorithm. In this approach, a quantum circuit is divided into two main components: the system register and the probe register. The system register stores the quantum state that encodes the target amplitude, and the probe register extracts the information of amplitude into the complex phase by   coherently interacting with the system register.
After a sequence of interactions, one measures the probe system from which the phase information is extracted.
 
Similar to the case of the QPE algorithm, it has been recognized that the performance of the algorithm heavily relies on the concrete structure of the probe state.
While the state of uniform superposition over all computational basis 
was originally used in Ref.~\cite{brassard2002amplitude},
it was pointed out in the context of QPE that such a state fails to achieve the HL scaling~\cite{berry2009perform, ji2008parameter}. Thus, QPE-based amplitude estimation with a superposition state is also expected not to achieve the HL scaling in terms of MSE (see discussion around Eq.~\eqref{eq:superposition-mse}). 
Some studies have proposed how to overcome this issue~\cite{Suzuki:2020amp, PRXQuantum.2.010346}, while it remains open how we can implement the QAE with minimal queries to state preparation oracle. 
This motivates the need to optimize the initial state of the probe register and the lower bound of the MSE.

\subsubsection{Sine state enables HL-scaling QAE}

To address the above problem, we show that QPE-based amplitude estimation using a \textit{sine state} attains the HL scaling, and that it is almost optimal regarding the query to the oracle $\mathcal{A}$.

\begin{lemma}
\label{lem:HL_scaling_QAE}
Let \(q\in\mathbb{N}\). For a given quantum oracle 
\begin{align}
    \mathcal{A} \ket{0}^{\otimes N} \ket{0} = \sqrt{1-a}\,\ket{\Psi_0} \ket{0} + \sqrt{a}\,\ket{\Psi_1} \ket{1},
\end{align}
with arbitrary \(N\)-qubit states \(\ket{\Psi_0}\), \(\ket{\Psi_1}\), and \(a\in[0,1]\), we consider the quantum amplitude estimation circuit that uses the \(q\)-qubit sine state,
\begin{align}
    \label{eq:def_sin_state}
    \sqrt{\frac{2}{2^q}} \sum_{k=0}^{2^q-1} \sin\Bigl(\frac{k \pi}{2^q}\Bigr) \ket{k},
\end{align}
as depicted in Fig.~\ref{fig:heisenberg-limited_improved_circuit}. When this circuit yields the measurement outcome \(l\) in the computational basis for \(l \in \{0,1,\ldots,2^q-1\}\), we set the estimator \(\hat{a} \coloneq \sin^2(l \pi/2^q)\) for the amplitude \(a\). Then, such an estimator $\hat{a}$ satisfies
\begin{align}
    \widehat{\mtr{MSE} } [\hat{a}] = \Bigl(\frac{\pi}{2^{q+1}}\Bigr)^2 + \mac{O}\Bigl(\Bigl(\frac{1}{2^q}\Bigr)^3\Bigr),
\end{align}
with $2^q+1$ uses of $\mac{A}$ and its inverse. Here, \( \widehat{\mtr{MSE}[\hat{a}]} \) denotes the supremum of the MSE for $a$, $\mtr{MSE}[\hat a] \coloneq \mathbb{E}[(a-\hat a)^2]$, over the interval $[0,1]$, i.e., the maximum value of MSE over all possible target values $a \in [0,1]$.
\end{lemma}

\begin{figure}[t]
\centering
\includegraphics[width=1.\linewidth]{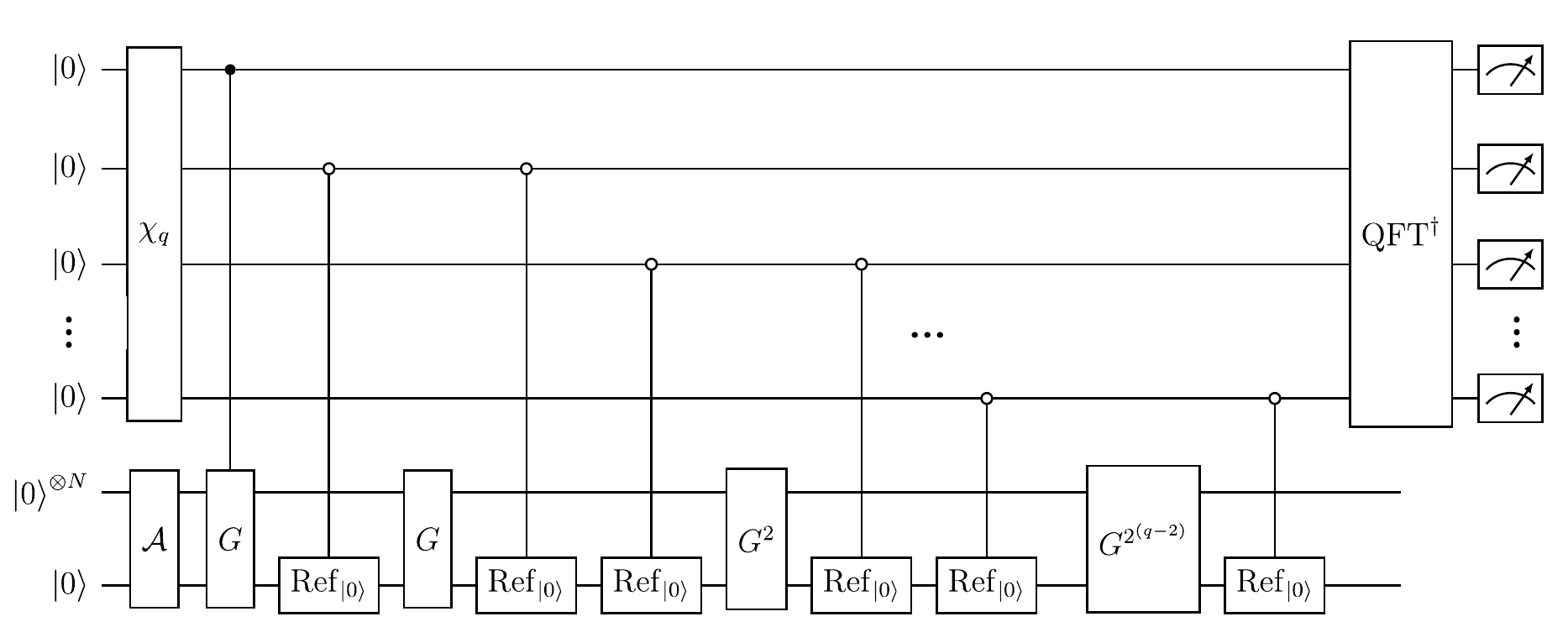}
\caption{
\justifying{
Improved circuit construction of QAE algorithm with an initial state preparation unitary $\chi_q$ inspired by the circuit in Ref.~\cite{Obrien2022efficient}. The gate labeled \(\mtr{QFT}^\dag\) implements the inverse quantum Fourier transform. In this circuit, the \((N+1)\)-qubit oracle \(\mac{A}\) embeds the target amplitude \(a\), and the operator $G \coloneqq \mac{A} \bigl(\1 - 2 \ketbra{0}^{\otimes (N+1)}\bigr) \mac{A}^\dagger \bigl(\1 - 2 \cdot \1_N \otimes \ketbra{0}\bigr)$ is defined to encode $e^{\pm i 2 \arcsin(\sqrt a )}$ on $\mac{A}\ket{0}^{\otimes (N+1)}$ with equal amplitude. Additionally, $\mtr{Ref}_{\ket{0} } $ denotes the reflection operator on $\ket{0} $, i.e., $\mtr{Ref}_{\ket{0} } \coloneq \1 -2 \ketbra{0} $. } }
\label{fig:heisenberg-limited_improved_circuit}
\end{figure}

\begin{figure}[t]
\centering
\includegraphics[width=0.7\linewidth]{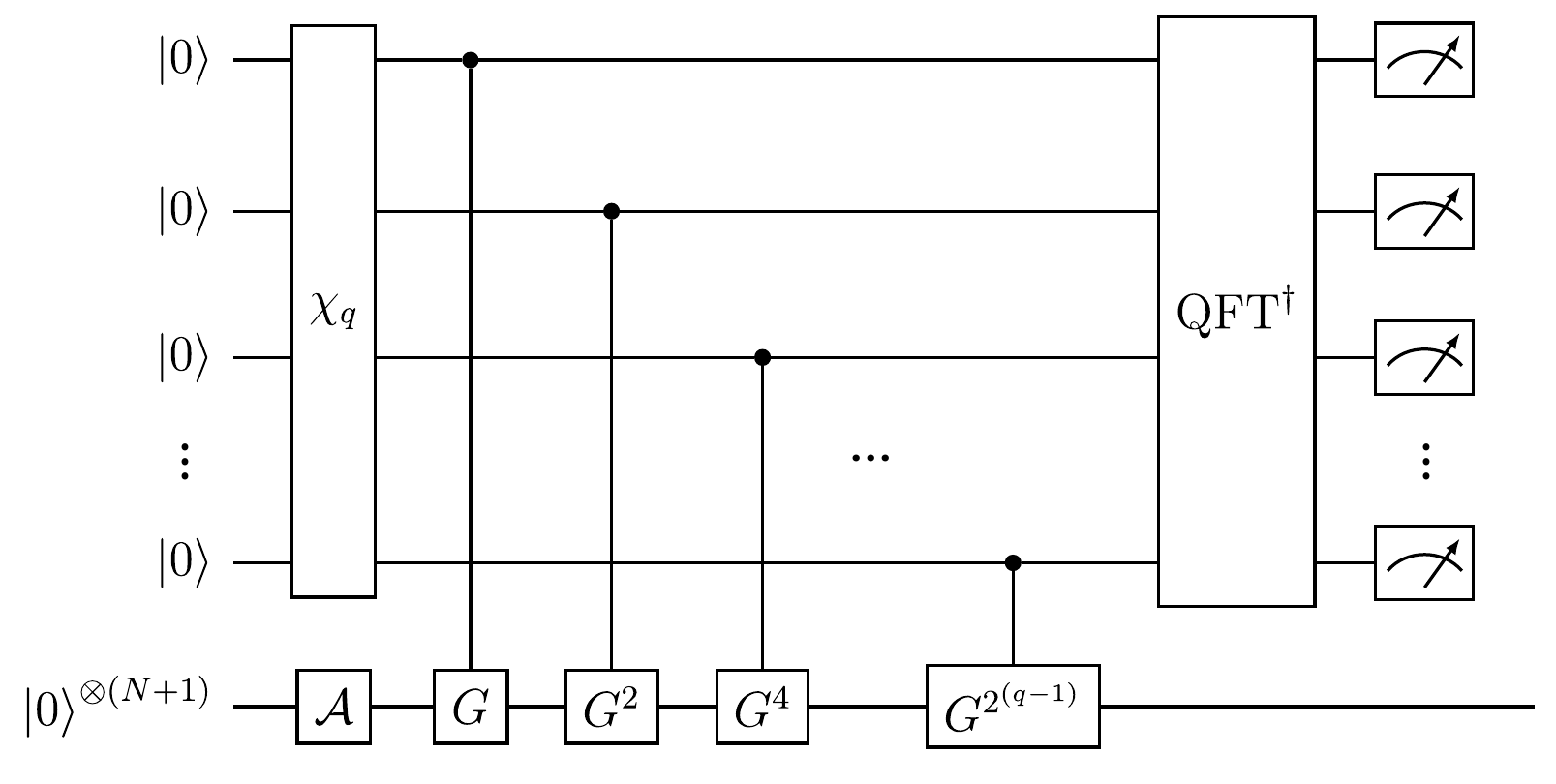}
\caption{
\justifying{
Standard circuit construction of QAE algorithm with an initial state preparation unitary $\chi_q$. The gate labeled \(\mtr{QFT}\) implements the quantum Fourier transform. In this circuit, the \((N+1)\)-qubit oracle \(\mac{A}\) embeds the target amplitude \(a\), and the operator $G \coloneq \mac{A} \bigl(\1 - 2 \ketbra{0}^{\otimes (N+1)}\bigr) \mac{A}^\dagger \bigl(\1 - 2 \cdot \1_N \otimes \ketbra{0}\bigr)$ is defined to encode $e^{\pm i 2 \arcsin(\sqrt a )}$ on $\mac{A}\ket{0}^{\otimes (N+1)}$ with equal amplitude. }}
\label{fig:heisenberg-limited_standard_circuit}
\end{figure}

\begin{proof}
To establish that the circuit in Fig.~\ref{fig:heisenberg-limited_improved_circuit} attains the HL scaling with the small constant factor, we proceed as follows. We begin by considering the action of a standard QPE-based amplitude estimation circuit shown in Fig.~\ref{fig:heisenberg-limited_standard_circuit} \cite{brassard2002amplitude}. We then show that, while functionally equivalent to the standard construction, the circuit implementation in Fig.~\ref{fig:heisenberg-limited_improved_circuit} achieves the same outcome with fewer queries \cite{Obrien2022efficient}. Next, we evaluate the probability of obtaining measurement outcome \(l\) in the computational basis, followed by an explicit derivation of \(\mtr{MSE}[\hat{a}]\) for an arbitrary initial state of the probe register. Finally, we specialize to the sine state as the initial state, demonstrating that it achieves Heisenberg-limited scaling with the small multiplicative prefactor.

In the circuit shown in Fig.~\ref{fig:heisenberg-limited_standard_circuit}, we introduce a unitary \(\chi_q\) preparing the $q$-qubit initial state:
\[
\chi_q: \ket{0}^{\otimes q} \mapsto \sum_{k=0}^{2^q-1} \alpha_k \ket{k},\quad \text{with } \sum_{k=0}^{2^q-1} |\alpha_k|^2 = 1.
\]
The Grover operator \(G\) is defined by
\begin{align}
    G \coloneq \mac{A} \bigl(\1 - 2 \ketbra{0}^{\otimes (N+1)}\bigr) \mac{A}^\dagger \bigl(\1 - 2\cdot\1_N\otimes \ketbra{0}\bigr).
\end{align}
When we denote $\theta = \arcsin (\sqrt{a}) \in [0,\pi/2]$, the matrix representation of this operator on the subspace spanned by \(\{\ket{\Psi_0}\ket{0},\,\ket{\Psi_1}\ket{1}\}\) is given by
\begin{align}
    G = \begin{pmatrix}
\cos 2\theta & -\sin 2\theta \\
\sin 2\theta & \cos 2\theta
\end{pmatrix},
\end{align}
which implies that \(G\) has eigenvalues \(e^{\mp i  2\theta}\) and corresponding eigenstates \(\ket{\Psi_\pm}:=(\ket{\Psi_0}\ket{0}\pm i \ket{\Psi_1}\ket{1})/\sqrt{2}\). 
Since we can expand the input state $\mac{A}\ket{0}^{\otimes (N+1)}$ into
\[
\mac{A}\ket{0}^{\otimes (N+1)} = \frac{e^{-i\theta}\ket{\Psi_+} + e^{i\theta}\ket{\Psi_-}}{\sqrt{2}},
\]
the resulting state after a single application of $G$ is given by, 
\begin{align}
    G (\mac{A} \ket{0}^{\otimes (N+1)}) = \frac{e^{-i\theta}e^{-i2\theta}\ket{\Psi_+} + e^{i\theta}e^{i2\theta}\ket{\Psi_-}}{\sqrt{2}}.
\end{align}
Therefore, just before the Quantum Fourier Transformation (QFT), the whole state is
\begin{align}
    \label{eq:state_before_QFT}
  \frac{1}{\sqrt{2}} \Biggl( e^{-i\theta} \sum_{k=0}^{2^q-1} e^{-ik\cdot2\theta}\,\alpha_k \ket{k}\ket{\Psi_+} + e^{i\theta} \sum_{k=0}^{2^q-1} e^{ik\cdot2\theta}\,\alpha_k \ket{k}\ket{\Psi_-} \Biggr).
\end{align}
After applying the inverse QFT, this state becomes
\begin{align}
  \frac{1}{\sqrt{2^{(q+1)}}} \sum^{2^{q}-1}_{l=0} \Biggl( e^{-i\theta} \sum_{k=0}^{2^q-1} e^{-i (k\cdot2\theta + 2\pi k \cdot \frac{l}{2^q} )}\,\alpha_k \ket{l}\ket{\Psi_+} + e^{i\theta} \sum_{k=0}^{2^q-1} e^{i(k\cdot2\theta - 2\pi k \cdot \frac{l}{2^q} )}\,\alpha_k \ket{l}\ket{\Psi_-} \Biggr).
\end{align}
The \(l\)-th measurement outcome in the computational basis occurs with the probability expressed as $P_l^{(+)} + P_l^{(-)}$ where
\begin{align}
    \label{eq:QAE_lth_probability}
    P_l^{(+)} &\coloneq \frac12 \Biggl|\sum_{k=0}^{2^q-1} \frac{\alpha_k}{\sqrt{2^q}}\,  e^{i2\pi\Bigl(1 - \frac{\theta}{\pi} - \frac{l}{2^q}\Bigr)k}\Biggr|^2, \\
    \label{eq:QAE_lth_probability2}
    P_l^{(-)} &\coloneq  \frac12 \Biggl|\sum_{k=0}^{2^q-1} \frac{\alpha_k}{\sqrt{2^q}}\, e^{i2\pi\Bigl(\frac{\theta}{\pi} - \frac{l}{2^q}\Bigr)k}\Biggr|^2,
\end{align}
for \(l\in\{0,1,\ldots,2^q-1\}\). 

Inspired by the QPE circuit implementation in Ref.~\cite{babbush2018encoding}, Ref.~\cite{Obrien2022efficient} proposed to use a circuit shown in Fig.~\ref{fig:heisenberg-limited_improved_circuit}, a variant of the  QAE circuit (Fig.~\ref{fig:heisenberg-limited_standard_circuit}) that halves the total query complexity of \(\mathcal{A}\) and its inverse. 
The key idea is that, 
whenever the \(q'\)-th qubit of the probe register is in the state \(\ket{0}\), we apply \((G^{2^{q'-1}})^\dagger\) instead of the identity operator. 
For simplicity, we introduce 
$
\mtr{Ref}_{\ket{0} } \coloneq 1 - 2 \ketbra{0}, 
$
which satisfies
\begin{eqnarray}
\quad (\1_N \ot \mtr{Ref}_{\ket{0} } ) (G)^n (\1_N \ot \mtr{Ref}_{\ket{0} } ) = ((\1_N \ot \mtr{Ref}_{\ket{0} } )  G (\1_N \ot \mtr{Ref}_{\ket{0} } ) )^n =  (G^\dagger)^n,\forall n \in \mathbb{N}.
\end{eqnarray}
From this relation, the \(n\)-th modified controlled phase gate,
$
\ketbra{0} \otimes (G^n)^\dagger + \ketbra{1} \otimes G^n,
$
can be implemented as
\begin{align}
  \Bigl(\ketbra{0} \ot \1_N \otimes \mtr{Ref}_{\ket{0} } + \ketbra{1} \ot \1_{N+1} \Bigr) \Bigl(\1 \ot G^n \Bigr)  \Bigl(\ketbra{0} \ot \1_N \otimes \mtr{Ref}_{\ket{0} } + \ketbra{1} \ot \1_{N+1} \Bigr).
\end{align}
By using this modified gate, the whole state before the inverse QFT in Fig.~\ref{fig:heisenberg-limited_improved_circuit} is expressed as 
\begin{align}
  \frac{1}{\sqrt{2}} \Biggl( e^{-i(2^{q} -3)\theta} \sum_{k=0}^{2^q-1} e^{-ik\cdot2\theta}\,\alpha_k \ket{k}\ket{\Psi_+} + e^{i(2^{q}-3)\theta} \sum_{k=0}^{2^q-1} e^{ik\cdot2\theta}\,\alpha_k \ket{k}\ket{\Psi_-} \Biggr).
\end{align}
Since no measurement is performed on the system register in the state \(\ket{\Psi^+}\) or \(\ket{\Psi^-}\), we trace out this register at the end. 
As a result, the phase difference regarding $e^{\pm i (2^q - 3)\theta}$ can be disregarded. 
When measuring this probe system after the inverse QFT, the probabilities of obtaining the \(l\)-th outcome are equivalent to those given in Eqs.~\eqref{eq:QAE_lth_probability} and \eqref{eq:QAE_lth_probability2}, respectively.

Regarding the query complexity, 
the modified controlled phase gate for $q'$-th qubit of the probe register requires $2^{q'}$ queries to $G$.
Hence, the total query complexity of our improved circuit is
\begin{align}
    \label{eq:the_query_comp_HLamp}
    1 + 2 + \sum_{q'=1}^{q-1} 2^{q'}  
    \;=\; 2^q + 1.
\end{align}
In comparison, the standard QAE circuit needs \(2^{q+1} - 1\) queries to \(\mathcal{A}\) and its inverse, so this method reduces the query complexity by approximately half.

Since we obtain the expression of the probability of the $l$-th measurement outcome, we next consider a simple estimator and evaluate its MSE.
Following the original approach in \cite{brassard2002amplitude}, a simple estimate \(\hat{\theta}_l \coloneq \frac{\pi l}{2^q}\) is adopted when the \(l\)-th outcome occurs. This phase estimator is then converted into an amplitude estimate via \(\hat{a}=\sin^2\hat{\theta}\). Using this estimator, the maximum MSE of \(\hat{a}\) can be written as
\begin{align}
     \widehat{\mtr{MSE}}[\hat{a}] &= \max_{\theta \in[0,\pi/2] } \mathbb{E}\Bigl[\Bigl(\sin^2\hat{\theta} - \sin^2\theta\Bigr)^2\Bigr] \\
    \label{eq:eval_MSE_amp_with_cos}
    &= \max_{\theta \in[0,\pi/2] }\frac{1}{4}\,\mathbb{E}\Bigl[\Bigl(\cos(2\hat{\theta}) - \cos(2\theta)\Bigr)^2\Bigr] \\
    \label{eq:explict_formula_eval_cos}
    &= \max_{\theta \in[0,\pi/2] } \frac{1}{4}\cos^2(2\theta) - \frac{1}{2}\cos(2\theta)\,\mathbb{E}[\cos(2\hat{\theta})] + \frac{1}{8} + \frac{1}{8}\,\mathbb{E} \qty[\cos(4\hat{\theta})].
\end{align}
Thus, it suffices to evaluate \(\mathbb{E}[\cos(2s\hat{\theta})]\) for \(s=1,2\). 
(Note that if we denote \(\theta = \arccos\sqrt{a}\) instead, the expression in Eq.~\eqref{eq:eval_MSE_amp_with_cos} remains unchanged, since \(\mathbb{E}[(\cos^2\hat{\theta} - \cos^2\theta)^2]\) is equivalent to $(1/4) \cdot \mathbb{E}\Bigl[\Bigl(\cos(2\hat{\theta}) - \cos(2\theta)\Bigr)^2\Bigr]$.)

From this context, we derive the explicit relationship between \(\mathbb{E}[\cos(2s\hat{\theta})]\) and amplitudes of the initial state $\alpha_k$.  
The probability of the \(l\)-th measurement outcome is given by Eq.~\eqref{eq:QAE_lth_probability}.Additionally, by adopting the estimator \(\hat{\theta}_l \coloneq \frac{\pi l}{2^q}\) for the \(l\)th outcome, we obtain
\begin{align}
    \mathbb{E}\left[\cos(2s\hat{\theta})\right] &= \Re\Biggl(\sum_{l=0}^{2^q-1} \Bigl(P_l^{(+)} + P_l^{(-)}\Bigr) e^{i2 s \hat \theta_l }\Biggr) \notag\\
    &= \Re\Biggl(\sum_{l=0}^{2^q-1} \Bigl(P_l^{(+)} + P_l^{(-)}\Bigr) e^{i2 \pi s l/2^q }\Biggr) . 
\end{align}
For the term \(P_l^{(\mp)}e^{i2\pi s \frac{l}{2^q}}\), inserting Eq.~\eqref{eq:QAE_lth_probability} yields
\begin{align}
    \sum_l P_l^{(\mp)} e^{i2\pi s \frac{l}{2^q}} &=
     \frac{1}{2^{(q+1)}}\, \sum_{k,k'} \alpha^*_{k'}\alpha_k\, e^{i2\pi \left( \pm \frac{\theta}{\pi}  - \frac{l}{2^q}  \right)k } e^{-i2\pi \left( \pm \frac{\theta}{\pi}  - \frac{l}{2^q}  \right)k' }  \sum_l e^{i2\pi\frac{l}{2^q}s} \\
     &= 
    \frac{1}{2^{(q+1)}}\, \sum_{k,k'} \alpha^*_{k'}\alpha_k\, e^{i2\pi\frac{\pm\theta}{\pi}(k'-k)} \sum_l e^{-i2\pi\frac{l}{2^q}(k'-k-s)} \\
    &= 
    \label{eq:formula2}
    \frac{1}{2^{(q+1)}}\,e^{i2\pi s \frac{\pm\theta}{\pi}} \sum_{k,k'} \alpha^*_{k'}\alpha_k\, e^{i2\pi\frac{\pm\theta}{\pi}(k'-k-s)} \sum_l e^{-i2\pi\frac{l}{2^q}(k'-k-s)}.
\end{align}
For the last term of Eq.~\eqref{eq:formula2}, we can use the following identity
\begin{align}
    \sum_l e^{-i2\pi\frac{l}{2^q}(\mu-\nu)} = 2^q \sum_{j\in\mathbb{Z}} \delta_{\mu+j2^q,\nu},
\end{align}
where \(\delta_{a,b}\) is the Kronecker delta.
Using this identity, we can express \(\mathbb{E}[\cos(2s\hat{\theta})]\) for \(s=1,2\) as
\begin{align}
    \label{eq:eval_ex_cos1}
    \mathbb{E}\left[\cos(2\hat{\theta})\right] &= \cos(2\theta) \Biggl[\sum_{k=0}^{2^q-2}\frac{\alpha^*_{k+1}\alpha_k + \text{(c.c.)}}{2}\Biggr] + \frac{\alpha^*_0\alpha_{2^q-1} + \text{(c.c.)}}{2}\cos\Bigl(2(2^q-1)\theta\Bigr), \\
    \label{eq:eval_ex_cos2}
    \mathbb{E}\left[\cos(4\hat{\theta})\right] &= \cos(4\theta) \Biggl[\sum_{k=0}^{2^q-3}\frac{\alpha^*_{k+2}\alpha_k + \text{(c.c.)}}{2}\Biggr] \notag\\[1mm]
    &\quad + \Biggl(\frac{\alpha^*_0\alpha_{2^q-2} + \text{(c.c.)}}{2} + \frac{\alpha^*_1\alpha_{2^q-1} + \text{(c.c.)}}{2}\Biggr)\cos\Bigl(2(2^q-2)\theta\Bigr),
\end{align}
where $\text{(c.c.)}$ denotes the complex conjugate of the preceding term.
By substituting Eqs.~\eqref{eq:eval_ex_cos1} and \eqref{eq:eval_ex_cos2} into Eq.~\eqref{eq:explict_formula_eval_cos}, the MSE can be expressed as a quadratic form:
\begin{align}\label{eq:def-msehat-QAE}
    \widehat{\mtr{MSE}}[\hat{a}] = \max_{\theta \in [0.\pi/2]  }\bm{\alpha}^\dagger W(\theta)\,\bm{\alpha} + C(\theta),
\end{align}
where \(\bm{\alpha} = [\alpha_0,\alpha_1,\ldots,\alpha_{2^q-1}]^T\) and 
$C(\theta)=\frac{1}{4}+\frac{1}{8}\cos(4\theta)$.
The matrix \(W(\theta)\) is a \(2^q\times2^q\) matrix given by
\begin{align}
    W(\theta) \coloneq \begin{bmatrix}
0 & a & b & 0 & \cdots & b_1 & a_0 \\
a & 0 & a & b & \cdots & 0 & b_1 \\
b & a & 0 & a & \cdots & 0 & 0 \\
0 & b & a & 0 & \cdots & 0 & 0 \\
\vdots & \vdots & \vdots & \vdots & \ddots & \vdots & \vdots \\
b_1 & 0 & 0 & 0 & \cdots & 0 & a \\
a_0 & b_1 & 0 & 0 & \cdots & a & 0
\end{bmatrix},
\end{align}
where the coefficients are defined as
\begin{gather}
a = -\frac{1}{4}\cos^2(2\theta),\quad b = \frac{1}{16}\cos(4\theta), \label{def:a_and_b}\\[1mm]
a_0 = -\frac{1}{4}\cos(2\theta)\cos\Bigl(2(2^q-1)\theta\Bigr),\quad b_1 = \frac{1}{16}\cos\Bigl(2(2^q-2)\theta\Bigr).
\end{gather}

Now the problem of computing the MSE is reduced to calculating the minimum of \(\bm{\alpha}^\dagger W(\theta)\bm{\alpha}\), or the smallest eigenvalue of \(W(\theta)\). The associated eigenvector provides the optimal initial state coefficients. 
Although an analytical solution for this eigenvector is challenging to obtain, we find that a state with $\alpha_k \propto \sin\Bigl(\frac{k\pi}{2^q}\Bigr)$ is nearly optimal. Furthermore, the unitary defined as
\begin{align}
    \quad \chi_q: \ket{0}^{\otimes q} \mapsto \sqrt{\frac{2}{2^q}} \sum_{k=0}^{2^q-1}\sin\Bigl(\frac{k\pi}{2^q}\Bigr)\ket{k},
\end{align}
can be implemented efficiently with a gate complexity of \(\tilde{O}(q)\) \cite{babbush2018encoding}. Using this state, referred to as the {\it sine state} in literature, the MSE can be computed as follows. 
\begin{align}
   \widehat{\mtr{MSE}}[\hat{a}]
    &=\max_{\theta \in [0,\pi/2]} -\frac{1}{2}\cos^2(2\theta)\cos\Bigl(\frac{\pi}{2^q}\Bigr)
    +\frac{1}{8}\cos(4\theta)\cos\Bigl(\frac{2\pi}{2^q}\Bigr) + C(\theta) \notag\\[1mm]
    &\quad + \frac{1}{2^q}\cdot\frac{1}{8}\cos(2\theta)\Bigl[1-\cos\Bigl(\frac{2\pi}{2^q}\Bigr)\Bigr]
    - \frac{1}{2^q}\cdot\frac{1}{8}\cos\Bigl(2(2^q-2)\theta\Bigr)\Bigl[1-\cos\Bigl(\frac{2\pi}{2^q}\Bigr)\Bigr].
\end{align}
Assuming that \(2^q\gg1\), we can approximate the cosine term as 
\[
\cos\Bigl(\frac{1}{2^q}\Bigr) = 1 - \frac{1}{2}\Bigl(\frac{1}{2^q}\Bigr)^2 + \mathcal{O}\!\Bigl(\Bigl(\frac{1}{2^q}\Bigr)^4\Bigr).
\]
By substituting this into the expression of the MSE, we obtain
\begin{align}
    \widehat{\mtr{MSE}}[\hat{a}]
    &= \max_{\theta \in [0,\pi/2]}-\frac{1}{2}\cos^2(2\theta)\Bigl(1 - \frac{1}{2}\Bigl(\frac{\pi}{2^q}\Bigr)^2\Bigr)
    + \frac{1}{8}\cos(4\theta)\Bigl(1 - \frac{1}{2}\Bigl(\frac{2\pi}{2^q}\Bigr)^2\Bigr) 
    + C(\theta) + \mathcal{O}\!\Bigl(\Bigl(\frac{1}{2^q}\Bigr)^3\Bigr)  \\[1mm]
    &= \max_{\theta \in [0,\pi/2]} -\frac{1}{2}\cos^2(2\theta) + \frac{1}{4}\cos(4\theta) + \frac{1}{4} + \Bigl(\cos^2(2\theta)-\cos(4\theta)\Bigr) \Bigl(\frac{\pi}{2\cdot2^q}\Bigr)^2 + \mathcal{O}\!\Bigl(\Bigl(\frac{1}{2^q}\Bigr)^3\Bigr) \\[1mm]
    &= \max_{\theta \in [0,\pi/2]} \sin^2(2\theta) \Bigl(\frac{\pi}{2\cdot2^q}\Bigr)^2 + \mathcal{O}\!\Bigl(\Bigl(\frac{1}{2^q}\Bigr)^3\Bigr) = \Bigl(\frac{\pi}{2\cdot2^q}\Bigr)^2 + \mathcal{O}\!\Bigl(\Bigl(\frac{1}{2^q}\Bigr)^3\Bigr),
\end{align}
which explicitly indicates that we have achieved the HL scaling.
\end{proof}
We importantly remark that the optimality discussed here holds specifically for the quantum algorithm shown in Fig.~\ref{fig:heisenberg-limited_improved_circuit} and does not necessarily apply to all QAE algorithms using $2^q$ queries to $\mathcal{A}$ and its inverse. While the circuit in Fig.~\ref{fig:heisenberg-limited_improved_circuit} requires a total of $2^{q} + 1$ queries, it remains possible that another quantum algorithm could achieve the same asymptotic performance with fewer queries.

To demonstrate the near-optimality of the sine state, we have numerically evaluated the MSE, comparing its performance with that of the optimal state derived from the analysis of \(W(\theta)\). As shown in Fig.~\ref{fig:comparsion_sine_num}, the sine state indeed achieves nearly optimal performance in QAE in terms of the MSE.

\subsubsection{Uniform superposition probe state does not achieve HL scaling}
In contrast to the sine state, the commonly-used uniform superposition state fails to attain the HL scaling. To illustrate this, we substitute  $\forall k \in \{0, 1,\ldots,2^{q-1}\}, \alpha_k = 1/\sqrt{2^q} $ into Eq.~\eqref{eq:def-msehat-QAE}:
\begin{align}
    \bm{\alpha}^\dagger W(\theta)\,\bm{\alpha} + C(\theta) 
    &= \frac{\qty[ 2(2^{q}-1)a  + 2(2^q-2)b +2a_0 + 4b_1 ]}{2^q} + C(\theta)  \\
    &= 2a + 2b + C(\theta) + \frac{2}{2^q}(a_0 - a) + \frac{4}{2^q} (b_1-b) .
\end{align}
Using the definitions of $C(\theta)$, $a$, and $b$, it can be shown that  \(2a + 2b + C(\theta)\equiv0\) using  the double-angle formula. We thus focus on the remaining part:
\begin{align} \label{eq:superposition-mse}
    \bm{\alpha}^\dagger W(\theta)\,\bm{\alpha} + C(\theta)
    &= \frac{2}{2^q}(a_0 - a) + \frac{4}{2^q} (b_1-b) \\
    \label{eq:sine_cosine_derive}
    &= \frac{1}{2^{q+1} } \cos (2\theta) \qty[\cos(2\theta) - \cos(2(2^{q}-1)\theta ) ] + \frac1{2^{q+2}} \qty[\cos(2(2^q-2)\theta) - \cos(4\theta) ] \\
    &= \frac{1}{2^{q+1} }  \cos (2\theta) \qty[2 \sin (2^q \theta) \sin ((2^q -2 ) \theta) ] - \frac1{2^{q+2}} [2 \sin (2^q \theta) \sin ((2^q-4)\theta) ] \\
    &= \frac{\sin(2^q \theta)}{2^{q+1}} \qty[2 \cos(2\theta) \sin ((2^q-2)\theta) - \sin((2^q-4 )\theta)  ] \\
    &= \frac{\sin^2 (2^q \theta)}{2^{q+1}} .
\end{align}
Here, we have used the product-to-sum formulas in the third equality, and the last equality arises from the sum-to-product formula, $2 \cos (2 \theta) \sin ((2^q-2) \theta) = \sin(2^q \theta) + \sin((2^q-4) \theta)$. Consequently, the $\widehat{\mtr{MSE}}[\hat{a}]$ is evaluated as follows,
\begin{align}
   \widehat{\mtr{MSE}}[\hat{a}] = \max_{\theta \in [0, \pi/2]} \frac{\sin^2(2^q \theta)}{2^{q+1}} = \frac1{2^{q+1}}.
\end{align}
The MSE scales quadratically worse than the case for the sine state, indicating that we can only achieve the standard quantum limit using the uniform superposition state.
We remark that when the target amplitude $\theta = l\cdot \pi/ 2^q, l \in \{0,1,\ldots,2^q \}$, then the eigenvalue of the uniform superposition state corresponds to the smallest eigenvalue of $W(\theta)$.

\begin{figure}[t]
    \centering
    \includegraphics[width=1.\linewidth]{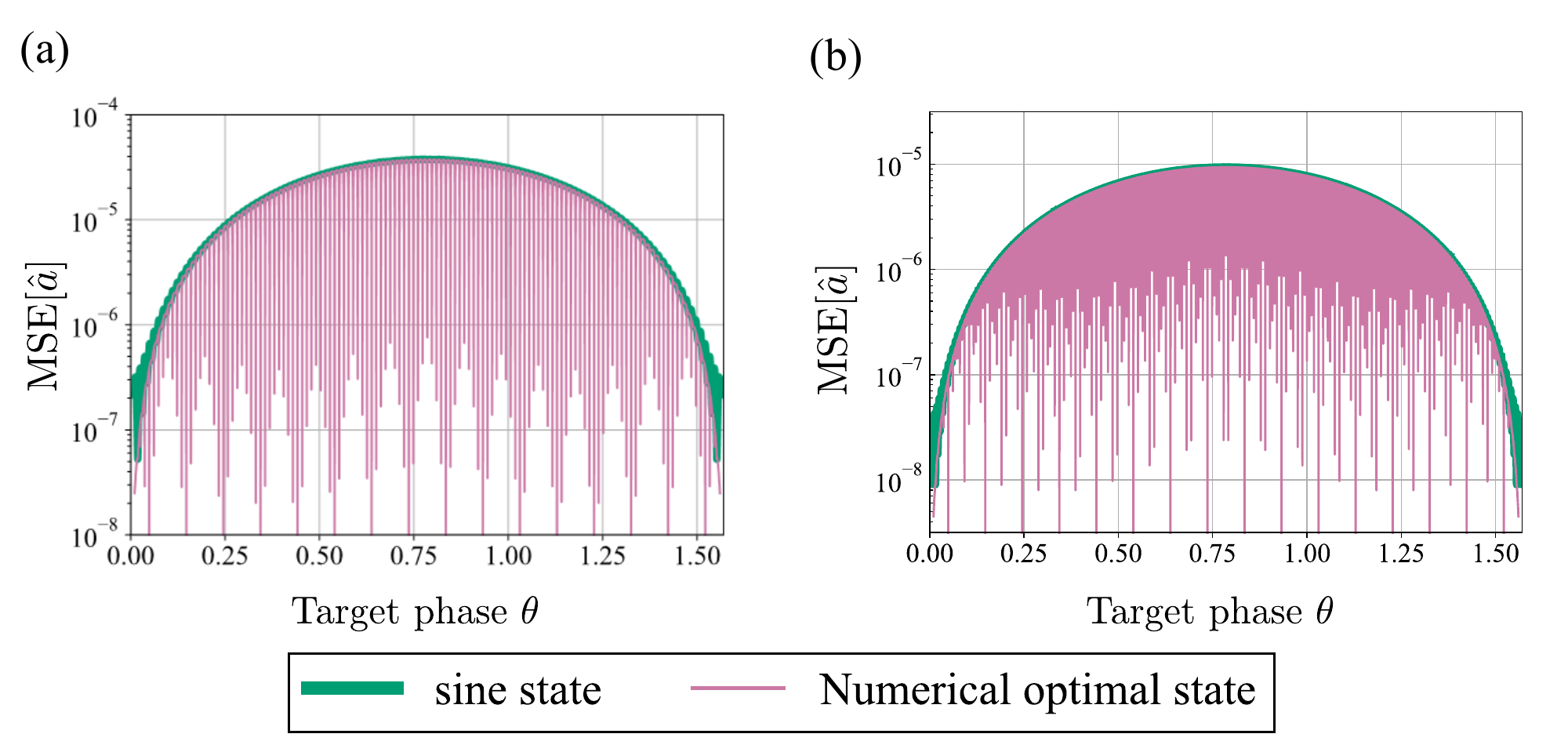}
    \caption{
    \justifying{
    The MSE in the amplitude estimation with (a) $q=8$ and (b) $q=9$ probe system qubits. The interval $\theta \in [0.01,\pi/2-0.01]$ was divided into 10,000 points, and the numerical optimal state was determined by computing the minimum eigenvalue of $W(\theta)$. The green line represents the envelope of the purple line; they nearly coincide except at a specific value of \(\theta\) where the MSE of the numerical optimal state is exceptionally small.
    }}
    \label{fig:comparsion_sine_num}
\end{figure}


\subsection{Expectation value estimation based on QAE algorithm}
\label{sec:Expectation_base_QAE}

Building on our previous discussion, we now focus on estimating the expectation value of an observable via the QAE algorithm. From the previous section, the QAE algorithm using a \(q\)-qubit probe system can estimate a parameter \(a\) with the MSE of 
$
\left(\frac{\pi}{2^{q+1}}\right)^2 + \mac{O}\qty( \qty(\frac{1}{2^q})^3 ),
$
when provided with an oracle that encodes the amplitude \(\sqrt{a}\). To bridge the gap between amplitude estimation and expectation value estimation, it is crucial to embed the expectation value onto an amplitude that can be processed by the QAE algorithm.

To achieve this goal, we present a concrete quantum circuit implementation that, given a block-encoding of the target observable, yields an estimator for its expectation value.
To tackle expectation value estimation problem, 
several studies have explored approaches that rely on phase-estimation-like circuits~\cite{knill2007optimal, rall2020estimating, PRXQuantum.2.010346, Obrien2022efficient}.
To deal with a broad class of observables, we consider the scenario in which an \(a\)-block-encoding \(B\) of the observable \(O\) is available.
By employing block-encoding techniques and a state-preparation oracle, we can embed the target value \(\bra{\psi} O \ket{\psi}\) into an amplitude. 
For instance, the quantum circuit shown in Fig.~\ref{fig:QAE_subroutine} can embed the expectation values into the amplitude as ${a} = \frac{1+\langle O \rangle }2 \geq 0$.
By using this circuit, we now present our main result, which rigorously gives the query complexity for the state preparation oracle and enables  expectation value estimation with the HL scaling:

\begin{lemma}[Expectation value estimation based on QAE algorithm]
    \label{thm:query_comp_HL_amp_est}
    Let $q$ be a positive integer, and \(O\) be an observable on \(N\) qubits with spectral norm \(\norm{O} \leq 1\). Given a state-preparation unitary \(U_\psi:\ket{0}^{\otimes N} \mapsto \ket{\psi}\) and a controlled \(a\)-block-encoding of \(O\), there exists a quantum algorithm that outputs an estimator \(\hat o\) for \(\bra{\psi} O \ket{\psi}\) satisfying 
    \[
    \mtr{MSE}[\hat o] = \qty(\frac{\pi}{2^q} )^2 + \mac{O} \qty(2^{-3q}) ,
    \]
   with $2^q +1$ uses of \(U_\psi\) and its inverse.
\end{lemma}

\begin{figure}[t]
    \centering
    \includegraphics[width=0.5\linewidth]{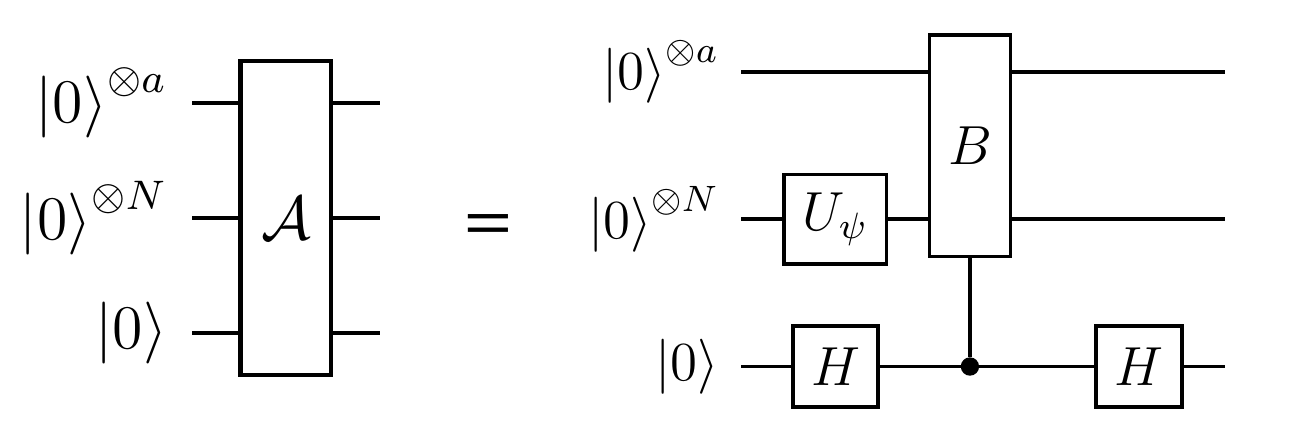}
    \caption{Implementation of a $(N+a+1)$-qubit oracle $\mac{A}$ that embeds the amplitude $\sqrt{\frac{1+ \bra{\psi}O\ket{\psi} }2 }$. Here, $U_\psi $ denotes a $N$-qubit state preparation unitary and $B$ is a $a$-block-encoding of a target observable $O$. }
    \label{fig:QAE_subroutine}
\end{figure}

\begin{proof}
We first consider the embedding of $\bra{\psi} O \ket{\psi} \in [-1,1] $ into the amplitude $\sin \theta$. Due to the difference of the interval between the expectation value and the amplitude, we need to modify the encoding $B$ of the target observable. To address this issue, we consider an oracle $\mac{A}$ that embeds a target amplitude $a = (1 + \langle O \rangle )/2 $ as follows (see also Fig.~\ref{fig:QAE_subroutine}),
\begin{align}
    \mac{A} = \Bigl(H\ot \1_{(a+N)}\Bigr) \left( \ketbra{1} \ot B  \right)  (H \otimes U_\psi \otimes \1_a).
\end{align}
We demonstrate that this oracle embeds the target amplitude as follows,
\begin{align}
    \mac{A} \ket{0} \ket{0}^{\otimes (a+N)}
    &= (H \otimes \1_{(a + N)}) \qty(\frac1{\sqrt 2} \ket{0} \ket{\psi} \ket{0}^{\otimes a}    
    +
    \frac1{\sqrt 2} \ket{1} B(\ket{\psi} \ket{0}^{\otimes a} ) ) \\
    &= \frac12 \qty(\ket{0} + \ket{1} )\ket{\psi} \ket{0}^{\otimes a}  + \frac12 \qty(\ket{0} - \ket{1} ) B(\ket{\psi} \ket{0}^{\otimes a} ) \\
    &= \ket{0} \qty(\frac{1+B}2)\ket{\psi} \ket{0}^{\otimes a} + \ket{1} \qty(\frac{1-B}2)\ket{\psi} \ket{0}^{\otimes a} \\
    &= \sqrt{\frac{1 + \langle O \rangle }2 } \ket{0}\ket{\Psi_0}
    + \sqrt{\frac{1 - \langle O \rangle }2 } \ket{1}\ket{\Psi_1},
\end{align}
where $\ket{\Psi_0} $ and $\ket{\Psi_1}$ are normalized quantum states.
In the last inequality, each amplitude is derived from
\begin{align}
    \sqrt{\bra{\psi}\bra{0}^{\otimes a} \qty(\frac{1+B^\dagger}2)
    \qty(\frac{1+B}2) \ket{\psi}\ket{0}^{\otimes a}} &= 
    \sqrt{
    \frac{1}2 + \frac{\bra{\psi} O + O^\dagger  \ket{\psi}}4 } \\
    &= \sqrt{ \frac{1 + \mtr{Re}[\langle O \rangle ] }2},  \\ 
    \sqrt{\bra{\psi}\bra{0}^{\otimes a} \qty(\frac{1-B^\dagger}2)
    \qty(\frac{1-B}2) \ket{\psi}\ket{0}^{\otimes a}} &= \sqrt{ \frac{1 - \mtr{Re}[\langle O \rangle ] }2}.
\end{align}
Here, we consider the block-encoding of an observable, so $\mtr{Re}[\langle O \rangle ] = \langle O \rangle $ holds.

By applying this oracle to the quantum circuit that uses $q$-qubit sine state expressed in Lemma \ref{lem:HL_scaling_QAE}, we can obtain an estimator $\hat a$ for $(1 + \langle O \rangle)/2$ with $2^q + 1 $ queries to $\mac{A}$ and its inverse. 
We convert it into $\hat{o} = 2 \hat a - 1$, and the MSE of this estimator is calculates as follows,
\begin{align}
    \widehat{\mtr{MSE}}[\hat{o}] &= \widehat{\mtr{MSE}}[{2 \hat{a} - 1}] = 4\widehat{\mtr{MSE}}[\hat a] \\
    &= \qty(\frac{\pi}{2^q} )^2 + \mac{O} (2^{-3q}).
\end{align}

\end{proof}

In the asymptotic regime where $2^q \gg 1$, we can neglect the term $\mac{O}\qty(2^{-3q})$. Hence, when the target MSE $\varepsilon \in (0,1)$ is sufficiently small, it suffices to choose $q = \lceil \log_2(\pi/\varepsilon) \rceil$ to ensure that $\varepsilon^2 \geq \widehat{\mtr{MSE}}[\hat o]$. 
With this choice, the total query complexity of \( U_\psi \) is bounded by
\begin{align}
    2^q + 1 
    &\leq 2^{\log_2(\pi/\varepsilon) + 1} + 1 \\
    &= \frac{2\pi}{\varepsilon} + 1 = \mac{O}(1)/\varepsilon.
\end{align}
Note that the number of queries is halved compared to the best known result~\cite{Obrien2022efficient} which approximated the MSE with the Holvevo variance under the assumption that the number of ancilla is sufficiently large.

\section{Framework of QGE algorithm for estimating observables} \label{sec:qge-general}
In this section, we review the quantum gradient estimation (QGE) algorithm and present a general framework for its adaptive variant applied to observable estimation. Section~\ref{sec:core_idea_Jordan} introduces the core idea behind the QGE algorithm, namely Jordan's algorithm~\cite{jordan2005fast}. Section~\ref{sec:general_framwork} describes the overall structure of the adaptive QGE algorithm for estimating multiple observables with HL scaling and highlights the key factors enabling its enhancement. Section~\ref{sec:WYY2024} demonstrates how the previously proposed adaptive QGE algorithm~\cite{wada2024Heisenberg} fits within our framework and provides a concise overview of its structure.

\subsection{Core of the QGE algorithm}
\label{sec:core_idea_Jordan}
For a given black-box oracle for a real scalar function \(f(\bm{x})\) with \(\bm{x} \in \mathbb{R}^M\), the quantum algorithm introduced by Stephan P. Jordan \cite{jordan2005fast} allows one to estimate the \(M\)-dimensional gradient \(\nabla f(\bm{0})\). 
The target point \(\bm{x} = 0\) can be generalized to any \(\bm{x}\) by appropriately redefining \(f(\bm{x})\). The quantum algorithm proceeds in three main steps: (i) preparing a superposition state over a grid of points \(\bm{x}\) around the target \(\bm{x} = 0\) on the probe system;
(ii) applying the black-box oracle to encode the phase \(e^{i f(\bm{x})}\) at each grid point $\bm x$; and (iii) performing an inverse quantum Fourier transform followed by a measurement in the computational basis. In the following, we describe the algorithm step by step.

To describe step (i), it is useful to define a grid \(G_p^M\) centered at \(\bm{x} = 0\) for evaluating \(f\):
\begin{align}
    G_p^M := \left\{ \frac{\mu}{2^p} - \frac{1}{2} + \frac{1}{2^{p+1}} : \mu \in \{0, 1, \ldots, 2^p - 1\} \right\}^M, 
\end{align}
where \(p\) is a positive integer determining the estimation precision. A one-to-one mapping \(\phi : \mu \mapsto \phi(\mu) := \frac{\mu}{2^p}  - \frac{1}{2} + \frac{1}{2^{p+1}}\) is defined on \(G_p\), which allows us to label the \(p\)-qubit computational basis states \(|\mu\rangle\) by the corresponding grid point \(\bm{x} \equiv \phi(\mu) \in G_p\).
By applying Hadamard gates to  \(pM\) qubits initialized as \(|0\rangle^{\otimes pM}\), we create a uniform superposition over the grid:
\begin{align}
    \label{eq:super_position}
    \frac{1}{\sqrt{2^{pM}}} \sum_{\bm{x} \in G_p^M} |\bm{x}\rangle 
    = \frac{1}{\sqrt{2^{pM}}} \sum_{(x_1, \ldots, x_M) \in G_p^M} |x_1\rangle |x_2\rangle \cdots |x_M\rangle, 
\end{align}
where each \(|x_j\rangle\) represents a \(p\)-qubit block corresponding to the \(j\)-th coordinate.

To illustrate step (ii), we assume that we have access to a phase oracle \(O_f\) for \(f(\bm{x})\), defined as
\begin{align}
    O_f : |\bm{x}\rangle \mapsto  e^{i f(\bm{x})} |\bm{x}\rangle, \quad \forall \bm{x} \in G_p^M.
\end{align}
For this context, Ref.~\cite{gilyen2019gradient} analyzes Jordan’s algorithm using phase (and probability) oracles, while the original work \cite{jordan2005fast} assumes the existence of an \(\eta\)-accurate binary oracle that outputs \(f(\bm{x})\) in binary with precision \(\eta\).
For simplicity, we consider the case $f$ is affine linear: $f(\bm x) = \bm g \cdot \bm x$, where $\bm g \in \mathbb{R}^M$ holds $\norm{\bm{g}}_{\infty} \leq 1/3 $ where $\norm{\cdot }_{\infty} \coloneq \max \{ \abs{g_1}, \abs{g_2}, \ldots, \abs{g_M} \}$.
Applying the modified phase oracle \(\left(O_f\right)^{2\pi 2^p}\) to the state in Eq.~\eqref{eq:super_position} yields
\begin{align}
    \label{eq:gradient_state}
    \frac{1}{\sqrt{2^{pM}}} \sum_{\bm{x} \in G_p^M} e^{2\pi i 2^p \bm g \cdot \bm x} |\bm{x}\rangle. 
\end{align}

In step (iii), we subsequently apply a slightly modified inverse QFT, denoted by \(\mtr{QFT}_{G_p}^\dagger\), to individual  \(p\)-qubit block. For every \( x \in G_p\), the operation of this variant of QFT is given as
\begin{align}
    \mtr{QFT}_{G_p} : |x \rangle \mapsto \frac{1}{\sqrt{2^p}} \sum_{k \in {G_p}} e^{2\pi i 2^p k\cdot x} | k \rangle. 
\end{align}
This transform is equivalent to the standard \(p\)-qubit QFT, up to conjugation by a tensor product of \(p\) single-qubit gates \cite{jordan2005fast}.

Finally, upon measuring the state in Eq.~\eqref{eq:gradient_state} after applying \((\mtr{QFT}_{G_p}^\dagger )^{\otimes M}\), the resulting measurement statistics mimic those of the conventional quantum phase estimation (QPE) algorithm \cite{nielsen2010quantum}. Specifically, if the measurement outcome is \((k_1, k_2, \ldots, k_M) \in G_p^M\), then for each \(j \in \{1, 2, \ldots, M\}\) we have
\begin{align}
    \mathrm{Pr}\left[ |k_j - g_j| > \frac{3}{2^p} \right] \leq \frac{1}{4}, 
\end{align}
which implies that when the target function is exactly affine linear, the gradient \(\bm g\) can be estimated with high confidence. Moreover, if the target function is nearly affine linear, meaning that the Euclidean norm expressed as
\begin{align}
    \norm{\frac{1}{\sqrt{2^{pM}}} \sum_{\bm{x} \in G_p^M} e^{2\pi i 2^p \bm g \cdot \bm x} |\bm{x}\rangle  - \frac{1}{\sqrt{2^{pM}}} \sum_{\bm{x} \in G_p^M} e^{2\pi i 2^p f(\bm x) } |\bm{x}\rangle}
\end{align}
is sufficiently small, then the estimate \(\bm g\) accurately approximates the gradient \(\nabla f(\bm{0})\) \cite{gilyen2019gradient, apeldoorn2023quantum}.
Based on this analysis, the QGE algorithm can be interpreted as multi-parameter extension of the QPE algorithm.
In particular, by focusing on the state expressed in Eq.~\eqref{eq:gradient_state}, we observe that this state is equivalent to the $M$ tensor product of states in the form of
    \begin{equation} \label{eq:qge-preqft-state}
        \frac{1}{{2^p} } \sum_{k,x\in G_p}e^{2\pi i2^p x(g-k)}\ket{k},
    \end{equation}
    after $({\rm QFT}_{G_p}^\dagger)^{\otimes M}$.
    We find that Eq.~\eqref{eq:qge-preqft-state} is the same as the following final state of the QPE algorithm:
    \begin{equation}
        \frac{1}{{2^p} }\sum_{\mu, \nu=0}^{2^p-1}e^{2\pi i \nu(\theta- \mu/2^p)} |\mu\rangle
    \end{equation}
    up to a global phase, if we replace the parameters $x,k,g$ appropriately. 
    Here, $\theta$ is connected with $g$ as $\theta=g+1/2-1/2^{p+1}$.
    Therefore, if we establish a scheme where the estimator of phase $\hat{\theta}$ satisfies
    \begin{equation}
        {\rm Pr}\left[|\hat{\theta}-\theta|>\frac{1}{2\pi}\right]\leq  \delta',
    \end{equation}
    where $\delta' \in (0,1)$,
    we can immediately convert this inequality into that of gradient $\hat{g}$ as
    \begin{equation}
        {\rm Pr}\left[|\hat{g}-g|>\frac{1}{2\pi}\right]\leq  \delta'.
    \end{equation}
    An error $\delta'$ in QPE corresponding that of QGE is exploited in discussions later (See Sec.~\ref{subsec:modify_probe_state}).
    

\subsection{General framework of adaptive QGE algorithm}
\label{sec:general_framwork}


If we have access to an oracle that embeds the function $f(\bm{x}) = \sum_{j=1}^M x_j \langle O_j \rangle$ as a phase, then the QGE algorithm can be applied to $M$ multiple observables estimation. More specifically, 
the brief procedure of multiple observables estimation with the QGE algorithm is expressed as 
\begin{align}
        \sum_{\bm x} c_{\bm x} \ket{\bm x} \xrightarrow{\text{encode info.}} \sum_{\bm x} c_{\bm x} e^{i \sum_j  \langle O_j \rangle x_j} \ket{\bm x} \xrightarrow{\text{inverse QFT}} \ket{\langle O_1 \rangle} \cdots \ket{\langle O_M \rangle} \xrightarrow{\text{get samples}} \{ \langle O_j \rangle  \}_{j=1}^M,
\end{align}
\begin{algorithm}[H]
\caption{General framework of adaptive QGE algorithm for multiple observables}
  \label{alg:unified_framework}
  \begin{algorithmic}[1]
  \Statex \textbf{Input:} $\log_2 d$-qubit state preparation unitary $U_\psi$ and its inverse, $M$ observables $\{O_j\}_{j=1}^M$ of bounded spectral norm $\norm{O_j} \leq 1$ where  $M \geq 2\log_2 d + 24 $, 
  confidence parameter $c \in (0, \frac{3}{8(1 + \pi)^2}]$, target precision parameter $\varepsilon \in (0, 1)$, an integer $p\geq 1$, and a set of integers $\{R^{(q)}\}_{q=0}^{\lceil \log_2(1/\varepsilon) \rceil }$
  \smallskip

  \Statex \textbf{Subroutine:} a probe-state preparation subroutine $\mathcal{U}_{\Upsilon}^{(q)}(\{A_j\}_{j=1}^M)$.
  This process is ensured to work as follows:
  for given integer $q\geq 0$ and observables $\{A_j\}_{j=1}^M$ of bounded spectral norm $\norm{A_j} \leq 1$
  , the process $\mathcal{U}_{\Upsilon}^{(q)}$ prepares $R^{(q)}$ copies of a $pM$-qubit quantum state \(\ket{\Upsilon(q)}\) that is close in the Euclidean norm to the following ideal probe state under fixed set of amplitudes $\{c_{\bm x}\}$:
  \begin{equation}
      \ket{\Upsilon(q) } \simeq \sum_{\bm x \in G_p^M} c_{\bm x} 
      e^{2\pi i {2^p}{\sum_{j=1}^M x_j 2^{q} \pi^{-1} \bra{\psi}A_j\ket{\psi}}} \ket{\bm x},~~\mbox{if}~~~ |\bra{\psi}A_j\ket{\psi}|\leq 2^{-q},~\forall j.
  \end{equation}

  \Statex \textbf{Output:} A sample from an estimator $\hat{u} = (\hat{u}_1, \dots, \hat{u}_M)$ whose $j$-th element estimates $\langle O_j \rangle := \langle \psi | O_j | \psi \rangle$ within MSE $\epsilon^2$ as 
  $$
  \max_{j=1,2,\dots,M} \mathbb{E}[(\hat{u}_j - \langle O_j \rangle)^2] \leq \epsilon^2 
  $$
  \State $\tilde{u}_j^{(0)} \leftarrow 0$ for $j = 1, 2, \dots, M$
  \For{$q = 0, 1, \dots, q_{\max} \coloneq \lceil \log_2(1/\epsilon) \rceil$}
      \State $A_j \leftarrow {O}_j-\tilde{u}^{(q)}_j\bm{1}$
      \State Call the subroutine $\mathcal{U}_{\Upsilon}(q,\{A_j\})$ for preparing $R^{(q)}$ copies of the quantum state $\ket{\Upsilon(q)}$
      \State Apply  $(\mtr{QFT}_{G_p}^\dagger)^{\otimes M}$ on each copy 
      \State Perform computational basis measurement to obtain output
      $(k_1, \ldots, k_M) \in G_p^M$.  
      \State $g_j^{(q)} \leftarrow$ coordinate-wise medians of the measurement results \label{step:get-coordinate-wise-median}
      \State $\tilde{u}_j^{(q+1)} \leftarrow \tilde{u}_j^{(q)} + {\pi}{2^{-q}} g_j^{(q)}$.
      \For{$j=1, ..., M$}
      \If{$\tilde{u}_j^{(q+1)}\geq 1$ (or $\leq -1$)} 
          \State $\tilde{u}_j^{(q+1)} \leftarrow 1$ (or $-1$) 
      \EndIf
      \EndFor
  \EndFor
  \end{algorithmic}
\end{algorithm}
\noindent
where $\bm x \in \mathbb{R}^M$ is defined as a specific grid point around $\bm x = \bm 0$ and $c_{\bm x}$ denotes the amplitude of the initial state. Based on this idea,
the proposals in Refs.~\cite{huggins2022nearly, apeldoorn2023quantum} achieve a query complexity of $\tilde{\mathcal{O}}(\sqrt{M}/\varepsilon_{\text{add}})$ for $U_\psi$ and its inverse, with additional space complexity $\mathcal{O}(M \log(1/\varepsilon_{\text{add}}))$. However, achieving the genuine HL scaling---$\tilde{\mathcal{O}}(\sqrt{M})/\varepsilon$---in total query complexity, along with reduced space complexity, is essential for efficient observable estimation. To this end, adopting an adaptive strategy has proven effective~\cite{wada2024Heisenberg}.
In what follows, we first present a general framework for the adaptive QGE algorithm for observable estimation, and subsequently introduce a novel circuit implementation that enables even faster estimation.

The core idea of the adaptive QGE algorithm \cite{wada2024Heisenberg} is to estimate each binary digit of multiple target expectation values \(\langle O_j \rangle\) iteratively. We obtain $1$-th temporary estimates \(\{\tilde{u}_j^{(1)}\}\) for \(\{\langle O_j \rangle\}\) by measuring the state
\begin{align}
    \ket{\Upsilon(0)} \simeq \sum_{\bm x \in G_p^M} c_{\bm x } e^{2 \pi i 2^p\sum_j x_j \pi^{-1} \langle O_j \rangle} \ket{\bm x}.
\end{align}
Based on these estimates, we adaptively construct quantum circuits that output the \((q+1)\)-th temporary estimates \(\{\tilde{u}_j^{(q+1)}\}\), using the previous estimates \(\{\tilde{u}_j^{(q)}\}\). To achieve this, it is crucial to prepare the following quantum state:
\begin{align}
    \ket{\Upsilon(q)} \simeq \sum_{\bm x \in G_p^M} c_{\bm x } e^{2 \pi i 2^p \sum_j x_j 2^q \pi^{-1} \langle O_j - \tilde{u}_j^{(q)} \cdot \mathds{1} \rangle} \ket{\bm x}.
\end{align}
By repeating this procedure up to \(q_{\max} \coloneq \lceil \log_2(1/\varepsilon) \rceil \), we expect to obtain final estimators that satisfy the target precision \(\varepsilon\).

To ensure successful estimation, a sufficient number of measurement samples is required to construct each temporary estimate \(\tilde{u}_j^{(q)}\), and the number of required samples \(R^{(q)}\) generally varies with \(q\). To formalize this, we define \(\mathcal{U}_{\Upsilon}^{(q)}\) as a {\it state preparation subroutine} which, given a set of bounded Hermitian operators \(\{A_j\}_{j=1}^M\) and integers \(\{R^{(q)}\}_{q=0}^{q_{\max}}\), prepares \(R^{(q)}\) copies of the state
\begin{align}
    \ket{\Upsilon(q)} \simeq \sum_{\bm x \in G_p^M} c_{\bm x } e^{2 \pi i 2^p \sum_j x_j 2^q \pi^{-1} \langle A_j \rangle} \ket{\bm x}.
\end{align}
Based on this idea, we summarize the generalized adaptive QGE algorithm for multiple observables in Algorithm \ref{alg:unified_framework}.

This framework provides a sufficient condition under which the MSE of the output estimators is guaranteed to be upper bounded by \(\varepsilon^2\) for any target precision \(\varepsilon \in (0,1)\). Following the proof of Reference~\cite{wada2024Heisenberg}, we formalize this result in the following theorem:
\begin{theorem}
\label{thm:performance_main_alg}
Let \(\varepsilon \in (0,1)\) be a target precision. 
Given \(M\) observables \(\{O_j\}_{j=1}^M\) of bounded spectral norm $\|O_j\|\leq 1$, a state-preparation oracle \(U_\psi\) and its inverse, and a set of integers $\{R^{(q)}\}_{q=0}^{q_{\max} }$ where $q_{\max} \coloneq \lceil \log_2(1/\varepsilon) \rceil$, suppose that there exists a probe-state preparation process \(\mac{U}_\Upsilon^{(q)} \) such that the estimates \(\{\tilde{u}_j^{(q+1)}\}_{j=1}^M\), obtained from the measurement outcomes in Step~8 in Algorithm~\ref{alg:unified_framework}, satisfy
\begin{align}
    \forall j,\quad \abs{\tilde{u}_j^{(q+1)} - \langle O_j \rangle} \le \frac{1}{2^{q+1}},
\end{align}
with probability at least \(1 - \delta^{(q)}\) where $\delta^{(q)} \coloneq \frac{c}{8^{q_{\max} - q} }$, under the assumption
\begin{align}
    \label{eq:assume-grad-est}
    \forall j,\quad \abs{\tilde{u}_j^{(q)} - \langle O_j \rangle} \le \frac{1}{2^{q}}.
\end{align}
Then, Algorithm~\ref{alg:unified_framework} outputs estimators \(\{\hat{u}_j\}_{j=1}^M\) such that
\begin{align}
    \label{eq:MSE_general_framework}
    \max_{j=1,\dots,M} \operatorname{MSE}[\hat{u}_j] \le \varepsilon^2.
\end{align}
\end{theorem}

\begin{proof}
    From the assumption of \(\mac{U}_\Upsilon^{(q)} \), we can divide the behavior of its estimators into two cases. One is when all estimates $\tilde{u}_j^{(q+1)}$ satisfy $\abs{\tilde{u}_j^{(q+1)} - \langle O_j \rangle} \le 1/2^{q+1}$ for all iterations, which we define as the success of estimation, and the other is when at least one estimate violates this condition, which we define this event as the estimation failure event. 
    We can easily analyze the additive error of the success of estimation. By considering the condition that all \( \{ \tilde{u}_j^{(q_{\max}+1)} \}_{j=1}^M \) must satisfy, we obtain
    \begin{align}
        \forall j, \quad \abs{\tilde{u}_j^{(q_{\max}+1)} - \langle O_j \rangle }^2 \leq \frac{1}{2^{2(q_{\max} + 1)}}.
    \end{align}
    
    In the case where estimation fails, we simply bound the medians $g_j^{(q)}$ within the interval $[-1/2, 1/2]$, since each measurement outcome $k_j$ lies within this range. 
    Based on this consideration, we analyze the additive error of the final estimate 
    when at least one of the $(q')$-th temporal estimates violates the condition at some intermediate iteration $q'$ where $0 \leq q' \leq q_{\max}$.  
    In such a case, we have
    \begin{align}
        \abs{\tilde u_j^{(q_{\max}+1)} - \langle O_j \rangle} &\leq \abs{\tilde u_j^{(q_{\max}+1)} - \tilde u_j^{(q')}} + \abs{\tilde u_j^{(q')} - \langle O_j \rangle} \\
        &\leq \pi \abs{\sum^{q_{\max}}_{q = q'} \frac{g_j^{(q)}}{2^q}} + \frac{1}{2^{q'}} \\
        &\leq \frac{1+\pi}{2^{q'}},
    \end{align}
    where the final inequality follows from \(\abs{g_j^{(q)}} \leq 1/2\). 
    
    By combining the results of two cases, the MSE of \(\hat u_j\) can be bounded as
    \begin{align}
        &\mathbb{E}[(\hat u_j - \langle O_j\rangle)^2] \\
        &\leq \frac{1}{2^{2(q_{\max} + 1)}} + (1+\pi)^2 \sum^{q_{\max}}_{q=0} \frac{\delta^{(q)}}{4^q} \\
        &
        \label{eq:the_choice_of_c}
        \leq \frac{1}{2^{2(q_{\max} + 1)}} + (1+\pi)^2 c\, 2^{-2 q_{\max}+1} \\
        &\leq \varepsilon^2.
    \end{align}
    In the last inequality, we use \(q_{\max} \coloneq \lceil \log_2(1/\varepsilon) \rceil\) and we choose \(c \in \left(0, \frac{3}{8(1+\pi)^2}\right]\). This bound holds for every \(j\), thereby establishing Eq.~\eqref{eq:MSE_general_framework}.
\end{proof}
As stated in Problem~\ref{prob:query_estimation}, our ultimate goal is to minimize the total query to $U_\psi$ and its inverse.
Importantly, we can easily see what determines this cost 
in the framework---probe-state preparation subroutine \(\mac{U}_\Upsilon\). Specifically, the total number of queries is given by
\begin{align}
    \text{(Total number of queries)} = \sum_{q=0}^{q_{\max}} \qty(\text{Cost of } \mac{U}_\Upsilon^{(q)}) .
\end{align}
We find that, to attain HL scaling, it suffices to ensure that $\qty(\text{Cost of } \mac{U}_\Upsilon^{(q)})$ is $ \mac{O} \qty(\aleph\cdot 2^q \log(1/\delta^{(q)}) ) $ for all $q$, where $\aleph$ is a positive factor independent of $q$ and $\varepsilon$. Indeed, the total query complexity is bounded if $\delta^{(q)}=c/8^{q_{\rm max} - q}$ as
\begin{align}
    \sum_{q=0}^{q_{\max}} \qty(\text{Cost of } \mac{U}_\Upsilon^{(q)}) &\leq 
     \aleph \sum_{q=0}^{q_{\max}} 2^q \log(1/\delta{(q)
     }) \\
     &= \aleph \cdot \qty[ 2^{q_{\max +1}} \log(8/c) - (q_{\max} + 2) \log 8 - \log(1/c) ] 
     \\
     &\leq\aleph \cdot 2^{q_{\max +1}} \log(8/c) =   \mac{O} (\aleph)/\varepsilon.
\end{align}
This indicates that it is essential to reduce the prefactor \(\aleph\) without affecting its dependence on \(q\) and \(\varepsilon\).
Motivated by this requirement, we present strategies for reducing this prefactor in Sec.~\ref{sec:all-symemtry-diven-QGE}. Table \ref{tab:factor_como} shows a brief comparison between the previously proposed adaptive QGE algorithm and our methods that are detailed in Secs.~\ref{subsec:Symmetry-QGE-algorithm},\ref{subsec:parallel_QGE_alg}, and \ref{subsec:Integral-QGE}.

\begin{table}[h]
\centering
\caption{
\justifying{
Comparison of the cost of the probe-state preparation circuit \(\mathcal{U}_{\Upsilon}\) and the prefactor \(\aleph\) for estimating \(M\) \(d\)-dimensional observables using Algorithm~\ref{alg:unified_framework}. In Method~I and Method~II, we assume access to block-encodings of direct-sum structured observables \(O_j = \bigoplus_{\lambda} O_j^{(\lambda)}\), where each \(\lambda\) labels a Hilbert subspace, and \(O_j^{(\lambda)}\) denotes the corresponding observable acting on that subspace. The dimension of the subspace labeled by \(\lambda\) is denoted by \(m_\lambda\). Here the subspace decomposition is assumed to be common among all the observables. For simplicity, we are interested in a single $\lambda$ in particular.
Note that in all methods, the number of samples at $q$-th iteration can be chosen as \(R^{(q)} = \mathcal{O}\bigl(\log(M/\delta^{(q)})\bigr)\), where \(\delta^{(q)} \coloneq c / 8^{q_{\max} - q}\). The adaptive QGE in Ref.~\cite{wada2024Heisenberg} also achieves the similar query complexity of Method I if block encoding to $O_j^{(\lambda)}$ are available.
}
}
\label{tab:factor_como}
\begin{tabular}{lcc}
\toprule
\textbf{Method} & $\qty(\text{Cost of } \mac{U}_\Upsilon^{(q)})$ & $\aleph$ \\
\midrule
$[\text{WYY2024}]$ \cite{wada2024Heisenberg}    
    & $\mathcal{O}\qty(2^{q+p} R^{(q)} \sqrt{M\log d})$
    & $\mathcal{O}\qty(\sqrt{M \log d} \log M )$ \\
\midrule
Method I (Sec.~\ref{subsec:Symmetry-QGE-algorithm}) 
    & $\mathcal{O}\qty(2^{p+q} R^{(q)} \sqrt{\norm{\sum_{j} [O^{(\lambda)}_j]^2 } \log m_\lambda})$
    & $\mathcal{O}\qty(\sqrt{\norm{\sum_{j} [O^{(\lambda)}_j]^2 } \log m_\lambda} \log M )$ \\
Method by parallel scheme (Sec.~\ref{subsec:parallel_QGE_alg}) 
    & $\mathcal{O}\qty(2^{q+p} \sqrt{MR^{(q)} \log(d/\delta^{(q)})})$
    & $\mathcal{O}\qty(\sqrt{M \log d  \log M} )$ \\
Method II (Sec:~\ref{subsec:Integral-QGE})
    & $\mathcal{O}\qty(2^{p+q} \sqrt{ \norm{\sum_{j} [O^{(\lambda)}_j]^2 } R^{(q)} \log(m_\lambda /\delta^{(q)})})$
    & $\mathcal{O}\qty(\sqrt{ \norm{\sum_{j} [O^{(\lambda)}_j]^2 } \log m_\lambda  \log M} )$ \\
\bottomrule
\end{tabular}
\end{table}

\subsection{Review on Adaptive QGE algorithm for multiple observable estimation}
\label{sec:WYY2024}

Before explaining our proposal, we review the adaptive QGE algorithm for observable estimation proposed by Ref.~\cite{wada2024Heisenberg} and describe how their method fits within the generalized framework. In their approach, the subroutine \(\mac{U}_{\Upsilon, \mathrm{WYY}}\) generates $R^{(q)}$ copies of a following probe state
\begin{align}
    \ket{\Upsilon (q) }_{\text{WYY}} \simeq \frac{1}{\sqrt{2^{pM}}} \sum_{\bm x \in G_p^M} 
    e^{2\pi i 2^p \sum_j x_j  \frac{2^q}{\pi} \bigl(\langle O_j \rangle - u_j^{(q)}\bigr) } \ket{\bm x}.
\end{align}
In this case, the amplitude \(c_{\bm{x}}\) is uniform and equals \(1/\sqrt{2^{pM}}\) for all \(\bm{x} \in G_p^M\), as in the original QGE algorithm \cite{jordan2005fast,gilyen2019gradient}. As discussed later, the preparation of one target probe state
requires \(\mathcal{O}(2^{p+q+2} \sigma) = \mathcal{O}(2^{p+q} \sqrt{M \log d})\) queries to \(U_\psi\) and its inverse. 
Based on numerical evaluation and Hoeffding's inequality, it suffices to set \(p = 3\) (which is smaller than the value of $p$ determined by an analytical approach), and the number of samples is bounded as \(R^{(q)} \le 9 \log(M/\delta^{(q)})\). Consequently, the total cost of $\mac{U}_{\Upsilon,\mtr{WYY}}^{(q)}$ is given by the product of \(R^{(q)}\) and the query complexity of preparing a single probe state: $\text{(Cost of } \mathcal{U}_{\Upsilon, \mathrm{WYY}}^{(q)} \text{)} =\mac{O} \qty(2^{p+q} \sqrt{M \log d })\cdot R^{(q)}  = \mac{O} \qty( \aleph \cdot 2^q \log(1/\delta^{(q)}) )$, where \(\aleph \propto \sqrt{M \log d} \log M\) is independent of both \(q\) and \(\varepsilon\). This implies that the adaptive QGE algorithm achieves HL scaling.

We now briefly review the query complexity of the preparation of one $\ket{\Upsilon(q)}_{\text{WYY}}$.
The preparation of $\ket{\Upsilon(q)}_{\text{WYY}}$ is achieved by 
applying to a superposition state $\ket{+}^{\otimes pM}$ a subroutine that approximately imposes a phase $e^{2\pi i 2^p \sum_j x_j  \frac{2^q}{\pi} \bigl(\langle O_j \rangle - u_j^{(q)}\bigr) }$ on each computational state $\ket{x_1}\cdots\ket{x_j}\cdots\ket{x_M}$.  
Such a procedure consists of
three main steps: (i) block encode a linear combination of observables with random coefficients, (ii) amplify the magnitude of the block-encoding via uniform singular value amplification, and (iii) apply quantum singular value transformation (QSVT) to the Hamiltonian obtained from the amplified block-encoding in conjunction with state-preparation oracles.

Before reviewing each step, we clarify the notation for the target observables. As marked in Sec.~\ref{sec:Prob_setting}, we denote \(B_j\) as an \(a\)-block-encoding of a \(d \times d\) observable \(O_j\).
Given the values of temporal estimates $\{u_j^{(q)} \}_{j=1}^M$, we can implement the $(a+3)$-block-encoding $B_j^{(q)}$ of $O_j^{(q)} \coloneq (O_j - u_j^{(q)}\cdot \1 )/2$ \cite{wada2024Heisenberg}. Since the spectral norm of \(O_j^{(q)}\) is still bounded by 1, we can handle \(B_j^{(q)}\) in the same manner as \(B_j\), up to the difference in the number of ancilla qubits.

We now begin with step (i). For any \(\bm x = (x_1, \ldots, x_M) \in G_{p}^M\), a modified quantum circuit \(\tilde{U}^{(\bm x)}\) is constructed, controlled by \(\bm x\), which is a block-encoding of
\begin{align}
    \tilde H^{(\bm x)} \coloneq \frac{1}{M} \sum_{j=1}^M x_j O^{(q)}_j,
\end{align}
by employing the LCU method with controlled versions of each \(B_j\). This implementation of an \(\epsilon\)-precise block-encoding of \(\tilde{H}^{(\bm{x})}\) requires a circuit depth of \(\mac{O}( pM \log(1/\epsilon))\) over \(\mac{O}(pM + a + \lceil \log_2 M \rceil + \log_2 d)\) qubits \footnote{More precisely, Ref.~\cite{wada2024Heisenberg} provides an \(\epsilon\)-precise block-encoding of \(\left(\frac{2M}{\pi 2^{\lceil \log_2 M \rceil}}\right)\), but we ignore this normalization factor as it can be adjusted to unity via uniform singular value amplification.}.

To achieve a speedup with respect to \(M\), one employs singular value amplification, as described in Appendix~\ref{sec:Quantum_arithmetic}, Lemma~\ref{lem:uniform_amplification}, during step~(ii). More concretely, we construct a block-encoding of the rescaled Hamiltonian
\begin{align}
    \label{eq:rescaled_Hamiltonian}
    H^{(\bm x)} = \frac{1}{\sigma} \sum_{j=1}^M x_j O^{(q)}_j, 
    \mbox{where}~ 
    \sigma = \sqrt{2 \left\| \sum_j (O^{(q)}_j)^2 \right\| \log(2d/\delta')} < M
\end{align}
for some fixed $\delta'$.
Note that this construction is valid for almost all \(\bm x \in G_p^M\), though not for all. However, the contribution from the exceptional cases is ultimately truncated during the measurement process and does not affect the estimation.

From this amplified block-encoding, the QSVT technique allows us to construct for any $t\in \mathbb{R}$ a phase oracle that encodes $e^{i t \sum_j x_j \langle O_j^{(q)} \rangle}$ controlled by ${\bm x}$ in step (iii). 
Based on the amplified block-encoding and a state-preparation oracle \(U_\psi\), we can construct an approximate block-encoding of 
\[
\sum_{\bm x \in G_{\bm p} } \left(\frac{x_j \langle O^{(q)}_j \rangle}{\sigma}\right) \ketbra{\bm x}.
\]
For this operator, the optimal Hamiltonian simulation, described in Appendix~\ref{sec:Quantum_arithmetic}, Lemma \ref{lem:optimal_HS}, implements an \(\epsilon'\)-precise block-encoding of the time-evolution operator \(e^{i \tilde f(\bm x) T}\) for $T \in \mathbb{R}/\{0\}$, where \(\tilde f(\bm x)\) approximates \(\sigma^{-1} \sum_j x_j \langle O^{(q)}_j \rangle\) for almost all \(\bm x \in G_p^M\).
This protocol requires \( \mac{O}(T + \log(1/\epsilon'))\) queries. Then, the preparation of one copy of $\ket{\Upsilon(q)}_{\text{WYY}}$  is achieved by applying this block encoding of \(e^{i \tilde f(\bm x) T}\) with \(T = 2^{p+q+2}\sigma \) to the superposition state \(\ket{+}^{\otimes pM}\) in the probe system and an ancilla state \(\ket{\bm 0}\). This procedure has complexity \(\mathcal{O}(T + \log(1/\varepsilon')) = \mathcal{O}(2^{p+q+2} \sigma) = \mathcal{O}(2^{p+q} \sqrt{M \log d})\), since \(\varepsilon'\) and \(\delta'\) can be chosen as constants. This concludes the discussion on the preparation of a single target probe state.

Note that although a quantum circuit for the Hamiltonian simulation of the block-encoded Hamiltonian requires a large number of parameterized gates for quantum signal processing (QSP), the number of parameterized gates is significantly reduced by using the QSP for the Chebyshev polynomials. 
This leads to an exponential reduction of classical computational resources to synthesize the corresponding quantum circuit at the expense of increasing the queries to the block-encoding by a factor of about 3; see Ref.~\cite{wada2024Heisenberg}.
In this work, our primary goal is to minimize total query complexity for numerical evaluation, so we adopt the standard optimal Hamiltonian simulation in our proposals.

\section{Our proposal}
\label{sec:all-symemtry-diven-QGE}

In this section, we propose methods to enhance the efficiency of the QGE algorithm by modifying the initial state Eq.~\eqref{eq:super_position}, by leveraging the symmetry inherent in the target state $\ket{\psi}$ for multiple observable estimation, and by incorporating auxiliary qubits. In Sec.~\ref{subsec:modify_probe_state}, we analyze various initial states to determine which amplitude $c_{\bm x}$ yields superior performance for the original QGE algorithm. In Sec.~\ref{subsec:Symmetry-QGE-algorithm}, we propose Method I
which fully exploits the symmetry inherent in the target state $\ket{\psi}$ to further enhance the efficiency. In Sec.~\ref{subsec:parallel_QGE_alg}, we propose a new state-preparation circuit, inspired by a parallel scheme,
that utilizes entanglement with. 
Finally, in Sec.~\ref{subsec:Integral-QGE}, we integrate these approaches into Method II. 


\subsection{Modified initial state for the QGE algorithm}
\label{subsec:modify_probe_state}
From the previous section, we observed that the Jordan algorithm can be regarded as an application of the QPE algorithm. This observation implies that the efficiency of the QGE algorithm may be further enhanced by employing an alternative initial state that has been proposed in earlier studies for the QPE algorithm. In particular, when the cost function is given by the Holevo variance approximating the MSE in the QPE algorithm, the optimal initial state of the probe register is known to be the sine state described in Sec.~\ref{sec:HL_amplitude} \cite{ji2008parameter,berry2009perform}. 
However, when one also considers the additive error and the failure probability inherent in QPE, 
another initial state  
can more effectively lower the failure probability in the asymptotic limit~\cite{rendon2022effects, greenaway2024case, patel2024optimal}. Because reducing the failure probability while using minimal ancillary resources directly improves $R^{(q)}$ in the adaptive QGE algorithm, our focus in this section is on optimizing the initial state to further enhance QGE performance.

Ref.~\cite{wada2024Heisenberg} adopts an uniform superposition state as the initial state:
\begin{align}
    \ket{\text{init}}_{\text{sp}} = \sum_{\bm x \in G_p^M} c_{\bm x }^{(\text{sp)} } \ket{\bm x}  \coloneq  \frac1{\sqrt{2^{pM}}} \sum_{\bm x \in G_p^M} \ket{\bm x},
\end{align}
where $c_{\bm x}^{(\text{sp})} = 1/\sqrt{2^{pM}} $ denotes the amplitude of the uniform superpositon state over the grid $G_p^M$. 
When estimating a gradient component of $f(\bm x) = \bm g \cdot \bm x$ where $\bm x \in \mathbb{R}^M$ with the QGE algorithm, we apply the oracle that embeds all $g_j$ as phase and the variant of QFT to this state. Consequently, we obtain the following state,
\begin{align}
    \label{eq:superposition_with_phase}
     \frac{1}{\sqrt{2^{pM}}} \sum_{\bm k ,\bm x \in G_p^M} e^{2\pi i 2^p \bm x \cdot (\bm g - \bm k)} \ket{\bm k},
\end{align}
For this state, we evaluate the probability that the additive error between each measurement outcome \(k_j\) and the target value \(g_j\) is bounded. This probability admits the following tight bound:
\begin{align}
\label{eq:tight_eval_estimate}
    \forall j \in \{1,\ldots,M\}, \quad \Pr\!\Bigl[ \bigl| k_j - g_j \bigr| > \frac{1}{2\pi} \Bigr] < 0.18.
\end{align}
This inequality holds when \(p = 3\). To ensure success in estimation via coordinate-wise medians, the error probability must be below \(1/2\); thus, \(p = 3\) is the minimal choice that satisfies this requirement.

Alternative initial states, however, may yield improved estimation performance, so we investigate the performance of various initial states. To end this, recall the observation that the state expressed in Eq.~\eqref{eq:superposition_with_phase} is equivalent to $M$ tensor product of final states of the QPE algorithm by replacing the parameters appropriately, as discussed in Sec.~\ref{sec:core_idea_Jordan}.
From this observation, we can analyze the property regarding a measurement outcome of the QGE algorithm through the lens of the QPE algorithm. 

Let \(\theta \in [0,1)\) be the target phase. We consider the initial state of the QPE algorithm as $
\sum_{\mu = 0}^{2^p - 1} c_\mu \ket{\mu} $ where $\sum \abs{c_\mu}^2 =1 $. 
After applying a series of gates that embed the target phase \(\theta\) followed by the inverse QFT, the resulting state can be expressed as
\begin{align}
     {\ket{\Gamma } = \frac{1}{\sqrt{2^p}} \sum_{\mu=0}^{2^p-1} c_\mu \sum_{k'=0}^{2^p-1} e^{2\pi i \left(\theta - \frac{k'}{2^p}\right)\mu} \ket{k'}.}
\end{align}
The probability of obtaining the \(l\)-th outcome in the computational basis is given by
\[
P(l\mid\theta) = \left|\langle l| {\Gamma} \rangle\right|^2.
\]
Accordingly, if we set an estimator of $\theta$ as $\hat{\theta} = l/2^p$, the probability that the estimation error exceeds the threshold can be written as
\begin{align}
    \Pr\!\Bigl[ \Bigl|\hat{\theta} - \theta\Bigr| > \frac{1}{2\pi} \Bigr] = \sum_{l \in S} P(l\mid \theta),
\end{align}
where \(S\) denotes the set of outcomes for which 
\(\Bigl|\frac{l}{2^p} - \theta \Bigr| > \frac{1}{2\pi}\).
Note that this probability corresponds to the probability term in Eq.~\eqref{eq:tight_eval_estimate}. 
Hence, if we choose \(c_\mu = 1/\sqrt{2^p}\) for all \(\mu\), the upper bound of this probability is given by 0.18, as shown in Eq.~\eqref{eq:tight_eval_estimate}.

Following this formulation, we numerically evaluated the error probability for several promising probe states. As candidates, we consider the probe states proposed in Refs.~\cite{rendon2022effects, greenaway2024case, patel2024optimal}. However, here we define the amplitude using \(\phi(\mu) = \frac{2\mu - 2^p + 1}{2^{p+1}}\) to simplify the correspondence between the computational basis \(\ket{\mu}\) and the basis used in the QGE algorithm. Hence, the amplitudes are given by
\begin{gather}
    \label{def:sp_state}
     c_{\mu}^{\mtr{(sp)} } = \frac{1}{\sqrt{2^p}}, \\
     \label{def:cosine_state}
     c_{\mu}^{(\cos,1)} = \frac{\sqrt{2}}{\sqrt{2^p+1}} \cos\Biggl(\frac{2^p\phi(\mu) \pi}{(2^p+1)}\Biggr), \\ 
     \label{def:cosine2_state}
     c_{\mu}^{(\cos,2)} = \frac{\sqrt{2}}{\sqrt{2^p}} \cos( {\phi(\mu)} \pi ), \\
     \label{def:Kaiser_state}
     c_{\mu}^{(\mtr{Kaiser})} \propto \frac{I_0 \Bigl(\pi \alpha \sqrt{1-(2\phi(\mu))^2}\Bigr)}{I_0(\pi \alpha)},
\end{gather}
where \(I_0\) is the modified Bessel function of the first kind of order 0 and $\alpha$ is a non-negative parameter. 
Figure~\ref{fig:compare_various_probe} shows the numerical results for $p=3$.
\begin{figure}[t]
    \centering
    \includegraphics[width=0.7\linewidth]{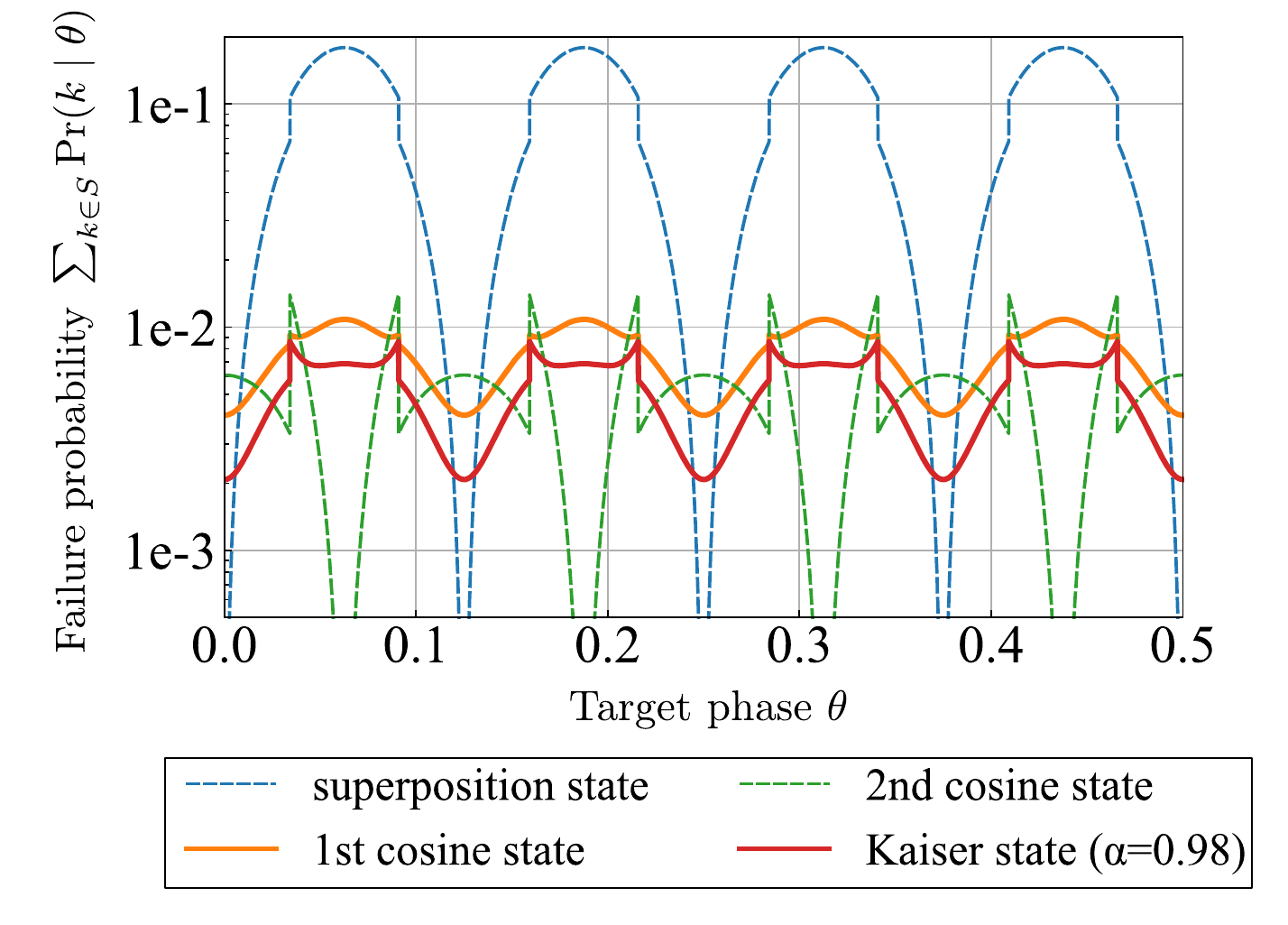}
    \caption{
    \justifying{
    Performance of various $(p=3)$-qubit probe states. The interval $\theta \in [0, 1]$ is discretized into $10^6$ points. For clarity, we present the failure probability in the range $\theta \in [0, 0.5]$, as the plot over $\theta \in [0.5, 1]$ exhibits identical behavior due to periodicity. The probe states evaluated are the superposition state defined in Eq.~\eqref{def:sp_state}, the first cosine state in Eq.~\eqref{def:cosine_state}, the second cosine state in Eq.~\eqref{def:cosine2_state}, and the Kaiser state with $\alpha = 0.98$ defined in Eq.~\eqref{def:Kaiser_state}.
    }}
    \label{fig:compare_various_probe}
\end{figure}
From this calculation, we obtain the maximum failure probability for each probe state,
\begin{gather}
    \max_{\theta \in [0,1) } \mtr{Pr}\!\left[ \left|\hat{\theta} - \theta\right| > \frac{1}{2\pi}\right]_{\mtr{sp}} 
     = 0.1789\ldots < 0.18,
    \\
    \label{eq:prob_sine_state}
    \max_{\theta \in [0,1) } \mtr{Pr}\!\left[ \left|\hat{\theta} - \theta\right| > \frac{1}{2\pi}\right]_{\cos,1} = 0.0108\ldots  < 0.011.
    \\
    \max_{\theta \in [0,1) } \mtr{Pr}\!\left[ \left|\hat{\theta} - \theta\right| > \frac{1}{2\pi}\right]_{\cos,2} = 0.0139  \ldots < 0.014
    \\
    \max_{\theta \in [0,1) } \mtr{Pr}\!\left[ \left|\hat{\theta} - \theta \right| > \frac{1}{2\pi}\right]_{\mtr{Kaiser}, \alpha=0.98} = 0.00860  \ldots < 0.009
\end{gather}
Here, we choose \(\alpha = 0.98\), as the Kaiser state shows the best performance near this value within the range \(\alpha \in [0, 30]\).
Based on these results, the Kaiser state with \(\alpha = 0.98\) provides the lowest error probability. However, due to the implementation difficulties associated with the Kaiser state \cite{obrien2024quantum} and the relatively small performance difference compared to the cosine state, we choose the cosine state defined in Eq.~\eqref{def:cosine_state} as the new probe state and denote $c_x^{(\cos)} \coloneq \frac{\sqrt 2}{\sqrt{2^p+1}} \cos(2^px\pi / (2^p +1))$ for all $x \in G_p$. 

We note that we have imposed symmetry on the amplitudes, i.e., \(c_{x} = c_{-x}\) for all \(x \in G_p\), in order to ensure that the subsequent amplification of the block-encoding (discussed in Sec.~\ref{subsec:Symmetry-QGE-algorithm}) succeeds. Additionally, we define a crucial factor \(v\) for numerical evaluation, which is determined by \(c_x\). Consider a random variable \(X\) on \(G_p\), with the probability given by \(\mathrm{Pr}[X = x] = |c_{x}|^2\), and \(v\) is defined as
\begin{align}
    v \coloneq \mathbb{E}[(2X)^2] = \sum_{x \in G_p} |2x|^2 |c_{x}|^2.
\end{align}
When we choose \(c_x^{\cos}\) for \(p = 3\), this value is bounded by 0.1652.

Consequently, our objective in the QGE algorithm is to 
construct the probe-state preparation whose amplitude profile matches the amplitude $c_{\bm x}^{(\cos)}$ defined as 
\begin{align}
    c^{(\cos) }_{\bm x} \coloneq c^{(\cos) }_{x_1}\, c_{x_2}^{(\cos)} \cdots c^{(\cos) }_{x_M},
\end{align}
for $\bm x \in G_p^M$.



\subsection{Method I : QGE algorithm under symmetry}
\label{subsec:Symmetry-QGE-algorithm}
In this section, we aim to accelerate the implementation of $\mac{U}_{\Upsilon}^{(q)}$ in our framework by exploiting symmetry conditions. The motivation behind imposing symmetry is twofold: firstly, enforcing a symmetry on the state of our interest is a natural and modest assumption; secondly, such symmetry often induces a decomposition of the original Hilbert space into smaller subspaces, which reduces the computational cost by applying QSVT restricted to a specific subspace—namely, subspace QSVT (See the details in Section \ref{sec:subspace_QSVT}). Based on these insights, we develop an efficient probe-state preparation scheme capable of handling a broad class of physical states, referring to this approach enhanced by symmetry condition as Method~I. The primary objective of this section is to evaluate the query complexity of our proposed method.

In order to state our main theorem, we initially outline the essential tools and frameworks that will be used throughout the analysis. We begin by the symmetry condition \cite{mitsuhashi2024unitary} to formally define symmetry and to utilize the subspace QSVT.

\begin{Def}
    For a group $\mac{G}$ and its representation $\mac{R}$ on the $n$-qubit Hilbert space $(\mathbb{C}^2)^{\otimes n} $, a linear operator $A$ on $(\mathbb{C}^2)^{\otimes n} $ is $\mac{G}$-symmetric if $A$ commutes with $\mac{R}(g)$ for all $g \in \mac{G}$. 
\end{Def}

The central concept of Method I is to shrink an effective size of a target Hilbert space by leveraging symmetries inherent in a target state. More specifically, in the expression for \(\sigma\) in Eq.~\eqref{eq:rescaled_Hamiltonian}, we can replace the collection \(\{O_j\}\) and the dimension \(d\) with the corresponding elements characterized by symmetry. This insight allows us to establish a relation between the MSE \(\varepsilon\) and the total number of queries to the state-preparation subroutine when a target state is restricted to a particular subspace.

\begin{theorem}[Method I]
\label{thm:evaluation_symmetry_QGE}
Let $\varepsilon \in (0,1)$ be a target precision and $\mac{G}$ be a group. For given $M$ $\mac{G}$-symmetric observables $\{O_j\}_{j=1}^M$ of bounded spectral norm $\norm{O_j}\leq 1$, assume that each $O_j$ has the irreducible decomposition $
O_j = \bigoplus_{\lambda \in \Lambda} \1(\mathbb{C}^{d_\lambda}) \otimes O^{(\lambda)}_j,
$
where $\lambda$ labels the irreducible subspaces and each $O^{(\lambda)}_j$ is an $m_\lambda$-dimensional operator.
For a target state $\ket{\psi}$ prepared by a given oracle $U_{\psi}$, we further assume that there exists an index subset $\Delta\subseteq\Lambda$ such that the irreducible subspaces specified by $\lambda\in \Delta$ contain the target state $\ket{\psi}$.
Then, there exists a quantum algorithm that outputs a sample from estimators $\{\hat{u}_j\}_{j=1}^M$ for the expectation values $\{\langle \psi|O_j|\psi \rangle\}$ satisfying
\begin{align}
    \label{eq:MSE_symmetric_QSVT}
    \max_{j=1,2,\ldots,M} \operatorname{MSE}[\hat{u}_j] \le \varepsilon^2 
\end{align}
using
\begin{align}
    \label{eq:query_comp_of_methodI}
    \varepsilon^{-1}\times \mac{O} \left(
    \log(M)\log^{1/2}(m_{\Delta})
    \max_{\lambda \in \Delta}~{ \norm{ \sum_{j=1}^M [O_j^{(\lambda)} ]^2 }^{1/2}}\right)
\end{align}
queries to the state preparation oracles $U_\psi$ and $U_\psi^\dagger$ in total, where $m_{\Delta}$ is defined as $m_{\Delta}:=\sum_{\lambda'\in \Delta}m_{\lambda'}$.
Here, the mean squared error of an estimator $\hat{u}_j$ is defined as
\[
\operatorname{MSE}[\hat{u}_j] := \mathbb{E}\Bigl[\Bigl(\hat{u}_j - \langle \psi|O_j |\psi\rangle\Bigr)^2\Bigr].
\]
\end{theorem}
\begin{Remark}
    We note that in this theorem, we use \(m_\lambda\) to denote the dimension of the irreducible (symmetric) subspace labeled by \(\lambda\), instead of \(d_\lambda\) as shown in the letter. The reason for using \(m_\lambda\) is that the symmetry inherent in the target state of interest and the target observable is defined by the unitary representation \(\mathcal{R}\) and the group \(\mathcal{G}\). From this perspective, the multiplicity \(m_\lambda\) corresponds to the dimension of the target observable \(O^{(\lambda)}\) on the irreducible (symmetric) subspace.
\end{Remark}


From the discussion in the previous section, we determine the probe state that exhibits superior performance. The next step toward enhancing the adaptive QGE algorithm with the cosine state is to construct a quantum circuit \(\mathcal{U}_{\Upsilon}^{(\cos)}\) that (approximately) generates the probe state \(\ket{\Upsilon(q)}_{\cos}\) at each iteration \(q\), defined as
\begin{align}
    \label{def:cosine_probe_state}
    \ket{\Upsilon(q)}_{\cos} \coloneq \sum_{\bm x \in G_p^M} c_{\bm x}^{(\cos)}
    e^{2 \pi i 2^p \sum_{j=1}^M x_j 2^q \pi^{-1} ( \langle O_j \rangle - \tilde{u}_j^{(q) } ) } \ket{\bm x}. 
\end{align}
The circuit \(\mathcal{U}_{\Upsilon}^{(\cos)}\) is constructed by a phase embedding circuit implementing the phase
$e^{2\pi i 2^p \sum_j x_j  \frac{2^q}{\pi} \bigl(\langle O_j \rangle - u_j^{(q)}\bigr)} $
to a cosine state over the grid \(G_p^M\), expressed as
\begin{align}
    \ket{\cos}^{\otimes M} \coloneq \left(\sum_{x \in G_p} c_x^{(\cos)} \ket{x}\right)^{\otimes M} = \sum_{\bm{x} \in G_p^M} c_{\bm{x}}^{(\cos)} \ket{\bm{x}}.
\end{align}
As described in Section \ref{sec:WYY2024}, the implementation of the phase embedding circuit proceeds in three main steps: (i) constructing a block-encoding for a linear combination of observables with random coefficients, (ii) enhancing this block-encoding via uniform singular value amplification, and (iii) generating the target probe state by applying the QSVT to the Hamiltonian derived from the amplified block-encoding together with the state-preparation oracles.

To prove this theorem, we first see that the imposed symmetry condition induces a decomposition of the original Hilbert space into a direct-sum structure of smaller subspaces or, more generally, an inequivalent irreducible decomposition. In what follows, we denote the set of all $\mac{G}$-symmetric $n$-qubit linear operators as $\mac{L}_{n,\mac{G}}$: 
  \begin{align}
    \mathcal{L}_{n,\mathcal{G}} \coloneq \{ L \in \mac{L} ((\mathbb{C}^2)^{\otimes n}) \mid [L, \mac{R}(g) ] =0, \forall g \in \mac{G} \}.
  \end{align}
From definition, \(\mathcal{G}\)-symmetric operators can be expressed in an inequivalent irreducible decomposition. To make this explicit, we first note that every unitary representation \(\mac{R}\) on \((\mathbb{C}^2)^{\otimes n}\) can be decomposed as follows,
\begin{align}
    \mac{R}(g) = \bigoplus_{\lambda \in \Lambda} \mac{R}_{\lambda}(g) \otimes \1(\mathbb{C}^{m_\lambda}) 
\end{align}
where $\lambda $ is the label of each irreducible space, $\Lambda$ is the set of the labels, $\mac R_\lambda$ is an irreducible representation of dimension $d_\lambda$, $\1(\mathbb{C}^{m_\lambda}) $ is the identity on $\mathbb{C}^{m_\lambda}$ where $m_\lambda$ is the multiplicity of each irreducible space, respectively. 
This decomposition implies the following isomorphism between the Hilbert space of $n$ qubits 
\begin{align}
    (\mathbb{C}^2)^{\otimes n} \simeq \bigoplus_{\lambda \in \Lambda}\mathbb{C}^{d_\lambda} \otimes \mathbb{C}^{m_\lambda},
\end{align}
Due to Schur's lemma, arbitrary $\mac G$-symmetric operator $O \in \mathcal{L}_{n,\mac G} $ also has the decomposition based on inequivalent irreducible representation, 
\begin{align}
    O = \bigoplus_{\lambda \in \Lambda} \1(\mathbb{C}^{d_\lambda}) \otimes O^{(\lambda)},
\end{align} 
where $O^{(\lambda)}$ acts on the multiplicity space $\mathbb{C}^{m_\lambda}.$
This direct sum structure precisely enables us to directly reduce the computational cost of $\mac{U}^{(\cos)}_{\Upsilon}$.

When $M$ observables $\{O_j \}_{j=1}^M$ are $\mac{G}$-symmetric and they have the inequivalent irreducible decomposition on the considering basis,
\begin{align}
    \forall j,   O_j  =
    \bigoplus_{\lambda \in \Lambda}\1(\mathbb{C}^{d_\lambda}) \otimes O_j^{(\lambda)},
\end{align}
the linear combination of $O_j$ also can be irreducibly decomposed. From this observation, for $\bm x \in G_{ p}^{M} $, the following equality holds, 
\begin{align}
    \qty( \sum_{j=1}^M x_j  O_j)  &= \sum_{j=1}^M x_j \bigoplus_{\lambda \in \Lambda} \1(\mathbb{C}^{d_\lambda}) \otimes O_j^{(\lambda)} \\
     &=
    \begin{pmatrix}
    \sum_{j=1}^M x_j \1(\mathbb{C}^{d_{1} }) \otimes O_j^{(1)}  & \bm O & \cdots & \bm O  \\
    \bm O  & \sum_{j=1}^M x_j \1(\mathbb{C}^{d_{2} }) \otimes O_j^{(2)}  & \cdots & \bm O  \\
    \vdots & \vdots & \ddots & \vdots \\
    \bm O  & \bm O  & \cdots & \sum_{j=1}^M x_j \1(\mathbb{C}^{d_{\abs{\Lambda}} }) \otimes O_j^{(\abs{\Lambda})}
    \end{pmatrix}.
\end{align}

In estimating expectation values of a quantum system, it is natural to consider a situation where a quantum state is subject to certain symmetries, such as U(1) or SU(2) symmetry. Under such a situation, the target state \(\ket{\psi}\) is typically supported on specific subspaces, which implies that any operator acting on the orthogonal subspaces does not affect \(\ket{\psi}\). 
To be more concrete, let \(\Delta \subseteq \Lambda\) be the relevant set of labels, and define \(\mathcal{H}_\Delta\) as the corresponding target subspace. Then, the projector $\Pi_\Delta$ on the subspace $\mac{H}_\Delta$ satisfies $\Pi_\Delta (\1(\mathbb{C}^{d_\lambda}) \otimes O^{(\lambda)}) \Pi_\Delta = 0$ for $\lambda \not \in \Delta$.

Thus, when our target state $\ket{\psi}$ belongs to a certain subspace $\mac{H}_\Delta$, we need only implement operators whose action is effective on \(\mathcal{H}_\Delta\). Applying this logic in the state-preparation method $\mac{U}^{(\cos)}_{\Upsilon}$, it is sufficient to construct the block-encoding of \(\frac{1}{\sigma}\sum_j x_j O_j\) so that it acts effectively only within \(\mac{H}_\Delta\). This observation implies that leveraging symmetry allows us to amplify the amplitude of the block-encoding in the desired symmetry sectors via subspace QSVT. 

Here, we omit the details of subspace QSVT and state only the core lemma regarding amplitude amplification on specific subspaces. The proofs of Lemma~\ref{lem:Amplitude_Amp_subspace} and the detailed framework of subspace QSVT are provided in Section \ref{sec:subspace_QSVT}.

\begin{lemma}[Amplitude amplification on specific subspaces]
    \label{lem:Amplitude_Amp_subspace}
    Let $\delta, \varepsilon \in (0, 1/2)$ and $\gamma > 1$. Suppose we have a group $\mac G$ and a $a$-block-encoding $U_O$ of $\mac{G}$-symmetric $d$-dimensional Hermitian operator $O$ satisfying $O =  \bigoplus_{\lambda \in \Lambda} O^{(\lambda)} $, where we define $\Lambda $ as a label set and $O^{(\lambda)}$ is a corresponding summand, both of which are determined by $\mac{G}$.
    For a subset $\Delta \subseteq \Lambda$, let \(\Pi_\Delta\) denote the projector onto the subspace \(\mathcal{H}_\Delta\subseteq\mathbb{C}^d\) associated with \(\Delta\), and assume that $\Pi_\Delta$ satisfies $\|\Pi_\Delta O \Pi_\Delta\| \leq \frac{1 - \delta}{\gamma}$. Then, we can implement a $(a+1)$-block-encoding of $\tilde O$ satisfying $\norm{\Pi_\Delta(\tilde O  - \gamma O) \Pi_\Delta } \leq \varepsilon$, 
    with $m = \mac{O}\left(\frac{\gamma}{\delta} \log\left(\frac{\gamma}{\varepsilon}\right)\right)$ queries to $U_O$ or $U_O^\dagger$, $2m$ uses of NOT gates controlled by $a$-qubit, $\mac{O}(m)$ single-qubit gates, and $\mac{O}(\mathrm{poly}(m))$ classical computation to find quantum circuit parameters.
\end{lemma}

From this lemma, we obtain the amplification procedure only effective to specific subspaces. Applying this technique to $\frac1M \sum_j x_j O_j = \frac1M \sum_j x_j \bigoplus_{\lambda \in \Lambda} \1(\mathbb{C}^{d_\lambda}) \otimes O_j^{(\lambda)}$, we can obtain the following lemma,

\begin{lemma}
    \label{lem:construction_sigma_delta}
  Let \(\delta' > 0\), \({\varepsilon'} \in (0, 1/2)\), \(p\) be a positive integer, and \(\mathcal{G}\) be a group.
  Suppose that for $\mac{G}$-symmetric $d$-dimensional observables $\{O_j\}_{j=1}^M$ with $\norm{O_j} \leq 1$, we have access to the oracle 
  $U'_{\mtr{SEL}}= \sum_{\bm  x \in G_{ p}^M} \ketbra{\bm x} \otimes U^{(\bm x)}$ where $U^{(\bm x)}$ is an $(a + \lceil \log_2 M \rceil)$-block encoding of observable
  $\tilde H^{(\bm x)} \coloneq M^{-1} \sum x_j O_j$ for some $a \in \mathbb{N}$.
  Here, we denote the irreducible decomposition of $O_j$ as ${O_j = \bigoplus_{\lambda \in \Lambda}  \1(\mathbb{C}^{d_\lambda}) \otimes O^{(\lambda)}_j  }$, where $\Lambda $ denotes a label set, $O_j^{(\lambda)}$ is a corresponding summand, $d_\lambda$ is the dimension of a representation for $\mac{G}$, $\1 (\mathbb{C}^{d_\lambda}) $ is the identity on $\mathbb{C}^{d_\lambda}$, and $m_\lambda = \dim O_j^{(\lambda)}$. 
  
  For a subset $\Delta \subseteq \Lambda$ and independent symmetric distributed random variables $\{X_j \}_{j=1}^M$ supported on $G_{p}^M$ with variance $v \coloneq \mathbb{E}[(2X_j)^2]$, we assume that the following inequality holds,
  \begin{align}
        \label{eq:def_sigma_Delta2}
        \sigma_\Delta\coloneq \max_{\lambda\in \Delta}\sqrt{2  \norm{\sum_{j=1}^M  [O_j^{(\lambda)}]^2 } \log(2 m_\Delta/\delta')} < M,
\end{align}
    or 
\begin{align}
    \label{eq:def_sigma_Delta}
      \sigma_\Delta \coloneq   \max_{\lambda \in \Delta} \sqrt{2 v \norm{\sum_{j=1}^M  [O_j^{(\lambda)}]^2 } \log(2 m_\Delta /\delta')} + \frac43 \log(2m_\Delta /\delta')  < M,
  \end{align}
  where $m_{\Delta}$ is defined as $m_{\Delta}:=\sum_{\lambda'\in \Delta}m_{\lambda'}$.
  Then, there exists a subset $F_\Delta \subset {G_p^M}$ 
  satisfying 
  \begin{align}
      \mtr{Pr}_{X_1, \ldots X_M} [ \bm X \in F_{\Delta} ] \geq 1 - \delta'.
  \end{align}
  and we can implement 
  a unitary (with $pM + \lceil \log_2 M \rceil + \log_2 d + a + 1$ qubits in total),
  \begin{align}
    U_{\mtr{obs}, \Delta} \coloneq \sum_{\bm x \in {G_{p}^M}} \ketbra{\bm x} \otimes U_{\mtr{obs}, \Delta}^{(\bm x)}\;,
  \end{align}
  such that $U^{(\bm x)}_{\mtr{obs},{\Delta}}$ is a $(a + \lceil \log_2 M \rceil + 1)$-block-encoding of the Hamiltonian $H^{(\bm x)}_\Delta $ satisfying
  \begin{align}
  \label{eq:subspace_block_encoding}
    \norm{ \Pi_\Delta \qty(H^{(\bm x)}_\Delta  - \frac1{\sigma_\Delta  } \sum_{j=1}^M x_j O_j ) \Pi_\Delta}\leq \varepsilon'  \;, \ \text{if } \bm x \in F_{\Delta} \subset {G_p^M}\;,
  \end{align}
  where $\Pi_\Delta $ is a projector on the subspace associated with $\Delta$.  
  For $m = \mac{O}(M \sigma_\Delta^{-1} \log(M \sigma_\Delta^{-1}/\varepsilon')) $, this implementation of ${U_{\mtr{obs},\Delta}}$ requires $m$ queries to $U'_{\mtr{SEL}}$ or its inverse, $2m$ X gates controlled by $a + \lceil \log_2 M \rceil $ qubits, $\mac{O}(m)$ single-qubit gates, $\mac{O}(\mtr{poly}(m))$ classical precomputation.
\end{lemma}

\begin{Remark}
    In this lemma, we present two alternative choices for \(\sigma_\Delta\) as specified in Eqs.~\eqref{eq:def_sigma_Delta2} and \eqref{eq:def_sigma_Delta}. Specifically, Eq.~\eqref{eq:def_sigma_Delta} provides a numerically smaller value, which is advantageous for precise numerical analysis in Sec.~\ref{sec:k-RDM-with-QGE}, whereas Eq.~\eqref{eq:def_sigma_Delta2} yields a simpler expression that facilitates asymptotic evaluations easily. In what follows, we prove the case when \(\sigma_\Delta\) is chosen as in Eq.~\eqref{eq:def_sigma_Delta}; however, by following the same proof strategy and applying the matrix series inequality from Theorem 35 in Ref.~\cite{apeldoorn2023quantum}, one can similarly prove the result for the choice given in Eq.~\eqref{eq:def_sigma_Delta2}.
\end{Remark}

\begin{proof}
    Let $\gamma \coloneq  M/\sigma_\Delta  > 1$. To perform the amplification in Lemma \ref{lem:Amplitude_Amp_subspace}, it is required that the spectral norm $\norm{ \Pi_\Delta \tilde{H}^{(\bm x) } \Pi_\Delta  }$ is upper bounded by $(1 -\delta )/\gamma$ for some $\delta \in (0, 1/2)$. 
    Following the proof of lemma 12 in Ref.~\cite{wada2024Heisenberg}, let us consider independent random variable $X_j, (j=1,\ldots,M)$ symmetrically distributed on $G_p$. For these variables, the following equality satisfies,
    \begin{align}
      \norm{\Pi_\Delta  \qty( \sum_j 2X_j O_j ) \Pi_\Delta }
      &= \norm{  \sum_j 2X_j \bigoplus_{\lambda \in \Delta } \qty( \1(\mathbb{C}^{d_\lambda}) \otimes O^{(\lambda)}_j ) } \\ 
      &= \norm{ \bigoplus_{\lambda \in \Delta } \sum_j 2X_j \qty(\1(\mathbb{C}^{d_\lambda}) \otimes O^{(\lambda)}_j ) }  \\
      \label{eq:direct_sum_use}
      &= \max_{\lambda \in \Delta } \norm{ \sum_j 2X_j \qty(\1(\mathbb{C}^{d_\lambda}) \otimes O^{(\lambda)}_j ) } \\
      \label{eq:tensor_prod_use}
      &=  \max_{\lambda \in \Delta } \norm{ \sum_j2 X_j O_j^{(\lambda)} },
    \end{align}
    where we use the fact that $ \norm{A \oplus B} = \max(\norm{A}$, $\norm{B}) $  in Eq.~\eqref{eq:direct_sum_use} and $\|\bm{1}\otimes A\|=\|A\|$ in Eq.~\eqref{eq:tensor_prod_use}. From this expression, we can derive the following inequality,
    \begin{align}
          \label{Berry}
          \mtr{Pr}_{X_1, \ldots, X_M} \qty[ \max_{\lambda \in \Delta} \norm{ \sum_j 2X_j O_j^{(\lambda)} } < t ] 
          &=   \mtr{Pr}_{X_1, \ldots, X_M} \qty[  \bigcap_{\lambda \in \Delta}  \qty( \norm{ \sum_j 2X_j O_j^{(\lambda)} } < t ) ] \\
          \label{eq:estimate_symmetry_norm}
          &\geq 1 - \sum_{{\lambda\in \Delta}} \mtr{Pr}_{X_1, \ldots, X_M} \qty[  \norm{ \sum_j 2X_j O_j^{(\lambda)} } \geq t ].
    \end{align}
    In the first inequality, we use the union bound: $\mtr{Pr}  \qty(\bigcap_i E_i) \geq 1 - \sum_i \mtr{Pr}(E_i^c)  $ for events $E_i$ and its complement.

    Since the random variables $X_j$ are symmetrically distributed on $G_p$, for $m_\lambda \times m_\lambda$ matrices $O_j^{(\lambda)}$
    the matrix Bernstein inequality holds 
    (See Sec.~\ref{appsub:proof_matrix_ineq_subgaussian}),  
    \begin{align}
        \label{eq:matrix_inequality_for_subgaussian}
        \mtr{Pr}_{X_1,\ldots, X_M} \qty[ \norm{\sum_{j=1}^M (2X_j) O_j^{(\lambda)}} \geq t ] \leq 2m_{\lambda}\cdot \exp[ - \frac{t^2/2}{\mathbb{E}[(2X_1)^2] \norm{\sum_{j=1}^M (O^{(\lambda)}_{j})^2   } + t/3 } ], \quad \text{for each $\lambda \in \Delta$. }
    \end{align}
    By choosing $t = {\sigma_\Delta}$, the following equality for all $\lambda \in \Delta$ holds,
    $$
    \exp \qty( -\frac{t^2}{{2\mathbb{E}[(2X_1)^2] \norm{\sum_{j=1}^M (O^{(\lambda)}_j)^2   } + 2t/3}  } ) \leq \frac{\delta ' } {2 \sum_{\lambda \in \Delta} m_\lambda }
    $$
    because we can rewrite this inequality as follows, when denoting $L = \log( \frac{2 \sum_{\lambda \in \Delta} m_\lambda }{\delta '} )$,
    \begin{align}
        t^2 - \frac{2t}3 L - 2 L \qty( \mathbb{E}[(2X_1)^2] \norm{\sum_{j=1}^M (O_j^{(\lambda)})^2   } )  \geq 0.
    \end{align}
    We can solve this inequality regarding $t$, 
    \begin{align}
        \frac13 \sqrt{18 \mathbb{E}[(2X_1)^2] \norm{\sum_{j=1}^M (O_j^{(\lambda)})^2   } L + L^2} + \frac13 L \leq t .
    \end{align}
    and it is sufficient to take $t = \sigma_\Delta$ to hold this inequality because $\sqrt{a+b} \leq \sqrt{a} + \sqrt{b}$ for $a,b>0$.
    
    Thus, inserting $M/\gamma$ into $t$, we obtain
    \begin{align}
        \forall \lambda, 
        \mtr{Pr}_{X_1,\ldots, X_M} \qty[ \frac1M \norm{\sum_{j=1}^M X_j O_j^{(\lambda)}} \geq \frac{1}{2\gamma }  ] 
        &\leq 2m_\lambda \cdot \exp[ - \frac{\sigma_{\Delta}^2/2}{\mathbb{E}[(2X_1)^2] \norm{\sum_{j=1}^M (O_j^{(\lambda)})^2   } + \sigma_{\Delta}/3 } ],  \\
        &\leq \frac{m_\lambda}{{\sum_{\lambda \in \Delta} m_\lambda}} \cdot \delta'.
    \end{align}
    From this inequality and Eq.~\eqref{eq:estimate_symmetry_norm}, we can derive
    \begin{align}
        \mtr{Pr}_{X_1,\ldots, X_M} \qty[ \norm{\Pi_\Delta  \qty( {{M}^{-1}}\sum_j X_j O_j ) \Pi_\Delta } \leq \frac1{2\gamma} ] \geq 1 - \sum_{\lambda \in \Delta } \frac{m_\lambda}{{\sum_{\lambda \in \Delta} m_\lambda}} \cdot \delta' = 1-\delta'.
    \end{align}
    
    Denoting $F_\Delta$ as the set of events for which the independent symmetrically distributed random variables $X_j$ satisfy $\norm{\Pi_\Delta  \qty({M}^{-1} \sum_j X_j O_j ) \Pi_\Delta } \leq \frac1{2\gamma} $, we can express $\mtr{Pr}_{X_1, \ldots, X_M} [\bm X \in F_\Delta] \geq 1-\delta'$. Equivalently, if $\bm x \in F_\Delta $, we can amplify the block-encoding of $\tilde H^{(\bm x)} $, thereby complete Lemma \ref{lem:Amplitude_Amp_subspace}.
\end{proof}

Next, we aim to show two quantum algorithms to prepare a quantum state approximates the state defined in Eq.~\eqref{def:cosine_probe_state}. 
Note that we modified the QSVT circuit for Hamiltonian simulation to reduce the use of $U_\psi$ and its inverse from $4Q $ in Ref.~\cite{wada2024Heisenberg} to $2Q$, where $Q$ is a degree of a polynomial $f(x)$ approximating $e^{i x t} $. (See Figs.~\ref{fig:HS_state_prep_circ} and \ref{fig:Rep_circ}.) 

\begin{figure}[t]
    \centering
    \includegraphics[width=1.\linewidth]{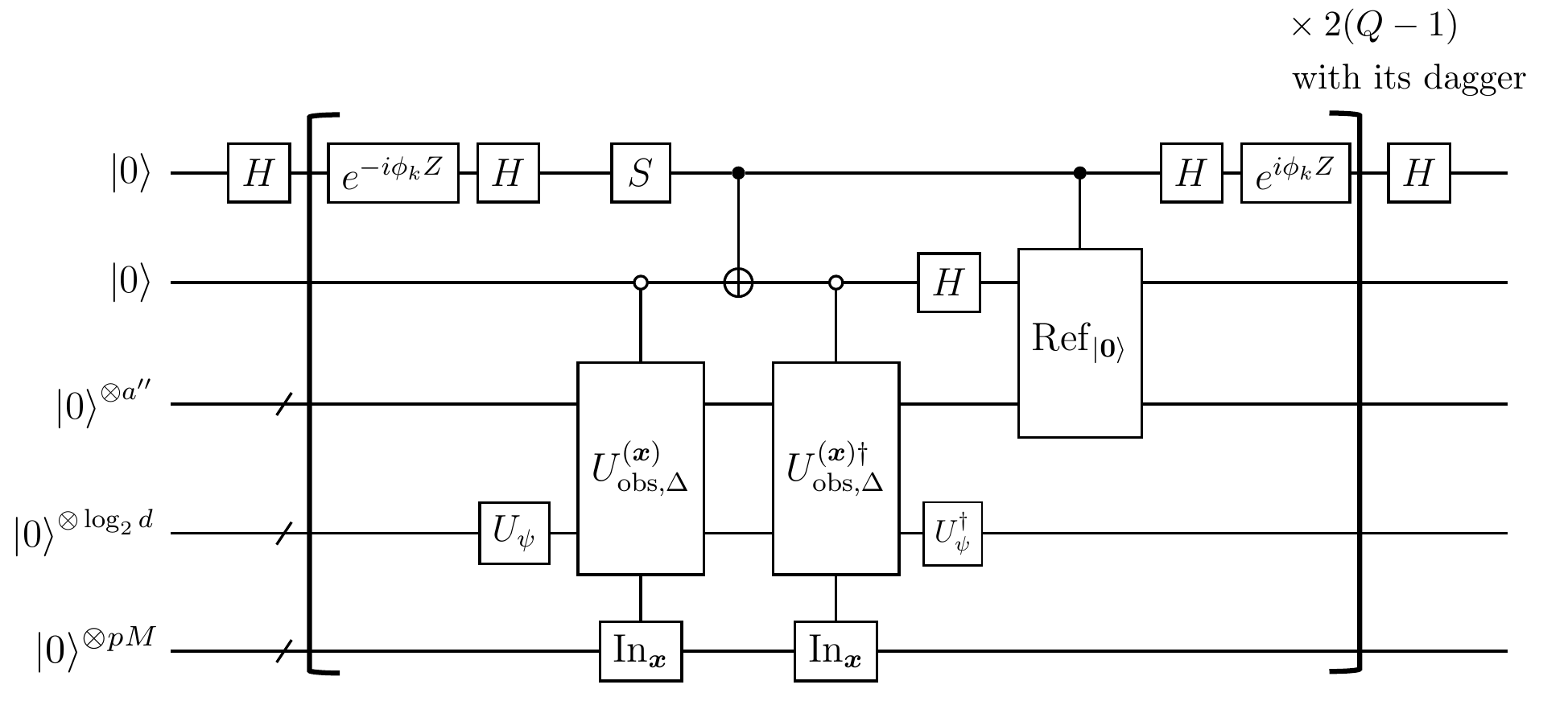}
    \caption{
    \justifying{
    Quantum circuit for the optimal Hamiltonian simulation for the Hamiltonian $\sum_{\bm x \in G_p^M} f(\bm x) \ketbra{\bm x} $ defined in Eq.~\eqref{eq:target_hamiltonian}. In this circuit, the unitary \( U_{\mtr{obs},\Delta}^{(\bm x)} \) is implemented with \( a' = a + \lceil \log_2 M \rceil + 1 \), as specified in Lemma~\ref{lem:construction_sigma_delta}. Here, \( Q = \mac{O}\Bigl(2^{q+p+1}\sigma_\Delta + \log(1/\varepsilon'')\Bigr) \), \( \mtr{Ref}_{\ket{\bm 0}} \) denotes the reflection about \( \ket{\bm 0} \) (i.e., \( \mtr{Ref}_{\ket{\bm 0}} \coloneq 2\ketbra{0} - \1 \)), and \( \mtr{In}_{\bm x} \) associated with $U_{\rm obs}^{({\bf x})}$ indicates that the unitary is an  \( \bm x \)-controlled operation. The phase angles $\{\phi_k\}$ for the top ancilla qubit are used for quantum signal processing to achieve \( \varepsilon'' \)-precise Hamiltonian simulation \cite{low2019hamiltonian}. 
    }}
    \label{fig:HS_state_prep_circ}
\end{figure}

\begin{figure}[t]
    \centering
    \includegraphics[width=1.0\linewidth]{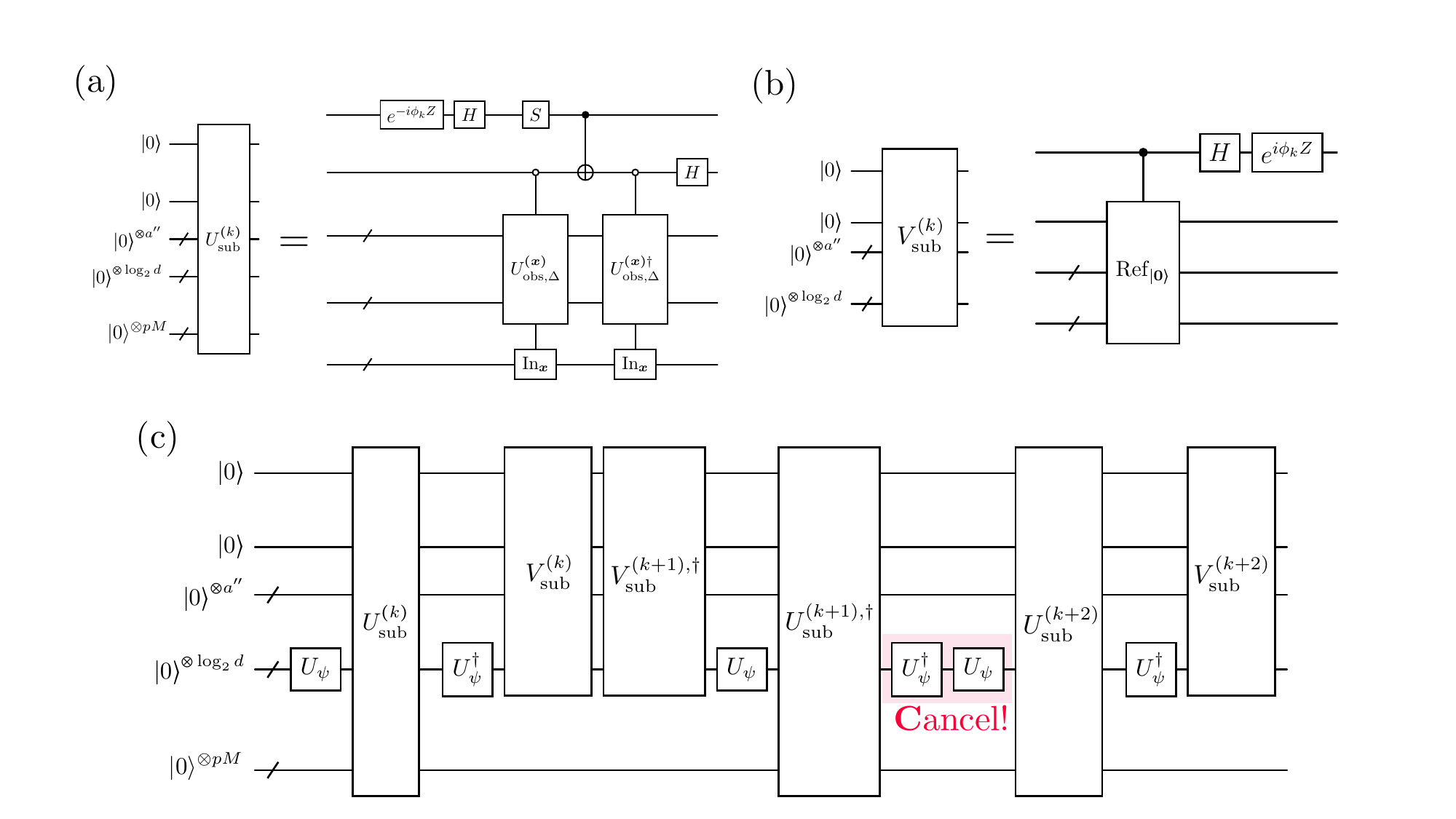}
    \caption{
    \justifying{ 
    The cancellation of $U_\psi$ and its inverse in the quantum circuit sequence shown in Fig.~\ref{fig:HS_state_prep_circ}. For clarity, we define (a) the subroutine circuit $U^{(k)}_{\mtr{sub}}$, which includes the $e^{i\phi_k Z}$ gate, and (b) the subroutine circuit $V^{(k)}_{\mtr{sub}}$, which includes the $e^{-i\phi_{k} Z}$ gate, where $\phi_k$ is the $k$-th phase angle of QSP. The unitary $U_{\mtr{obs}, \Delta}^{(\bm x)}$, the phase sequence $\{ \phi_k \}$, and the reflection $\mtr{Ref}_{\ket{\bm 0}}$ are defined as in Fig.~\ref{fig:HS_state_prep_circ}. (c) The explicit description of sequence of the quantum circuit depicted in Fig.~\ref{fig:HS_state_prep_circ}. Owing to this cancellation, the query complexity of the state preparation unitaries is calculated as $2Q$, where $Q$ is a degree of a polynomial $f(x)$ approximating $e^{ixt}$.
    }
    }
    \label{fig:Rep_circ}
\end{figure}

\begin{lemma}[Probe-state preparation method enhanced by symmetry condition]
  \label{lemma:symmetry_state_prep_by_Hamiltonian_simulation}
  Let $p$ be a positve integer, $\mac{G}$ be a group and $O_j$ be $\mac{G}$-symmetric $d$-dimensional observables with $\norm{O_j} \leq 1$. Suppose that $U_\psi$ is $\log_2 d$-qubit state preparation oracle, 
  $\langle O_j \rangle \coloneq \bra{\psi} O_j \ket{\psi}$, and $u_j \in [-1,1]$. 
  For a subset $\Delta$ of the label set for inequivalent irreducible subspaces defined by $\mathcal{G}$,
  the target state $\ket{\psi}$ is assumed to be supported on the subspaces specified by $\Delta$ and $\Pi_\Delta \ket{\psi} = \ket{\psi} $ holds for the corresponding projector $\Pi_{\Delta}$.
  Additionally, we have access to $U_{\mtr{SEL}}'= \sum_{\bm{x}\in G_p^M}\ketbra{\bm x }\otimes U^{(\bm x)} $ 
  where $U^{(\bm x)}$ is an $(a+\lceil \log_2 M\rceil)$-block-encoding of $M^{-1} \sum_{j=1}^M x_j O_j$ for some $a \in \mathbb{N}$.
  We assume $\sigma_\Delta$ is expressed in Eq.~\eqref{eq:def_sigma_Delta} and holds for $\delta ' =2^{-10}$ and $v \coloneq \mathbb{E}[(2X)^2]$ where \(X\) is a random variable on \(G_p\) with \(\mathrm{Pr}[X = x] = |c_{x}^{(\cos)}|^2\).
    
  Then, for any non-negative integer $q$, there exists a quantum circuit $\mac{U}_{\Upsilon, \text{sym}}$ preparing the $pM$-qubit probe state expressed by 
  \begin{align}
    \ket{\Upsilon(q)}_{\cos} = \sum_{\bm x \in G_p^M} c^{(\cos)}_{\bm x} \exp(2 \pi i 2^p \sum^M_{j=1} x_j \cdot \frac{2^{q}}{\pi} \qty( \langle O_j \rangle -u_j   ) )\ket{\bm x},
  \end{align}
  up to $1/12$ Euclidean distance error, with \(2Q\) uses of \(U_\psi\) or \(U_\psi^\dagger\), \(4mQ\) uses of controlled \(U'_{\mathrm{SEL}}\) (or its inverse), \(\mac{O}(mQ)\) uses of NOT gates controlled by at most \(\mac{O}\Bigl(a + \lceil \log_2 M \rceil + \log_2 d\Bigr)\) qubits, \(\mac{O}(mQ + pM)\) uses of single-qubit or two-qubit gates, an additional \( (\lceil \log_2 M \rceil + \log_2 d + a + 3) \) ancilla qubits, and \(\mac{O}\bigl(\mathrm{poly}(Q) + \mathrm{poly}(m)\bigr)\) classical precomputation for determining the quantum circuit parameters, where 
    \[
    Q = \mac{O}(2^{p+q+1} \sigma_\Delta ) \quad \text{and} \quad m = \mac{O}\Bigl( \frac{M}{\sigma_\Delta}(p+q+\log M)\Bigr).
    \]
\end{lemma}

\begin{proof}
    To prove our claim, we consider two quantum circuits, \(W\) and \(V\), which produce the state \(\ket{\Upsilon(q)}_{\cos}\). The circuit \(W\), constructed by uniform amplification (in Lemma.~\ref{lem:construction_sigma_delta}) and Hamiltonian simulation, encodes the information about the expectation values \(\langle O_j \rangle\), while \(V\) incorporates parameters \(u_j\) for each \(j \in \{1,\ldots,M\}\). In this proof,
    We first analyze the approximation error in \(W\) that arises from the use of Hamiltonian simulation. Next, we argue that the discrepancy between \(V\) and its ideal operation remains sufficiently small. With these observations in hand, and anticipating the final step's uniform amplification analysis, we derive an upper bound on the distance between the exact state \(\ket{\Upsilon}_{\cos}\) and the approximate probe state constructed by \(W\) and \(V\). Finally, we incorporate the error arising from uniform amplification into the bound, completing our estimation of the total distance.

    We begin with the analysis of $W$ regarding the Hamiltonian simulation.
  Let $\varepsilon'' = 2^{-10} \in (0,1)$ and $\varepsilon' = \sqrt{\delta'}/(2^{p+q+2 }{\sigma_{\Delta}}) \in (0,1/2)$.
  From the assumption and Lemma \ref{lem:construction_sigma_delta}, we have a unitary 
  \begin{align}
    U_{\mtr{obs}} \coloneq \sum_{\bm x \in G_p^M} \ketbra{\bm x} \otimes U_{\mtr{obs}, \Delta}^{(\bm x)}\;,
  \end{align}
  where $U^{(\bm x)}_{\mtr{obs}, \Delta}$ is a $(1, a'' = a+\lceil \log_2 M \rceil +1, \varepsilon')$-block-encoding of $H^{(\bm x)}_{\Delta}$ that satisfies
  $\norm{ \Pi_\Delta \qty(H^{(\bm x)}_\Delta  - {(\sigma_\Delta ) }^{-1} \sum_{j=1}^M x_j O_j ) \Pi_\Delta}\leq \varepsilon' $ if $\bm x \in F_{\Delta} \subset G_p^M$.
  Thus, using state-preparation oracle and this unitary, we can implement a $(1, a'' + \log_2 d, 0)$-block-encoding of the following Hamiltonian
  \begin{align}
    \label{eq:target_hamiltonian}
    \sum_{\bm x \in G_p^M} \tilde{f}(\bm x ) \ketbra{\bm x}\;, \text{where } \tilde{f}(\bm x) &\coloneq \bra{0}^{\otimes a''+\log_2 d} U_\psi^\dagger U^{(\bm x)}_{\mtr{obs}, \Delta} U_\psi \ket{0}^{\otimes a'' +\log_2 d}\;.
  \end{align}
  For this encoding, the Hamiltonian simulation circuit described in Fig.~\ref{fig:HS_state_prep_circ} yields the quantum circuit $W$ for a $(1, a'' + \log_2 d + 2, \varepsilon'')$-block-encoding of time evolution operator 
  \begin{align}
    \label{eq:def_tilde_f}
      \sum_{\bm x \in G_p^M} e^{i \tilde f(\bm x )t} \ketbra{\bm x}, \quad t := 2^{p+q+1} {\sigma_{\Delta}}.
  \end{align}
    We remark that a block encoding $W'$ of $\varepsilon''$-precise time evolution operator $e^{iHt}$ for some Hamiltonian yields $O(\sqrt{\varepsilon''})$-precise 
  time evolved state, 
  \begin{align}
    \label{eq:upper_bound_state_block_encoding}
    \forall\ket{\psi},  \norm{W' \ket{\bm 0 } \ket{\psi} - \ket{\bm 0} e^{iHt} \ket{\psi}} \leq \varepsilon'' + \sqrt{2 \varepsilon''}\;.
  \end{align}
  
  Additionally, we consider the following circuit $V^{(q)}$ such that 
    \begin{align}
        \norm{V^{(q)}  - \sum_{\bm x \in G_p^M } e^{- 2 \pi 2^p \sum_j x_j 2^q \pi^{-1} u_j} \ketbra{\bm x} } \leq \delta_V. 
    \end{align}
    Due to the separable structure, we expect that this circuit is precisely implemented compared with $W$, so we assume that $\delta_V$ is sufficiently small against $\varepsilon''$ and $ \delta'$.

    Next, we derive the upper bound of the distance between the ideal state and the approximating probe state constructed by $W$ and $V^{(q)}$.
   From the above discussion, we obtain  
  \begin{align}
    &\norm{( \1_{a''+\log_2 d + 2} \otimes  V^{(q)} )W \ket{\bm 0} \ket{\cos}^{\otimes M} - \sum_{\bm x \in G_p^M} c_{\bm x}^{(\cos) } e^{2\pi i 2^p \sum x_j 2^{q} \pi^{-1} \qty( \langle O_j \rangle-u_j ) } \ket{\bm 0 } \ket{\bm x}} \\
    &\leq  \norm{( \1_{a''+\log_2 d + 2} \otimes  V^{(q)}) \qty[
    W \ket{\bm 0} \ket{\cos}^{\otimes M} -  \sum_{\bm x \in G_p^M} c_{\bm x}^{(\cos) } e^{2\pi i 2^p \sum x_j 2^{q} \pi^{-1} \langle O_j \rangle  } \ket{\bm 0 } \ket{\bm x}
    ] }  \notag\\
    &+ \norm{( \1_{a''+\log_2 d + 2} \otimes  V^{(q)} )
    \sum_{\bm x \in G_p^M}c_{\bm x}^{(\cos) } e^{2\pi i 2^p \sum x_j 2^{q} \pi^{-1} \langle O_j \rangle  } \ket{\bm 0 } \ket{\bm x} - 
    \sum_{\bm x \in G_p^M} c_{\bm x}^{(\cos) } e^{2\pi i 2^p \sum x_j 2^{q} \pi^{-1} \qty( \langle O_j \rangle -u_j)  } \ket{\bm 0 } \ket{\bm x} 
    } \\
    &= \norm{ 
    W \ket{\bm 0} \ket{\cos}^{\otimes M} -  \sum_{\bm x \in G_p^M} c_{\bm x}^{(\cos) } e^{2\pi i 2^p \sum x_j 2^{q} \pi^{-1} \langle O_j \rangle  } \ket{\bm 0 } \ket{\bm x} } \notag\\
    &+  \norm{\sum_{\bm x \in G_p^M} c_{\bm x}^{(\cos) }e^{2\pi i 2^p \sum x_j 2^{q} \pi^{-1} \langle O_j \rangle  } (V^{(q)} - e^{-2\pi i 2^p \sum x_j 2^{q} \pi^{-1} u_j }) \ket{\bm x}} \\ 
    &\leq \varepsilon'' + \sqrt{2 \varepsilon''} +  \norm{\sum_{\bm x\in G_p^M}  c_{\bm x}^{(\cos) } e^{i \tilde{f}(\bm x)t} \ket{\bm x}- \sum_{\bm x \in G_p^M }  c_{\bm x}^{(\cos) } e^{2\pi i 2^p \sum x_j 2^{q} \pi^{-1} \langle O_j \rangle } \ket{\bm x}} + \delta_V \\
    &\leq \varepsilon'' + \sqrt{2 \varepsilon''} + \sqrt{5 \delta'}+ \delta_V  < \frac1{12}.
  \end{align}
  We use the triangle inequality in the first inequality and apply Eq.~\eqref{eq:upper_bound_state_block_encoding} in the second inequality. The third inequality follows from the following inequality,
  \begin{align}
    \norm{\sum_{x\in G_p^M} c_{\bm x}^{(\cos) }  \qty( e^{i \tilde{f}(\bm x)t} - e^{2\pi i 2^p \sum x_j 2^{q} \pi^{-1} \langle O_j \rangle } ) \ket{\bm x}}^2 
    &\leq \sum_{\bm x \in G_p^M} 
    \abs{c_{\bm x}^{(\cos) }}^2 \abs{e^{i \tilde f (\bm x) t} - e^{2\pi i 2^p \sum x_j 2^{q} \pi^{-1} \langle O_j \rangle } }^2 \\
    \label{eq:inequality_unvalid}
    &\leq 5\delta'.
    \end{align}
    Assuming the validity of Eq.~\eqref{eq:inequality_unvalid}, we conclude that the unitary \(\mathcal{U}^{(q)}_{\Upsilon, \text{sym}} = (\mathds{1}_{a'' + \log_2 d + 2} \otimes V^{(q)}) W (\mathds{1}_{a'' + \log_2 d + 2} \otimes (\chi_p^{(\cos)})^{\otimes M} ) \) generates the desired probe state, where \(\chi_p^{(\cos)} \) denotes a unitary circuit that generates $p$-th qubit cosine state.
    
    To complete our proof, it remains to verify that Eq.~\eqref{eq:inequality_unvalid} is satisfied. To this end, we consider the sequence of $M$ independent random variables $\bm X = (X_1, \ldots,X_M) $ where each $X_j$ is supported on $G_p$ and $\forall j, \mtr{Pr} [X_j = x_j ] = \abs{c_{x_j}^{(\cos) } }^2  $. Since these random variables are symmetric distributed, we can obtain,  
    \begin{align}
    &\sum_{\bm x \in G_p^M} 
    \abs{c_{\bm x}^{(\cos) }}^2 \abs{e^{i \tilde f (\bm x) t} - e^{2\pi i 2^p \sum x_j 2^{q} \pi^{-1} \langle O_j \rangle } }^2 \\
    &= \sum_{\bm x \in G_p^M} \mtr{Pr} [\bm X = \bm x] \abs{e^{i \tilde f (\bm x) t} - e^{2\pi i 2^p \sum x_j 2^{q} \pi^{-1} \langle O_j \rangle } }^2 \\
    &\leq \sum_{\bm x \in G_p^M/F_\Delta} \mtr{Pr} [\bm X = \bm x] \cdot 4 +
    \sum_{\bm x \in F_\Delta} \mtr{Pr} [\bm X = \bm x] \abs{e^{i \tilde f (\bm x) t} - e^{2\pi i 2^p \sum x_j 2^{q} \pi^{-1} \langle O_j \rangle } }^2
    \quad (\because \ \abs{e^{ia}-e^{ib}}^2 \leq 4) \\ 
    &\leq \mtr{Pr} [\bm X \not \in  F_\Delta]\cdot 4 +   \sum_{\bm x \in F_\Delta} \mtr{Pr} [\bm X = \bm x] \abs{e^{i \tilde f (\bm x) t} - e^{2\pi i 2^p \sum x_j 2^{q} \pi^{-1} \langle O_j \rangle }}^2  \\
    &\leq  4 \cdot \delta'+  \sum_{\bm x \in F_\Delta} \mtr{Pr} [\bm X = \bm x] \abs{\tilde{f} (\bm x)t - 2\pi 2^p \sum^M_{j=1} x_j \frac{2^{q} }{\pi} \langle O_j \rangle }^2 \quad (\because \ \abs{e^{ia}-e^{ib}} \leq \abs{a-b}) \\
    \label{eq:eval_eq}
    &= 4 \cdot  \delta' + t^2\sum_{\bm x \in F_\Delta} \mtr{Pr} [\bm X = \bm x] \abs{ 
    \ev**{H_\Delta^{(\bm x)} - \frac1{\sigma_\Delta}\sum_{j=1}^M x_j O_j }{\psi} }^2 \\
    &= 4 \cdot  \delta' + t^2 \sum_{\bm x \in F_\Delta} \mtr{Pr} [\bm X = \bm x] \abs{ 
    \ev**{ \Pi_\Delta \qty( H_\Delta^{(\bm x)} - \frac1{\sigma_\Delta}\sum_{j=1}^M x_j O_j ) \Pi_\Delta }{\psi} }^2  \quad (\because \ket{\psi} = \Pi_\Delta \ket{\psi} )\\
    &\leq 4 \cdot  \delta' +  \mtr{Pr} [\bm X \in F_{\Delta}] (t \varepsilon')^2 \quad (\because H_\Delta^{(\bm x} \text{ satisfies } Eq.~\eqref{eq:subspace_block_encoding}. ) \\
    &\leq4 \cdot  \delta' +  (t \varepsilon')^2 \leq 5 \delta' \;.
  \end{align}
  We use Eq.~\eqref{eq:def_tilde_f} and $t =  2^{p+q+1}\sigma_\Delta$ to derive Eq.~\eqref{eq:eval_eq}. The last equality comes from choosing $t = 2^{p+q+1}\sigma_\Delta \rightarrow \varepsilon' = \sqrt{\delta'}/t$. The fourth inequality comes from $\abs{\bra{\psi}U_1 - U_2 \ket{\psi}  } \leq \norm{U_1-U_2}$ for any unitary operations $U_1$ and $U_2$.
\end{proof}

Here, we finally get ready to evaluate the relation between the root MSE $\varepsilon$ and the total queries to the state preparation oracle when it comes to considering the specific subspace characterized by symmetry.
\\

\noindent
\textit{Proof of Theorem \ref{thm:evaluation_symmetry_QGE}.~}
        Our constructive proof of this theorem is based on Algorithm~\ref{alg:unified_framework} with the subroutine constructed by the state preparation method by Lemma~\ref{lemma:symmetry_state_prep_by_Hamiltonian_simulation}.
        Thanks to the general framework of the adaptive QGE algorithm shown in Algorithm \ref{alg:unified_framework}, we only need to derive the sufficient number of samples $R^{(q)}$ such that the assumption of Theorem~\ref{thm:performance_main_alg} holds. 
        
        We initially evaluate the sufficient number of samples $R^{(q)}$ of expectation values.
        Suppose that at the beginning of iteration \(q\), the gradient estimation condition
    \begin{align}
        \label{eq:cond-grad-est}
        \forall j, \quad \left| \langle O_j \rangle - \tilde{u}_j^{(q)} \right| \le 2^{-q}
    \end{align}
    is satisfied. Under this assumption, a single-shot measurement for the probe state generated by Lemma~\ref{lemma:symmetry_state_prep_by_Hamiltonian_simulation} produces an outcome \(\bm{k} = \{k_1,\ldots,k_M\} \in G_p^M\) with
    \[
    \Pr\!\Biggl[\Bigl| k_j - \frac{2^q\bigl(\langle O_j\rangle - \tilde{u}_j^{(q)}\bigr)}{\pi} \Bigr| > \frac{1}{2\pi}\Biggr] < \mu,
    \]
    where \(\mu = 0.011 + 1/12\) by fixing $p=3$. (Here, the value \(0.011\) is derived from Eq.~\eqref{eq:prob_sine_state} in Sec.~\ref{subsec:modify_probe_state}, and \(1/12\) accounts for the Euclidean distance error as indicated in Lemma~\ref{lemma:symmetry_state_prep_by_Hamiltonian_simulation}.) Next, we define the event
    \[
    A_j \coloneq \left\{ \Bigl| g_j^{(q)}(R^{(q)}) - \frac{2^q\bigl(\langle O_j\rangle - \tilde{u}_j^{(q)}\bigr)}{\pi} \Bigr| \le \frac{1}{2\pi} \right\},
    \]
    where \(g_j^{(q)}(R^{(q)})\) is the median of \(R^{(q)}\) independent outcomes for the \(j\)th observable measurement. 
    When the event \(\bigcap_{j} A_j\) holds, the updated estimate \(\tilde{u}_j^{(q+1)}\) satisfies \(|\langle O_j\rangle - \tilde{u}_j^{(q+1)}| \le 2^{-(q+1)}\). Thus, it suffices to choose \(R^{(q)}\) such that the event \(\bigcap_{j} A_j\) occurs with at least the probability $1-\delta^{(q)}$. 
    To this end, we choose \(R^{(q)}\) so that the tail probability
    \begin{align}
        \label{eq:tail_prob}
        \Phi\bigl(R^{(q)}\bigr) \coloneq \Pr\!\left[\left| g_j^{(q)}(R^{(q)}) - \frac{2^q(\langle O_j\rangle - \tilde{u}_j^{(q)})}{\pi} \right| > \frac{1}{2\pi} \right] \le \frac{\delta^{(q)}}{M}
    \end{align}
    holds for all \(j\). Then, by the union bound, $\Pr\!\left[\bigcap_{j} A_j\right] \ge 1 - \delta^{(q)} $ holds.
    Here, the tail probability satisfies \(\Phi(R) \le 1 - F(\lceil R/2 \rceil - 1; \mu, R)\), where \(F(k; \mu, n)\) denotes the cumulative distribution function of a binomial distribution with success probability \(\mu\) and \(n\) trials (see Sec.~\ref{appsub:eval_tail_prob} for details). By Hoeffding’s inequality, this is further bounded as
    \[
    1 - F\bigl(\lceil R/2 \rceil - 1; \mu, R\bigr) \le \exp\left[-2R\left(\frac{1}{2} - \mu\right)^2\right].
    \]
    Therefore, to satisfy Eq.~\eqref{eq:tail_prob}, it is sufficient to set $R^{(q)} = \frac{1}{2\left(\frac{1}{2} - \mu\right)^2} \log\left(M/{\delta^{(q)}}\right).$
    In particular, when \(\mu = 0.011+1/12 \le 3/8\), we can take \(R^{(q)} = 4 \log(M / \delta^{(q)})\).

    Finally, we derive the query complexity for the state preparation procedure in our algorithm.
    At each iteration \(q\), we prepare \(R^{(q)}\) copies of a quantum state approximating \(\ket{\Upsilon(p)}_{\cos}\).
    The implementation of \(\mathcal{U}_{\Upsilon}^{(\cos)}\) involves \(2Q^{(q)} = \mathcal{O}(2^{p+q+1} \sigma_\Delta)\) queries to \(U_\psi\) and its inverse. Substituting the expression for \(\sigma_\Delta\) from Eq.~\eqref{eq:def_sigma_Delta}, we obtain the query complexity per iteration as
    \begin{align}
        2^{p+q+1} \sigma_\Delta R^{(q)} 
        &= 2^q \log(1/\delta^{(q)}) \cdot
        \mathcal{O}\!\left( \max_{\lambda \in \Delta} \sqrt{ \left\| \sum_{j=1}^M \left(O_j^{(\lambda)}\right)^2 \right\| \log\!\left( \sum_{\lambda' \in \Delta} m_{\lambda'} \right)} \log M \right).
    \end{align}
    Thus, the factor \(\aleph\) is given by
    \begin{align}
        \aleph = \mathcal{O}\!\left( \max_{\lambda \in \Delta} \sqrt{ \left\| \sum_{j=1}^M \left(O_j^{(\lambda)}\right)^2 \right\| \log\!\left( \sum_{\lambda' \in \Delta} m_{\lambda'} \right)} \log M \right).
    \end{align}
    Hence, the total query complexity of the adaptive QGE algorithm with $\mac{U}_{\Upsilon}^{(\cos)}$ is calculated as Eq.~\eqref{eq:query_comp_of_methodI}. 
\hfill \hfill \qed

\subsection{QGE algorithm inspired by parallel scheme}
\label{subsec:parallel_QGE_alg}

In this section, we develop the state-preparation method enhanced by the parallel scheme and analyze its total query complexity. In the original adaptive QGE algorithm, the state preparation process is repeated \(R^{(q)}\) times at each $q$-th iteration, with each repetition requiring approximately \(\mathcal{O}(\sqrt{M})\) queries to simultaneously estimate \(M\) expectation values. Consequently, the total query complexity per iteration is \(\mathcal{O}(R^{(q)} \sqrt{M})\).
Observing that this procedure is equivalent to estimating overall \(MR^{(q)}\) expectation values—and given that the QGE algorithm offers a quadratic reduction in query complexity with respect to $M$—it becomes clear that further improvements are possible through the use of auxiliary qubits. Specifically, if we prepare a larger probe state that estimates all \(MR^{(q)}\) observables simultaneously, only \(\mathcal{O}(\sqrt{MR^{(q)}})\) calls to the target oracle are required, at the expense of additional qubit resources.
Motivated by this observation, we propose a modified QGE algorithm that leverages auxiliary qubits to accelerate the estimation process. We refer to this approach as the \textit{parallel QGE algorithm}, and our main result regarding this method is stated in the following theorem.
\begin{theorem}[Method by parallel scheme]
\label{thm:parallel_QGE_query}
Let $\varepsilon \in (0,1)$ be a target precision. For given $M$ observables $\{O_j\}_{j=1}^M$, and a state preparation oracle $U_\psi$ in $d$ dimension, there exists a quantum algorithm that outputs a sample from estimators $\{\hat{u}_j\}_{j=1}^M$ for $\{\langle O_j \rangle\}$ satisfying
\[
\max_{j=1,2,\ldots,M} \operatorname{MSE}[\hat{u}_j] \le \varepsilon^2, \tag{9}
\]
using
\begin{align}
    \label{eq:comp_parallel}
    \mac{O}\Bigl(   \sqrt{ M \log d \log M} \Bigr) \Big/ \varepsilon
\end{align}
queries to the state preparation oracles $U_\psi$ and $U_\psi^\dagger$ in total. Here, the mean squared error of an estimator $\hat{u}_j$ is defined as
\[
\operatorname{MSE}[\hat{u}_j] := \mathbb{E}\Bigl[\Bigl(\hat{u}_j - \langle O_j \rangle\Bigr)^2\Bigr].
\]
\end{theorem}
As anticipated from our rough reasoning, the parallel QGE algorithm achieves a quadratic reduction in the \(\log M\) factor. 

In the parallel scheme, the core idea is to coherently prepare multiple copies of the probe state, enabling simultaneous estimation within a single quantum circuit. Our ultimate goal based on this scheme is to prepare the following state,
\begin{align}
    \ket{\tilde \Upsilon(q)}_{\cos}
     \coloneq  \bigotimes^{R^{(q)}}_{r=1} \left(\sum_{\bm x \in G_p^{M  }}
     c_{\bm x}^{(\cos) }
    e^{2 \pi i 2^p \sum_{j=1}^M x_j 2^q \pi^{-1} ( \langle O_j \rangle - \tilde{u}_j^{(q) } ) } \ket{\bm x}\right). 
\end{align}
To this end, we consider applying an enlarged phase embedding circuit implementing the phase $e^{2\pi i 2^p \sum_j x_j  \frac{2^q}{\pi} \bigl(\langle O_j \rangle - u_j^{(q)}\bigr)} $
simultaneously to all \(R^{(q)}\) cosine states. In what follows, we describe the implementation of this enlarged phase embedding circuit, step by step.

Firstly, we need to construct $(a + \lceil \log_2(M R) \rceil + 1)$-block encoding of 
\begin{align}
    \tilde{H}^{(\bm x_R)} = (M R)^{-1} \sum^{R}_{r=1} \sum^{M}_{j=1}  x_{j,r} O_j
\end{align}
for any element $\bm x_R= (x_{1,1}, \ldots, x_{1,R}, x_{2,1}, \ldots, x_{M,R}) \in (G_p)^{(MR)} $. Based on the original method utilizing the LCU method, the corresponding SELECT operation denoted by $\tilde{U}_{\mtr{SEL}}$ is given by
\begin{align}
    \tilde{U}_{\mtr{SEL}} \coloneq \sum^{M}_{j=1}   \ketbra{j} \otimes \1_R \otimes B_j.
\end{align}
where $B_j$ is a block-encoding of $O_j$. In order to avoid an exponentially deep quantum circuit, we utilize a modified version of the PREPARE operation as described in Ref.~\cite{wada2024Heisenberg}. Specifically, we implement this operation by sequentially applying the $\ket{j}$-bitwise controlled gate (see Figs.~13 and 14 in Ref.~\cite{wada2024Heisenberg}) over the $x_{j,r}$-register, an additional single ancilla, and the $\lceil \log_2 (MR) \rceil$ qubits. This procedure yields a $\bm x_R$-controlled block encoding of $e^{i \mac{Q}^{(\bm x_R)}}$, where 
$\mac{Q}^{(\bm x_R)} \coloneq \sum^R_{r=1} \sum^M_{j=1} x_{j,r} \ketbra{j,r}.
$ Next, by applying an eigenvalue transformation corresponding to the logarithm of unitaries, we obtain an $\epsilon$-precise block encoding of $(2/\pi) \mac{Q}^{(\bm x_R)}$. This transformation requires $O(\log(1/\epsilon))$ calls for the block-encoding of $e^{i \mac{Q}^{(\bm x)}}$. Utilizing the SELECT operation and this block encoding, we then construct an $\epsilon$-precise block-encoding of 
\begin{align}
    \frac{2 MR}{\pi 2^{\lceil \log_2 M R \rceil }}  \frac1{MR}  \sum^{R}_{r=1} \sum^{M}_{j=1}  x_{j,r} O_j 
    = \qty(\bra{\bm 0} H^{\otimes \lceil \log_2(M R)\rceil} \otimes \1 ) 
    \tilde{U}_{\mtr{SEL}} \cdot \frac{2}\pi \mac{Q}^{(\bm x_R) } 
    \qty( H^{\otimes \lceil \log_2(M R) \rceil}\ket{\bm 0} \otimes \1 ) ).
\end{align}
This implementation incurs a cost of at most $\mac{O}(p M R \log(1/\varepsilon))$ operations and requires $\mac O\bigl( p M R  + a + \lceil \log_2(M R) \rceil + \log_2 d \bigr)$ qubits.

The amplification process of the block-encoding requires access to the oracle $U_{\mtr{SEL} }' = \sum_{\bm x \in G_{p}^M }\ketbra{\bm x} \otimes U^{(\bm x)}$, where $U^{(\bm x)} $ is a block-encoding of $M^{-1} \sum_{j=1}^M x_j O_j$ and disregards the internal structure of each $O_j$. Since we possess a $\bm x_R$-controlled block-encoding of $(M R)^{-1} \sum_r \sum_j x_{j,r} O_j $, so we can immediately state the following lemma.

\begin{lemma}
    \label{lem:coherent_amplification}
  Let $\delta'>0, \varepsilon' \in (0,1/2)$, let $ p $ be a positive integer. Suppose that for $d$-dimensional observables $\{O_j\}_{j=1}^M$ with $\norm{O_j} \leq 1$, we have access to the oracle 
  $U'_{\mtr{SEL}}= \sum_{\bm x_R \in (G_{p})^{(MR)}} \ketbra{\bm x_R} \otimes U^{(\bm x)}$ where $U^{(\bm x)}$ is an $(a + \lceil \log_2 M \rceil)$-block encoding of observable
  $\tilde H^{(\bm x)} \coloneq (M R)^{-1} \sum_r \sum_j x_{j,r} O_j $ for some $a \in \mathbb{N}$.
  
  For independent symmetric distributed random variables $\bm X \coloneq \{X_{j,r} \}$ for $j \in \{1,\ldots,M\}, \ r \in \{1,\ldots,R\}$ supported on $(G_{p})^{MR} $ with variance $v \coloneq \mathbb{E}[X_{j,r}^2]$, we assume that the following inequality holds,
    \begin{align}
        \label{eq:def_tilde_sigma2}
        \tilde \sigma \coloneq \sqrt{2  R \norm{\sum_{j=1}^M  O_j^2 } \log(2d/\delta')} < MR,
    \end{align}
    or
    \begin{align}
    \label{eq:def_tilde_sigma}
      \tilde{\sigma}  \coloneq \sqrt{2 v R \norm{\sum_{j=1}^M  O_j^2 } \log(2d/\delta')} + \frac43 \log(2d /\delta')  < MR.
  \end{align}
  Then, there exists a subset $\tilde F \subset (G_{p})^{MR}$ 
  satisfying 
  \begin{align}
      \mtr{Pr}_{
      \bm X} [ \bm X \in \tilde F ] \geq 1 - \delta'.
  \end{align}
  and we can implement 
  a unitary (with $(pM R + \lceil \log_2 M  \rceil + \lceil \log_2 R \rceil + \log_2 d + a + 1)$ qubits in total),
  \begin{align}
    U_{\mtr{obs}} \coloneq \sum_{\bm x_R \in (G_{p})^{MR}} \ketbra{\bm{x}_R} \otimes U_{\mtr{obs}}^{(\bm{x}_R)}\;,
  \end{align}
  such that $U^{(\bm x)}_{\mtr{obs}}$ is a $(1, a + \lceil \log_2 M \rceil+ \lceil \log_2 R \rceil + 1, \varepsilon')$-block-encoding of the Hamiltonian $H^{(\bm{x}_R)} $ satisfying
  \begin{align}
  \label{eq:subspace_block_encoding2}
    \frac1{\tilde{\sigma}  } \sum_{r=1}^R \sum_{j=1}^M x_{j,r} O_j   \;, \ \text{if } \bm x_R \in \tilde F \subset (G_p)^{RM}\;,
  \end{align}
  For $m = \mac{O}(M \tilde{\sigma}^{-1} \log(M  \tilde{\sigma}^{-1}/\varepsilon')) $, this implementation of $U_{\mtr{obs}}$ requires $m$ queries to $U'_{\mtr{SEL}}$ or its inverse, $2m$ X gates controlled by $a + \lceil \log_2 M \rceil $ qubits, $\mac{O}(m)$ single-qubit gates, $\mac{O}(\mtr{poly}(m))$ classical precomputation.
\end{lemma}

\begin{Remark}
    In this theorem, we present two alternative choices for \(\tilde \sigma\) as specified in Eqs.~\eqref{eq:def_tilde_sigma2} and \eqref{eq:def_tilde_sigma}. Specifically, Eq.~\eqref{eq:def_tilde_sigma} yields a numerically smaller value, which is advantageous for precise numerical analysis in Sec.~\ref{sec:k-RDM-with-QGE}, while Eq.~\eqref{eq:def_tilde_sigma2} offers a simpler expression that facilitates asymptotic evaluations.
\end{Remark}

\noindent
\textit{Sketch of Proof.~}
We only need to follow the proof of Lemma \ref{lem:Amplitude_Amp_subspace} without imposing a specific symmetry to target observables. In this case, we consider the random variables $\{X_{j,r}\}_{j,r} $ that are identical and symmetrically distributed on $G_p^{(RM)}$. 
By replacing $M$ into $MR$ and $\norm{\sum_{j=1}^M O_j^2}$ into $\norm{\sum_{r=1}^M \sum_{j=1}^M O_j^2} = R \norm{\sum_{j=1}^M O_j^2}$ in $\sigma_\Delta$ of Eqs.~\eqref{eq:def_sigma_Delta2} or Eqs.~\eqref{eq:def_sigma_Delta} without imposing a specific symmetry, we can evaluate the upper bound of $\tilde \sigma $ satisfying the following matrix inequality,
\begin{align}
    \mtr{Pr} \qty[ \frac1M \norm{\sum^R_{r=1} \sum^M_{j=1} X_{j,r} O_j} \geq \frac{\tilde \sigma}{2MR}] \leq \delta'.
\end{align}
Such $\sigma$ is evaluated as Eq.~\eqref{eq:def_tilde_sigma} or Eq.~\eqref{eq:def_tilde_sigma2}, which results in the output block-encoding of the Hamiltonian $H^{(\bm x)}$ satisfying Eq.~\eqref{eq:subspace_block_encoding2} and the gate complexity.
\hfill \hfill \qed 

Similarly, we can derive an analogous version of Lemma \ref{lemma:symmetry_state_prep_by_Hamiltonian_simulation} for the parallel QGE algorithm. 
Note that the difference of selecting $\delta' $ occurs compared with the original lemma to make $\mac E$ Euclidean distance error approximate state.

\begin{lemma}[parallel-scheme state preparation by Hamiltonian simulation]
  \label{lemma:parallel_prep_by_Hamiltonian_simulation}
  Let $p,R$ be a positive integer, $\{O_j\}_{j=1}^M$ be $M$ $d$-dimensional observables with $\norm{O_j} \leq 1$. Suppose that $U_\psi$ is $N$-qubit state preparation oracle, 
  $\langle O_j \rangle \coloneq \bra{\psi} O_j \ket{\psi}$, and $u_j \in [-1,1]$. 
  Additionally, we have access to $U_{\mtr{SEL}}'= \sum_{\bm x_R \in (G_p)^{MR} }\ketbra{\bm x_R }\otimes U^{(\bm x)}_{\mtr{obs}} $ 
  where $U_{\mtr{obs}}^{(\bm x)}$ is an $(a+\lceil \log_2 MR \rceil)$-block-encoding of $H^{(\bm x)} = (M R)^{-1} \sum_r \sum_j x_{j,r} O_j$ for some $a \in \mathbb{N}$.
  We assume $\tilde \sigma$ is provided as
  \begin{align}
      \tilde{\sigma}  \coloneq   \sqrt{2 v R \norm{\sum_{j=1}^M  O_j^2 } \log(2 d /\delta')} + \frac43 \log(2 d /\delta')  < {MR},
  \end{align}
  and holds for $\delta ' =\mac{E}^2/20 $ where $\mac{E}$ satisfies $0 < \mac{E} <1 $  and $v \coloneq \mathbb{E}[(2X)^2]$ where \(X\) is a random variable on \(G_p\) with \(\mathrm{Pr}[X = x] = |c_{x}^{(\cos)}|^2\).

  Then for any non-negative integer $q$, there exists a state-preparation subroutine  $\mac{U}_{\Upsilon,\text{para}}$ preparing $pM R$-qubit state expressed by 
  \begin{align}
    \ket{\tilde{\Upsilon}(q) }_{\cos} = \bigotimes_{r=1}^{R} \left[\sum_{\bm x \in G_p^M} c^{(\cos)}_{\bm x} \exp(2 \pi i 2^p \sum^M_{j=1} x_j \cdot \frac{2^{q}}{\pi} \qty( \langle O_j \rangle -u_j   ) )\ket{\bm x}\right]
  \end{align}
  up to $\mac{E} $ Euclidean distance error,  with $2Q$ uses of \(U_\psi\) or \(U_\psi^\dagger\), \(4mQ\) uses of controlled \(U'_{\mathrm{SEL}}\) (or its inverse), \(\mac{O}(mQ)\) uses of NOT gates controlled by at most \(\mac{O}\Bigl(a + \lceil \log_2 M \rceil +\lceil \log_2 R \rceil + \log_2 d\Bigr)\) qubits, \(\mac{O}(mQ + pMR)\) uses of single-qubit or two-qubit gates, an additional \( (\lceil \log_2 M \rceil +\lceil \log_2 R \rceil+ \log_2 d + a + 3) \) ancilla qubits, and \(\mac{O}\bigl(\mathrm{poly}(Q) + \mathrm{poly}(m)\bigr)\) classical precomputation for determining the quantum circuit parameters, where 
    \[
    Q = \mac{O}(2^{p+q+1} \tilde{\sigma} + \log(1/\mac{E})) \quad \text{and} \quad m = \mac{O}\Bigl( \frac{MR}{\tilde{\sigma}}(p+q+\log M + \log R)\Bigr).
    \]
\end{lemma}
\begin{proof}
     Following the proof of Lemma \ref{lemma:symmetry_state_prep_by_Hamiltonian_simulation},we consider two quantum circuits, \(\tilde W\) and \(\tilde V^{(q)} \), which produce the state \(\ket{\tilde \Upsilon(q)}_{\cos}\). The circuit \(\tilde W\), built using uniform amplification (Lemma~\ref{lem:coherent_amplification}) and Hamiltonian simulation, encodes the $R$ copies of $M$ expectation values \(\langle O_j \rangle\), while \(\tilde V^{(q)} \) incorporates parameters \(u_j\) for each \(j \in \{1,\ldots,M\}\) to corresponding qubits on the probe system. We first analyze the approximation error in \(\tilde{W}\) from Hamiltonian simulation, then show that the discrepancy in \(V\) remains small. Finally, we  derive an upper bound on the distance between the ideal probe state and the approximate probe state \(\ket{\tilde \Upsilon (q)}\).

    Let $\varepsilon'' = \mac{E}^2/ 2^{6} \in (0,1)$ and $\varepsilon' = \sqrt{\delta'}/(2^{p+q+2 }\tilde{\sigma}) \in (0,1/2)$. 
    we begin with the analysis of $\tilde W$ regarding with the Hamiltonian simulation.
  From the assumption and Lemma \ref{lem:coherent_amplification}, we have a unitary 
  \begin{align}
    U_{\mtr{obs}} \coloneq \sum_{\bm{x}_R \in G_p^{MR}} \ketbra{\bm{x}_R} \otimes U_{\mtr{obs}}^{(\bm{x}_R)}\;,
  \end{align}
  where $U^{(\bm x)}_{\mtr{obs}}$ is a $(1, a'' = a+\lceil \log_2 M \rceil + \lceil \log_2 R \rceil  +1, \varepsilon')$-block-encoding of $\frac1{\tilde{\sigma}  } \sum_{r=1}^R \sum_{j=1}^M x_{j,r} O_j   \;, \ \text{if } \bm x_R \in \tilde F \subset (G_p)^{RM}$.
  Thus, using state-preparation oracle and this unitary, we can implement a $(1, a'' + \log_2 d, 0)$-block-encoding of the following Hamiltoinan 
  \begin{align}
    \label{eq:target_hamiltonian}
    \sum_{\bm{x}_R \in G_p^(MR)} \tilde{f}(\bm{x}_R) \ketbra{\bm{x}_R}\;, \text{where } \tilde{f}(\bm x_R) &\coloneq \bra{0}^{\otimes a''+\log_2 d} U_\psi^\dagger U^{(\bm{x}_R)}_{\mtr{obs},} U_\psi \ket{0}^{\otimes a'' +\log_2 d}\;.
  \end{align}
  For this encoding, the Hamiltonian simulation circuit described in Fig.~\ref{fig:HS_state_prep_circ} yields the quantum circuit $W$ for a $(1, a'' + \log_2 d + 2, \varepsilon'')$-block-encoding of time evolution operator 
  \begin{align}
    \label{eq:def_tilde_f}
      \sum_{\bm x \in G_p^M} e^{i \tilde f(\bm x )t} \ketbra{\bm x}, \quad t := 2^{p+q+1} {\tilde{\sigma}}.
  \end{align}

  Additionally, we consider the following circuit $\tilde V^{(q)}$ such that 
    \begin{align}
        \norm{\tilde{V}^{(q)}  - \bigotimes_{r=1}^R \sum_{\bm x \in G_p^M } e^{- 2 \pi 2^p \sum_j x_j 2^q \pi^{-1} u_j} \ketbra{\bm x} } \leq \delta_V. 
    \end{align}
    Due to the separable structure, we expect that this circuit is precisely implemented compared with $\tilde W$, so we assume that $\delta_V$ is sufficiently small against $\varepsilon''$ and $ \delta'$.

    Next, we derive the upper bound of the distance between the ideal state and the approximating probe state constructed by $\tilde W$ and $\tilde V^{(q)}$. Here, we denote $c_{\bm x_R}^{(\cos)} \coloneq \prod_{r=1}^R  \prod_{j=1}^M c_{x_{j,r} }^{(\cos)} $. From the above discussion, we obtain  
  
    \begin{align}
    &\norm{( \1_{a''+\log_2 d + 2} \otimes  \tilde V^{(q)} ) \tilde W \ket{\bm 0} \ket{\cos}^{\otimes (MR)} - \bigotimes_{r=1}^R \sum_{\bm x \in G_p^M} c_{\bm x}^{(\cos) } e^{2\pi i 2^p \sum_j x_j 2^{q} \pi^{-1} \qty( \langle O_j \rangle-u_j ) } \ket{\bm 0 } \ket{\bm x}} \\
    &\leq  \norm{( \1_{a''+\log_2 d + 2} \otimes  \tilde V^{(q)}) \qty[
    \tilde W \ket{\bm 0} \ket{\cos}^{\otimes (MR)} -  \bigotimes_{r=1}^R \sum_{\bm x \in G_p^M} c_{\bm x}^{(\cos) } e^{2\pi i 2^p \sum_j x_j 2^{q} \pi^{-1} \langle O_j \rangle  } \ket{\bm 0 } \ket{\bm x}
    ] }  \notag\\
    &+ \norm{( \1_{a''+\log_2 d + 2} \otimes  \tilde V^{(q)} )\bigotimes_{r=1}^R  \qty[
    \sum_{\bm x \in G_p^M}c_{\bm x}^{(\cos) } e^{2\pi i 2^p \sum_j x_j 2^{q} \pi^{-1} \langle O_j \rangle  } \ket{\bm 0 } \ket{\bm x} - 
    \sum_{\bm x \in G_p^{M}} c_{\bm x}^{(\cos) } e^{2\pi i 2^p \sum_j x_j 2^{q} \pi^{-1} \qty( \langle O_j \rangle -u_j)  } \ket{\bm 0 } \ket{\bm x} ]
    } \\
    &= \norm{ 
    \tilde W \ket{\bm 0} \ket{\cos}^{\otimes (MR)} -  \bigotimes_{r=1}^R \sum_{\bm x \in G_p^M} c_{\bm x}^{(\cos) } e^{2\pi i 2^p \sum_j x_j 2^{q} \pi^{-1} \langle O_j \rangle  } \ket{\bm 0 } \ket{\bm x} } \notag\\
    &+  \norm{\sum_{\bm x_R \in G_p^{(MR)} } c_{\bm x_R}^{(\cos) }e^{2\pi i 2^p \sum_j x_{j,r} 2^{q} \pi^{-1} \langle O_j \rangle  } (\tilde V^{(q)} - e^{-2\pi i 2^p \sum_j x_{j,r} 2^{q} \pi^{-1} u_j }) \ket{\bm x_R}} \\ 
    &\leq \varepsilon'' + \sqrt{2 \varepsilon''} +  \norm{\sum_{\bm x\in G_p^{(MR)}}  c_{\bm x_R}^{(\cos) } e^{i \tilde{f}(\bm x_R)t} \ket{\bm x_R}- \sum_{\bm x_R \in G_p^{(MR)} }  c_{\bm x_R}^{(\cos) } e^{2\pi i 2^p \sum x_j 2^{q} \pi^{-1} \langle O_j \rangle } \ket{\bm x_R}} + \delta_V \\
    &\leq \varepsilon'' + \sqrt{2 \varepsilon''} + \sqrt{5 \delta'} + \delta_V\\
        &= \mac{E}^2/2^6 +  \mac{E}/2\sqrt{2} + \mac{E}/2 + \delta_V \leq \mac{E}  \;.
    \end{align}
    In the last inequality, we assume that \(\delta_V\) is relatively small compared to \(\varepsilon''\) and \(\delta'\). The third inequality follows from the following bound:
    \begin{align}
        \norm{\sum_{\bm{x} \in G_p^{(MR)}} c_{\bm{x}_R}^{(\cos)} e^{i \tilde{f}(\bm{x}_R) t} \ket{\bm{x}_R} - \sum_{\bm{x}_R \in G_p^{(MR)}} c_{\bm{x}_R}^{(\cos)} e^{2\pi i 2^p \sum x_j 2^q \pi^{-1} \langle O_j \rangle} \ket{\bm{x}_R}} \leq \sqrt{5 \delta'}.
    \end{align}
    This inequality can be derived using the same procedure as in the last part of the proof of Lemma~\ref{lemma:symmetry_state_prep_by_Hamiltonian_simulation}. Thus, we define $\mac{U}_{\text{para}}^{(q)} \coloneq ( \1_{a''+\log_2 d + 2} \otimes  \tilde V^{(q)} ) \tilde W (\1_{a''+\log_2 d + 2} \otimes (\chi_p^{(\cos)})^{\otimes (MR)} )$ as a state preparation subroutine that generates the target probe state.
\end{proof}

From this discussion, we can construct the quantum circuit implementation to generate a probe state approximating $\ket{\tilde{\Upsilon}(q)}_{\cos}$. By using the adaptive QGE algorithm with this circuit, we can complete the proof of Theorem \ref{thm:parallel_QGE_query}.

\ 
\\

\noindent
\textit{Proof of Theorem \ref{thm:parallel_QGE_query}. }
By using the general framework of the adaptive QGE algorithm presented in Algorithm~\ref{alg:unified_framework}, it suffices to determine the required number of samples \(R^{(q)}\), and to evaluate the total query complexity based on the computational cost of \(\mac{U}_{\tilde \Upsilon}^{(\cos)}\) given in Lemma~\ref{lemma:parallel_prep_by_Hamiltonian_simulation}.

We begin by determining the sufficient number of samples \(R^{(q)}\) and the Euclidean distance \(\mathcal{E}\) between the ideal probe state and the generated probe state, in order to bound the additive error of the temporal estimates $\{ \tilde u_j^{(q)} \}_{j=1}^M$. 
To this end, we analyze the measurement outcomes obtained from the ideal state \((\mathrm{QFT}^\dagger_{G_p})^{\otimes M R^{(q)}} \ket{\tilde{\Upsilon}(q)}_{\cos}\), and derive an appropriate choice of \(R^{(q)}\). Based on this analysis, we then select \(\mathcal{E}\) such that the intermediate estimates \(\tilde{u}_j^{(q+1)}\), obtained by measuring the state prepared by \(\mathcal{U}_{\tilde{\Upsilon}}^{(\cos)}\), approximate \(\langle O_j \rangle\) with success probability at least \(1 - \delta^{(q)}\).

For $(\mtr{QFT}^\dagger_{G_p})^{(M R^{(q)})}\ket{\tilde \Upsilon(q) }_{\cos} $ with $p=3$, measurement results $\bm k_R^{(q)} \coloneq \{k_{1,1}, k_{1,2}, \ldots, k_{1,R^{(q)}}, k_{2,1}, \ldots, k_{M,R^{(q)}} \}$ satisfy,
\begin{align}
    \forall j \in \{1,\ldots, M\}, r \in \{1,\ldots, R^{(q)}\}, 
        \mtr{Pr} \qty[ \abs{k_{j,r} - \frac{2^q(\langle O_j\rangle - \tilde u_j^{(q)})}\pi} > \frac1{2 \pi} ] <0.011,
\end{align}
under the assumption of 
\begin{align}
        \forall j, \quad \left| \langle O_j \rangle - \tilde{u}_j^{(q)} \right| \le 2^{-q}.
\end{align}
Here, we define the $j$-th median of the samples $k_{j,1}, \ldots, k_{j, R^{(q)}}$ as $\tilde{g}^{(q)}_j$ and consider the event
\begin{align}
    B_j^{(q)} \coloneq \left\{ \abs{\tilde{g}_j^{(q)}(R^{(q)}) -\frac{2^q(\langle O_j\rangle - \tilde u_j^{(q)})}\pi } \leq \frac{1}{2\pi} \right\}. 
\end{align}
Since each measurement result is obtained independently, following the argument presented in Thm.~\ref{thm:evaluation_symmetry_QGE}, if we choose \(R^{(q)}\) to bound the tail probability 
\begin{align}
    \label{eq:tail_prob-2}
    \tilde \Phi\bigl(R^{(q)}\bigr) \coloneq \Pr\!\Bigl[\Bigl| \tilde{g}_j^{(q)} - \frac{2^q\bigl(\langle O_j\rangle - \tilde{u}_j^{(q)}\bigr)}{\pi} \Bigr| > \frac{1}{2\pi}\Bigr] \le \frac{\delta^{(q)}}{2M},
\end{align}
then by the union bound,
\[
\Pr\!\Bigl[\bigcap_{j} B_j^{(q)}\Bigr] \ge 1-\delta^{(q)}/2.
\]
An upper bound of $R^{(q)} $ satisfying Eq.\eqref{eq:tail_prob-2} is evaluated as $3 \log(2M/\delta^{(q)} ) $ by using Hoeffding inequality (See the detail in Sec.~\ref{appsub:eval_tail_prob}).

Since we actually work with a probe state that approximates \(\ket{\tilde{\Upsilon}(p)}_{\cos}\) within Euclidean distance \(\mathcal{E}\), as shown in Lemma~\ref{lemma:parallel_prep_by_Hamiltonian_simulation}, the probability of the event \(\bigcap_j B_j^{(q)}\) for the probe state prepared by \(\mathcal{U}_{\tilde{\Upsilon}}^{(\cos)}\) changes by at most \(\mathcal{E}\) \cite{apeldoorn2023quantum}. 
Therefore, by choosing \(\mathcal{E} = \delta^{(q)}/2\), we ensure that the probability of \(\bigcap_j B_j^{(q)}\) is at least \(1 - \delta^{(q)}\).

    Finally we consider the total query complexity. If a single probe state is prepared via the technique in Lemma~\ref{lemma:parallel_prep_by_Hamiltonian_simulation}, we need to implement $\mac{E}^2/2^6 (= \delta^{(q)}/2^8)$-precise Hamiltonian simulation circuit. Then, this implementation requires \(2Q^{(q)} = \mac{O}(2^{p+q+1}\tilde \sigma^{(q)} + \log(\mac E^{-1} ) ) = \mac{O}(2^{p+q+1}\tilde \sigma^{(q)}  + \log(1/\delta^{(q)}  ) )  \) queries are needed for \(U_\psi\) and its inverse, where $\tilde \sigma^{(q)} $ is defined by choosing $R = R^{(q)}$ and $\delta' = \mac{E}^2/20 = (\delta^{(q)})^2/80 $ in $\tilde \sigma$. Note that $\tilde \sigma^{(q)} = \mac{O} \qty(\sqrt{R^{(q)}M \log (2d/\delta^{(q) } )} )$, so the dominant scaling of $2Q^{(q)} $ is $\mac O (2^{p+q+1}\tilde \sigma^{(q)} )$.
    Here the evaluation of $\tilde \sigma$ is expressed as follows 
    \begin{align}
        \tilde \sigma^{(q)} &= \sqrt{2 R^{(q)} M \log(160d/(\delta^{(q)})^2 )} \\
        &\leq \sqrt{6M \log(160d/(\delta^{(q)})^2 ) \log(2M/\delta^{(q)}) } \\
        &\leq \mac{O}\qty(\sqrt{M \log d\log M} ) +\mac{O}  \qty( \sqrt{M} (\sqrt{\log d} + \sqrt{\log M}) ) \log(1/\delta^{(q)} )  .
    \end{align}
    The first inequality comes from $R^{(q)} \leq 3 \log(2M/\delta^{(q)} ) $.
    Hence, the query complexity of \(\mathcal{U}_{\tilde{\Upsilon}}^{(\cos)}\) is evaluated as \(\mathcal{O}\left(2^q \log (1/\delta^{(q)}) \sqrt{M \log d \log M}\right)\), and the corresponding factor \(\aleph\) becomes \(\mathcal{O}\left(\sqrt{M \log d \log M}\right)\). 
    Accordingly, the total query complexity of the adaptive QGE algorithm with \(\mathcal{U}_{\Upsilon}^{(\cos)}\) is given by Eq.~\eqref{eq:comp_parallel}.
\hfill \hfill \qed

\subsection{Method II : Integrating all}
\label{subsec:Integral-QGE}

In the Sec.~\ref{subsec:Symmetry-QGE-algorithm} and \ref{subsec:parallel_QGE_alg}, we have implemented Method I and the QGE algorithm inspired by the parallel scheme, respectively. In this section, we propose Method II by integrating these methods and concisely state the total query complexity.

\begin{theorem}[Method II]
\label{thm:evaluation_symmetry_parallel_QGE}
Let $\varepsilon \in (0,1)$ be a target precision and $\mac{G}$ be a group. For given $M$ $\mac{G}$-symmetric observables $\{O_j\}_{j=1}^M$ of bounded spectral norm $\norm{O_j}\leq 1$, assume that each $O_j$ has the irreducible decomposition $
O_j = \bigoplus_{\lambda \in \Lambda} \1(\mathbb{C}^{d_\lambda}) \otimes O^{(\lambda)}_j,
$
where $\lambda$ labels the irreducible subspaces and each $O^{(\lambda)}_j$ is an $m_\lambda$-dimensional operator.
For a target state $\ket{\psi}$ prepared by a given oracle $U_{\psi}$, we further assume that there exists an index subset $\Delta\subseteq\Lambda$ such that the irreducible subspaces specified by $\lambda\in \Delta$ contain the target state $\ket{\psi}$.
Then, there exists a quantum algorithm that outputs a sample from estimators $\{\hat{u}_j\}_{j=1}^M$ for the expectation values $\{\langle \psi|O_j|\psi \rangle\}$ satisfying
\begin{align}
\max_{j=1,2,\ldots,M} \operatorname{MSE}[\hat{u}_j] \le \varepsilon^2, \tag{9}
\end{align}
using
\begin{align}
    \label{eq:total_query_unified}
    \varepsilon^{-1} \cdot \mac{O}\qty(  \log^{1/2}(M) \log^{1/2}(m_\Delta)  \max_{\lambda \in \Delta}   \norm{ \sum_{j=1}^M [O_j^{(\lambda)}]}^{1/2} ).
\end{align}
queries to the state preparation oracles $U_\psi$ and $U_\psi^\dagger$ in total, where $m_{\Delta}$ is defined as $m_{\Delta}:=\sum_{\lambda'\in \Delta}m_{\lambda'}$.
Here, the mean squared error of an estimator $\hat{u}_j$ is defined as
\[
\operatorname{MSE}[\hat{u}_j] := \mathbb{E}\Bigl[\Bigl(\hat{u}_j - \langle \psi|O_j |\psi\rangle\Bigr)^2\Bigr].
\]
\end{theorem}
\noindent
\textit{Sketch of Proof. }
By integrating Lemmas \ref{lem:construction_sigma_delta} and \ref{lem:coherent_amplification}, we can implement a following unitary circuit for any positve integers $R$,
\begin{align}
    \tilde U_{\mtr{obs}, \Delta} \coloneq \sum_{\bm x_R \in G_p^{(RM)}} \ketbra{\bm x_R } \otimes \tilde U_{\mtr{obs}, \Delta}^{(\bm x_R) },
\end{align}
where $\tilde U_{\mtr{obs}, \Delta}^{(\bm x_R) }$ is a $(a+ \lceil \log_2 M \rceil + \lceil \log_2 R \rceil +1$-block encoding of the Hamiltonian $\tilde H_\Delta^{(\bm x_R)}$ satisfying 
\begin{align}
  \label{eq:subspace_block_encoding_para}
    \norm{ \Pi_\Delta \qty(\tilde H^{(\bm x)}_\Delta  - \frac1{\tilde \sigma_\Delta  } \sum_{r=1}^R \sum_{j=1}^M x_{j,r} O_j ) \Pi_\Delta}\leq \varepsilon'  \;, \ \text{if } \bm x_R \in \tilde F_{\Delta} \subset {G_p^{(MR)}}\;,
\end{align}
if $\tilde \sigma_\Delta$ satisfies 
\begin{align}
        \label{eq:def_tilde_sigma_Delta2}
        \tilde \sigma_\Delta \coloneq \max_{\lambda \in \Delta } \sqrt{2 R \norm{\sum_{j=1}^M  [O_j^{(\lambda)}]^2 } \log(2 m_\Delta/\delta')} < MR,
\end{align}
By using this block-encoding, we can implement a probe state preparation subroutine $\mac{U}_{\text{unified}}$ that generates $\ket{\tilde \Upsilon(q)}_{\cos}$ up to $\mac{E}$ Euclidiean distance error where $\mac{E} \in (0,1)$, with $2Q$ uses of $U_\psi$ and its inverse. Here, $Q = \mac{O}(2^{p+q+1} \tilde \sigma_\Delta  + \log(1/\mac{E}) )$.

To calculate the total query complexity when using \(\mathcal{U}_{\text{unified}}\) in the general framework of the adaptive QGE algorithm presented in Algorithm \ref{alg:unified_framework}, it suffices to determine the required number of samples \(R^{(q)}\) and a parameter \(\mathcal{E}\). However, this choice is independent of the symmetry condition, so we can choose \(R^{(q)}\) and \(\mathcal{E}\) as in the proof of Theorem \ref{thm:parallel_QGE_query}. Additionally, the evaluation of the query complexity of \(\mathcal{U}_{\text{unified}}^{(q)}\) follows the same procedure. The evaluation of \(\tilde{\sigma}^{(q)}_\Delta\) expressed as 
\begin{align}
    \tilde \sigma_\Delta ^{(q)}\coloneq \max_{\lambda \in \Delta } \sqrt{2 R^{(q)} \norm{\sum_{j=1}^M  [O_j^{(\lambda)}]^2 } \log(2 m_\Delta/\delta')}
\end{align}
is given by
\begin{align}
    \tilde{\sigma}^{(q)}_\Delta &= \max_{\lambda \in \Delta} \sqrt{2 R^{(q)} \left\| \sum_{j=1}^M [O_j^{(\lambda)}]^2 \right\| \log\left( \frac{160 m_\Delta}{(\delta^{(q)})^2} \right)} \\
    &\leq \max_{\lambda \in \Delta} \sqrt{6 \left\| \sum_{j=1}^M [O_j^{(\lambda)}]^2 \right\| \log\left( \frac{160 m_\Delta}{(\delta^{(q)})^2} \right) \log\left( \frac{2M}{\delta^{(q)}} \right)} \\
    &\leq \mathcal{O}\left( \max_{\lambda \in \Delta} \sqrt{\left\| \sum_{j=1}^M [O_j^{(\lambda)}]^2 \right\| \log m_\Delta \log M} \right) 
    + \mathcal{O}\left( \max_{\lambda \in \Delta} \sqrt{\left\| \sum_{j=1}^M [O_j^{(\lambda)}]^2 \right\|} (\sqrt{\log m_\Delta} + \sqrt{\log M}) \log\left( \frac{1}{\delta^{(q)}} \right) \right).
\end{align}
The first inequality comes from \(R^{(q)} \leq 3 \log\left( \frac{2M}{\delta^{(q)}} \right)\). Hence, the query complexity of \(\mathcal{U}_{\text{unified}}^{(q)}\) is evaluated as
\(\mathcal{O}\qty(2^q \log (1/\delta^{(q)})  \log^{1/2}(M) \log^{1/2}(m_\Delta)  \max_{\lambda \in \Delta}   \norm{ \sum_{j=1}^M [O_j^{(\lambda)}]^2}^{1/2} ) \) 
and the corresponding factor \(\aleph\) becomes
\(\mathcal{O}\left( \log^{1/2}(M) \log^{1/2}(m_\Delta)  \max_{\lambda \in \Delta}   \norm{ \sum_{j=1}^M [O_j^{(\lambda)}]^2}^{1/2}\right)\).
Accordingly, the total query complexity of the adaptive QGE algorithm with \(\mathcal{U}_{\text{unified}}\) is given by Eq.~\eqref{eq:total_query_unified}.
\hfill \hfill \qed


\section{Applications}
\label{sec:Application}
One major area where quantum computers offer a significant speedup is the study of electron correlation effects in quantum systems. To investigate these effects, the $k$-body reduced density matrix ($k$-RDM) serves as an essential tool for capturing the physical properties of quantum many-body systems. The $k$-RDM of a pure state $\ket{\chi}$ supported on $N$ fermionic modes is represented as a tensor specified by $N^{2k}$ matrix elements, given by
\[
\bra{\psi} a_{p_1}^\dagger \cdots a_{p_k}^\dagger a_{q_1} \cdots a_{q_k} \ket{\psi} =: \bra{\psi} A^{\bm p}_{\bm q} \ket{\psi} ,
\]
where the indices $p_j, q_j$ run over the $N$ fermionic modes.

In this section, we investigate the efficiency of the QGE algorithm for measuring $k$-RDM elements on a practical-size system. For simplicity, we set that the ultimate goal of estimating of $k$-RDM elements is to efficiently obtain each estimators $\hat u^{\bm p}_{\bm q}, \hat v^{\bm p}_{\bm q} $ the following values for given a pure state $\ket{\psi}$ and precision $\varepsilon \in (0,1)$:
\begin{align}
    \label{eq:cond1}
    \forall \bm p \neq \bm q, \mathbb{E} \qty[ \abs{\hat u^{\bm p}_{\bm q} - \bra{\psi} A^{\bm p}_{\bm q} + A^{\bm q}_{\bm p}  \ket{\psi} }^2 ] \leq \varepsilon^2,  \\
    \label{eq:cond2}
    \forall \bm p \neq \bm q, \mathbb{E} \qty[ \abs{\hat v^{\bm p}_{\bm q} - \bra{\psi} i(A^{\bm p}_{\bm q} - A^{\bm q}_{\bm p} ) \ket{\psi} }^2 ] \leq \varepsilon^2, \\
    \label{eq:cond3}
    \forall \bm p ,\qty[ \abs{\hat u^{\bm p}_{\bm p} - \bra{\psi} (A^{\bm p}_{\bm p}\ket{\psi} }^2 ] \leq \varepsilon^2
\end{align}

Section \ref{sec:main-result} summarizes the main improvements by showing both the asymptotic and numerical results.
Section \ref{sec:k-RDN-with-CS} reviews the classical shadow algorithm for this setting. Section \ref{sec:k-RDM_with_QAE} presents the query complexity of the HL QAE algorithm. In Section \ref{sec:k-RDM-with-QGE}, we evaluate the total query complexity of proposed QGE algorithms.

\subsection{Main results}
\label{sec:main-result}

In this section, we present our main result on the asymptotic computational cost and the numerical evaluation of the total query complexity for $k$-RDM elements estimation in advance. The details of the analysis are presented in later sections.

\begin {table}[t]
\label{tab:title} 
\caption{
\justifying{Comparison of the worst-case complexities for determining the fermionic $k$-RDM of an $N$-qubit system with fixed $\eta$ particles where $\eta = k + \mac{O}(1)$ or $\eta = N - \mac{O}(1)$. The complexities are evaluated in terms of state preparation oracle queries and the required number of qubits across various strategies for measuring multiple observables. The compared approaches include Bell sampling+gentle measurement, naive sampling, classical shadow tomography, quantum amplitude estimation, prior QGE algorithms, and the proposed method. As shown in Appendix \ref{appsub:construction-kRDM}, we have access to block-encoding of target observables associated with $k$-RDM elements, which have the direct-sum structure. The dagger $(^\dagger)$ indicates that, if we have access to the block-encoding of projection operator on $\eta$ particle subspace, adaptive QGE algorithm in Ref.~\cite{wada2024Heisenberg} shows the same query complexity as in Method I. Here, $\tilde{\mac{O}}$ hides logarithmic factors in the dominant term.
}}
\begin{center}
 \label{tab:quantum_limits_bt}
   \begin{tabular}{lccc}
    \toprule
    \textbf{Method} & \textbf{Reference} & \textbf{Query Complexity} & \textbf{Space Complexity} \\
    \midrule
    Bell sampling+Gentle measurements  & \cite{huang2021information}
    &
    $\tilde{\mac{O} }(k \log(M) / \varepsilon^4 )$ & $2N$ \\
    Naive Sampling &
      & \( {\mathcal{O}}\bigl(N^{2k}\bigr)/\varepsilon^2\)
      & \(N\) \\
    Fermionic classical shadow tomography & \cite{Zhao:2020vxp} 
      & \( {\mathcal{O}}\bigl(N^k \bigr) /\varepsilon^2 \)
      & \(N\) \\
    Commuting Majorana pairs partitioning
    & \cite{bonet2020nearly}
    & $\mac{O}(N^k)/\varepsilon^2$ 
    & $N$ \\
    \midrule
     QAE algorithm & Sec.~\ref{sec:Expectation_base_QAE}
      & \({\mathcal{O}}\bigl(N^{2k}\bigr)/\varepsilon\)
      & \(N + \log_2(1/\varepsilon) + o(N)\) \\
     QGE algorithm & \cite{huggins2022nearly} 
      & \(\tilde{\mathcal{O}}\bigl(N^{k}/\varepsilon\bigr)\)
      & \(\mathcal{O}\bigl(N^{2k}\log(1/\varepsilon)\bigr)\) \\
     Adaptive QGE algorithm$^\dagger$ & \cite{wada2024Heisenberg} 
      & \(\tilde{\mathcal{O}}\bigl(N^{k+1/2}\bigr)/\varepsilon\)
      & \(\mathcal{O}\bigl(N^{2k}\bigr)\) \\
      \hline 
      \hline 
    Method I (ours) & Sec.~\ref{subsec:Symmetry-QGE-algorithm}
      & \(\tilde{\mathcal{O}}\bigl(N^{k/2}\bigr)/\varepsilon \)
      & \(\mathcal{O}\bigl(N^{2k}\bigr) \) \\
    Method II (ours)& Sec.~\ref{subsec:Integral-QGE}
      & \(\tilde{\mathcal{O}}\bigl(N^{k/2}\bigr)/\varepsilon\)
      & \(\mathcal{O}\bigl(k\,N^{2k}\log(N/\varepsilon)\bigr) \) \\
    \bottomrule
  \end{tabular}
\end{center}
\end{table}

We first present in Table~\ref{tab:quantum_limits_bt} the comparison of computational cost required to estimate all elements of fermionic $k$-RDM with root MSE of $\varepsilon$. Here, we consider the $N$-qubit system with a fixed particle number $\eta $ ($\eta = N - O(1), \eta = k + O(1)$). 
While shadow tomography~\cite{huang2021information} achieves an exponential reduction in the number of measurements at the cost of precision, high-precision estimation is required to accurately capture the rich structure of strongly correlated fermionic systems.
In this regard, we can see that our method achieves a clear asymptotic advantage when $\varepsilon \in o(N^{-k/6})$.
We further note that the best-known fermionic shadow tomography algorithm~\cite{Zhao:2020vxp} and the most efficient partitioning method~\cite{bonet2020nearly} require ${\mathcal{O}}(N^k)/\varepsilon^2$ queries, while our protocol reduces both the complexity on observable count $M$ and the scaling on the target precision $\varepsilon$. Remarkably, the scaling on the number of modes $N$ is improved even compared to methods that attain the Heisenberg-limited (HL) scaling. Concretely, our proposal realizes a quartic speedup over QAE algorithm, and a quadratic improvement over all prior QGE algorithms.
Taken together, these results establish our method as the most asymptotically efficient protocol for fermionic systems in the high-precision regime.

From a practical point of view, it is crucial to investigate how the constant and logarithmic factors pile up. As summarized in Fig.~\ref{fig:main_result_app}, we have  performed numerical evaluations to compare the performance of various algorithms in practical settings, in terms of the number of calls to the target state preparation unitary $U_\psi$ and its inverse. While the evaluation is valid for arbitrary fermionic system,
 the parameters setup implicitly assumes application to (1) the FeMo cofactor using 152 spin orbitals for the active space filled by 113 electrons and (2) $7/8$-filling Fermi-Hubbard model on $N$ qubits. 
The panels (a-c) of Fig.~\ref{fig:main_result_app} show the results for the FeMo cofactor. We find that Method II achieves the lowest query complexity for 1-RDM and 2-RDM estimation when $\varepsilon \leq 10^{-3}$. It is further remarkable that, for 3-RDM estimation, both Methods I and II outperform existing algorithms. The panels (d-f) of Fig.~\ref{fig:main_result_app} demonstrate that, for the $7/8$-filling Fermi-Hubbard model, Method II offers the best performance for 1-RDM estimation when $N \geq 80$. For 2-RDM estimation, both Methods I and II surpass all competitors for  in the range $10 \leq N \leq 100$.

\begin{figure}[h]
    \centering    
    \includegraphics[width=0.9\linewidth]{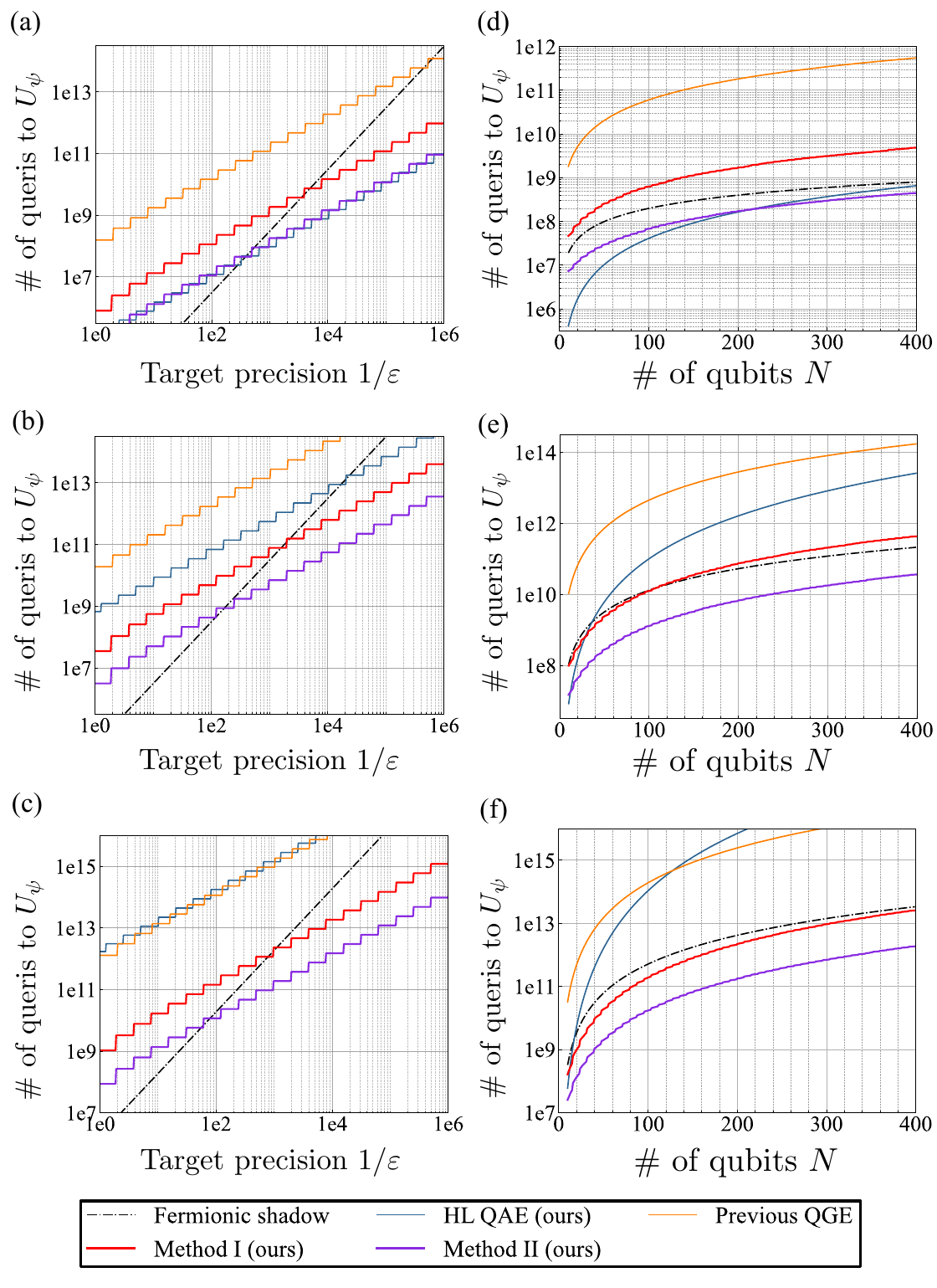}
    \caption{
    \justifying{
    The total query complexity in terms of state preparation unitary $U_\psi$ to evaluate (a) $1$-RDM elements, (b) $2$-RDM elements, and (c) $3$-RDM elements for the FeMo cofactor using active space of $N=152$ spin orbitals filled with $\eta=113$ electrons. The panels (d), (e), and (f) are similar plots for $N$-mode fermionic system with $\eta = \lceil (7/8) N \rceil$ particles, under fixed target precision $\varepsilon = 10^{-3}.$ 
    The approaches compared include, the Fermionic shadow tomography\cite{Zhao:2020vxp}, the quantum amplitude estimation (QAE) algorithm with HL scaling, the previous QGE algorithm \cite{wada2024Heisenberg}, and our proposals. 
    }
    }
    \label{fig:main_result_app}
\end{figure}

\subsection{$k$-RDM measurement with Classical shadow tomography}
\label{sec:k-RDN-with-CS}

Classical shadow techniques have been extensively studied for the task of $k$-RDM estimation, with various approaches and situations proposed in Refs.~\cite{Zhao:2020vxp, low2022classical, wan2023matchgate, babbush2023quantum, oGorman2022fermionic}. When only a single copy of the target state $\rho = \ketbra{\psi}$ is available, any sampling-based algorithm for estimating fermionic $k$-RDMs in $N$-qubit systems faces a lower bound of $\Omega(N^k / \varepsilon^2)$ queries.
Among these, the method by Zhao, Rubin, and Miyake~\cite{Zhao:2020vxp} achieves this lower bound, and thus, we focus on this classical shadow tomography protocol.

Let  $\gamma_j$ be a Majorana operator defined by 
\begin{align}
    \gamma_{2p} \coloneq a_p + a_p^\dagger, \gamma_{2p+1} \coloneq -i(a_p - a_p^\dagger).
\end{align}
We further define $2k$-degree Majorana operator as
\begin{align}
    \Gamma_{\bm \mu^{(2k)}} \coloneqq (-i)^k \gamma_{\mu_1} \cdots \gamma_{\mu_{2k} }, 
\end{align}
where $\bm \mu^{(2k)} \coloneq (\mu_1,\ldots,\mu_{2k} ), 0 \leq \mu_1 < \mu_2 < \cdots < \mu_{2k} $.  $\gamma_j$ 
By linearity, all $k$-RDM elements can be equivalently expressed by using $2k$-degree Majorana operator. For instance, 
\begin{align}
    a_1^\dagger a_0 = \frac14 (-\gamma_0 \gamma_2 - i \gamma_0\gamma_3 + i \gamma_1 \gamma_2 + \gamma_1 \gamma_3), \\
    a_0^\dagger a_1 = \frac14 (\gamma_0 \gamma_2 - i \gamma_0\gamma_3 + i \gamma_1 \gamma_2 - \gamma_1 \gamma_3).
\end{align}
Hence, the estimates of expectation values of Majorana operators lead to the estimates of $k$-RDM elements. 

It has been recognized that classical shadow tomography for fermionic systems benefits from using the number-conserving subgroup of fermionic Gaussian unitaries FGU$(N)$ for randomized basis transformation before the measurement~\cite{Zhao:2020vxp}.
By using such a group, the upper bound of the shadow norms over all $k$-degree Majorana operators can be evaluated as follows,
\begin{lemma}[Shadow norm using number-conserving FGU  ensemble \cite{Zhao:2020vxp}]
    Consider all $2k$-degree Majorana operators $\Gamma_{\bm \mu^{(2k)}} $ on $N$ fermionic modes, where $\bm \mu $ runs the set of all $2k$-combinations of $\{ 0, \ldots, 2N-1\}$. Under the FGU ensemble, the shadow norm satisfies,
    \begin{align}
        \norm{\Gamma_{\bm \mu^{(2k)} }}^2_{\mtr{FGU}} \leq \mqty(2N \\ 2k)\;\Big/\;\mqty(N \\ k) =  O(N^k).
    \end{align}
    where the shadow norm for any traceless observable \(O_j\) and an ensemble of random unitaries \(\mathcal{U}\) is defined as
    \begin{align}
    \|O\|_\mathcal{U}^2 
    \;\coloneq\; 
    \max_{\sigma}\,
    \mathbb{E}_{U \sim \mathcal{U},\, z \sim U \sigma U^\dagger}
    \Bigl[\bigl\langle z \bigr| U \,\mathcal{M}_{\mathcal{U}}^{-1}(O_j)\,U^\dagger \bigl| z \bigr\rangle^2\Bigr].
\end{align}
Here, \(\mathbb{E}_{U \sim \mathcal{U}, |z\rangle \sim U \rho U^\dagger}\) denotes the expectation value over random unitary transformations $\mac{U}$ and the possible measurement result $z$.
\end{lemma}

Based on this shadow norm, we can state the lemma regarding as the query complexity of classical shadow tomography against $k$-RDM elements estimation,
\begin{lemma}
    Let $\varepsilon \in (0,1) $.
    For a state preparation oracle $U_\psi$, 
    the classical shadow tomography yields the estimators $\hat u^{\bm p}_{\bm q}, \hat v^{\bm p}_{\bm q} $ satisfying Eqs.~\eqref{eq:cond1}, \eqref{eq:cond2}, and \eqref{eq:cond3}, requiring at most
    $
        \varepsilon^{-2} \cdot \mqty(2N \\ 2k)\;\Big/\; \mqty(N \\ k)
        = \mac{O}(N^k) / \varepsilon^2.
    $
    queries for $U_\psi$.
\end{lemma}

\begin{proof}

We now show that the shadow norm of the target observables can be upper bounded by that of $2k$-Majorana operators. For $k$-RDM estimation, the relevant observables $\{O_j\}_{j=1}^M$—corresponding to $\bigl(A_{\bm q}^{\bm p} + A_{\bm p}^{\bm q}\bigr)$, $\,i\bigl(A_{\bm q}^{\bm p} - A_{\bm p}^{\bm q}\bigr)$, or $A_{\bm p}^{\bm p}$—can be written in the form
\begin{align}
    h_{\emptyset}  \cdot \1 + \sum_{l=1}^{k} \sum_{\bm \mu^{(2l)} } h_{\bm \mu^{(2l)}} \Gamma_{\bm \mu^{(2l)} },
\end{align}
where $h_{\bm{\mu}^{(2l)}}, h_{\emptyset} \in \mathbb{R}$ are coefficients, and $h_{\emptyset}\,\1$ represents the identity part, which does not affect the MSE. Because $\|O_j\| \le 1$, one finds
$\abs{h_{\emptyset} } + \sum_{l=1}^{k} \sum_{\bm \mu^{(2l)} } \abs{h_{\bm \mu^{(2l)}} } \leq 1$.
Additionally, the shadow norm obeys the triangle inequality. Therefore, the shadow norm of the above expression satisfies
\begin{align}
    \norm{ h_{\emptyset}  \cdot \1 + \sum_{l=1}^{k} \sum_{\bm \mu^{(2l)} } h_{\bm \mu^{(2l)}} \Gamma_{\bm \mu^{(2l)} }
    }_{\mtr{FGU}} &\leq \sum_{l=1}^{k} \sum_{\bm \mu^{(2l)} } \abs{h_{\bm \mu^{(2l)}} } \norm{\Gamma_{\bm \mu^{(2l)} } }_{\mtr{FGU}}
    \\ 
    \label{eq:second_ineq_eval}
    &\leq \sum_{l=1}^{k} \sum_{\bm \mu^{(2l)} } \abs{h_{\bm \mu^{(2l)}} } \mqty(2N \\ 2k)\;\Big/\; \mqty(N \\ k) \\
    \label{eq:third_ineq_eval} 
    &\leq \mqty(2N \\ 2k)\;\Big/\; \mqty(N \\ k) .
\end{align}
The first inequality arises from the triangle inequality, and the second uses the fact that $l \leq k \Rightarrow \norm{\Gamma_{\bm \mu^{(2l)} } }_{\mtr{FGU}} \leq \mqty(2N \\ 2k)\;\Big/\; \mqty(N \\ k)$
Hence, the shadow norms of all observables associated with $k$-RDM elements are bounded above by those of $2k$-degree Majorana operators. Consequently, to ensure that classical-shadow estimators achieve an MSE below $\varepsilon^2$, it suffices to take
$L = \varepsilon^{-2} \cdot \mqty(2N \\ 2k)\;\Big/\; \mqty(N \\ k)$
samples.
\end{proof}

Note that we cannot relax our estimate of the number of classical-shadow samples in Eqs.~\eqref{eq:second_ineq_eval} and \eqref{eq:third_ineq_eval}. For $k$-RDM estimation, there exists an observable involving only $2k$ Majorana operators, which can be written as
\begin{align}
    \label{eq:decompose_2k_majorana}
    \sum_{\bm \mu^{(2k)} } {h_{\bm \mu^{(2k)}} } {\Gamma_{\bm \mu^{(2k)} } },
\end{align}
where $\sum_{\bm \mu^{(2k)} } \abs{h_{\bm \mu^{(2k)}} }  =1$. For instance,
\begin{align}
    a_1^\dagger a_0 + a_0^\dagger a_1 
    &= \frac12\!\bigl(\Gamma_{(0,3)} \;-\;\Gamma_{(1,2)}\bigr),
    \\
    a_3^\dagger a_2^\dagger a_1 a_0 + a_1^\dagger a_0^\dagger a_3 a_2
    &= \frac{1}{8}\Bigl(
    \Gamma_{(0,2,4,6)}
    -\Gamma_{(0,2,5,7)}
    +\Gamma_{(0,3,4,7)}
    +\Gamma_{(0,3,5,6)}
    +\Gamma_{(1,2,4,7)}
    +\Gamma_{(1,2,5,6)}
    -\Gamma_{(1,3,4,6)}
    +\Gamma_{(1,3,5,7)}
    \Bigr).
\end{align}
The shadow norm of such observables satisfies 
\begin{align}
    \Bigl\|
    \sum_{\bm{\mu}^{(2k)}}
    h_{\bm{\mu}^{(2k)}}\,\Gamma_{\bm{\mu}^{(2k)}}
    \Bigr\|_{\mathrm{FGU}}
    \;\;\le\;\;
    \sum_{\bm{\mu}^{(2k)}} \bigl\lvert h_{\bm{\mu}^{(2k)}}\bigr\rvert
    \,\bigl\|\Gamma_{\bm{\mu}^{(2k)}}\bigr\|_{\mathrm{FGU}}
    \;\;=\;\;
     \bigl\|\Gamma_{\bm{\mu}^{(2k)}}\bigr\|_{\mathrm{FGU}}.
\end{align}
In the last equality, we use $\sum_{\bm \mu^{(2k)} } \abs{h_{\bm \mu^{(2k)}} }  =1$.
To ensure that the MSEs of estimators for this observable are all below $\varepsilon^2$, it is necessary that $L \geq \varepsilon^{-2} \cdot \norm{\Gamma_{\bm \mu^{(2k)} } } $ hold for all $\bm{\mu}^{(2k)}$ with nonzero coefficient $h_{\bm{\mu}^{(2k)}}$ in Eq.~\eqref{eq:decompose_2k_majorana}.

From the above discussion, we can numerically compute the query complexity of classical-shadow-based estimators. The pseudocode for this calculation is given below.

\begin{algorithm}[H]
    \caption{ Calculation of the total query complexity of classical shadow tomography for $k$-RDM. }
  \label{alg:calc-query_CS}
  \begin{algorithmic}[1]
  \Statex \textbf{Input:} Qubit number $N$, Parameter $k \in \{1,2,\ldots,N\} $ that determines the kind of target RDM, Target precision parameter $\varepsilon \in (0, 1)$. 
  \Statex \textbf{Output:} 
  Total query complexity $L$ for $U_\psi$ and its inverse, sufficient to produce an estimator $\hat{u}$ that approximates the real or imaginary components of the entire \(k\)-RDM with MSE \(\varepsilon^2\).
  \State \textbf{return} $L := \left\lceil \varepsilon^{-2} \cdot \mqty(2N \\ 2k) \Big/ \mqty(N \\ k) \right\rceil$
  \end{algorithmic}
\end{algorithm}

\subsection{$k$-RDM measurement with QAE algorithm with HL scaling }
\label{sec:k-RDM_with_QAE}

In this section, we evaluate the total query complexity of the QAE algorithm for fermionic $k$-RDM estimation, assuming that we use the improved QAE algorithm presented in Fig.~\ref{fig:heisenberg-limited_improved_circuit} in Sec.~\ref{sec:est_with_QAE}. Given block-encodings of observables associated with $k$-RDM elements and a state preparation oracle, we can compute the total query complexity of the QAE algorithm. Since the explicit construction of block-encodings is somewhat involved, we defer the details to Sec.~\ref{appsub:construction-kRDM}, and focus solely on evaluating the number of calls to $U_\psi$ and its inverse here. For an $N$-mode fermionic system, the number of target observables to perform $k$-RDM estimation is given by $\binom{N}{k} \times \binom{N}{k}$.
Therefore, we can directly derive the following lemma by using Lemma.~\ref{thm:query_comp_HL_amp_est}:
\begin{lemma}
    Let $\varepsilon \in (0,1) $.
    For a state preparation oracle $U_\psi$ and $\mqty(N \\ k)^2$ $a$-block-encodings for observables associated with $k$-RDM elements, 
    the HL amplitude estimation algorithm yields the estimators $\hat u^{\bm p}_{\bm q}, \hat v^{\bm p}_{\bm q} $ satisfying Eqs.~\eqref{eq:cond1}, \eqref{eq:cond2}, and \eqref{eq:cond3}, requiring ${\mac{O} \qty(N^{2k})}/ {\varepsilon}$
    queries for $U_\psi$.
\end{lemma}

Additionally, we can numerically calculate the query complexity of QAE-based estimators. The pseudocode for this calculation is given in Algorithm~\ref{alg:calc-query_CS}.

\begin{algorithm}[H]
    \caption{ Calculation of the total query complexity for $k$-RDM estimation using QAE algorithm. }
  \label{alg:calc-query_CS}
  \begin{algorithmic}[1]
  \Statex \textbf{Input:} Qubit number $N$, Parameter $k \in \{1,2,\ldots,N\} $ indicating the locality of the RDM, target precision $\varepsilon \in (0, 1)$. 
  \Statex \textbf{Output:} 
  Total query complexity $L$ for $U_\psi$ and its inverse, sufficient to produce an estimator $\hat{u}$ that approximates the real or imaginary components of the entire \(k\)-RDM with MSE \(\varepsilon^2\).
  \State Set $q = \lceil\log(\pi/\varepsilon) \rceil$ and $M = \binom{N}{k}^2$.
  \State \textbf{return} $L \leftarrow  M(2^{q} +1)$, in accordance with Eq.~\eqref{eq:the_query_comp_HLamp}.
  \end{algorithmic}
\end{algorithm}

\subsection{$k$-RDM measurement with the QGE algorithm}
\label{sec:k-RDM-with-QGE}

In this section, we consider the property of observables corresponding to $k$-RDM elements, which leads to significant speedup for Method I proposed in the Sections \ref{subsec:Symmetry-QGE-algorithm}. Our main goal is to demonstrate that our proposal quadratically reduces total query complexity regarding the number of the target observables.

For $k$-RDMs elements estimation, Methods I and II can boost the efficiency of the QGE algorithm because fermionic pairs preserve the number conservation. Additionally, their nonzero elements on the computational basis are scarce, resulting in the reduction of $\norm{\sum O_j^2}$, which offers a significant benefit for the query complexity. Such an observation can be formally stated as the following lemma, using slightly abused notations of ${\rm Re}[A_{\bm p}^{\bm q}]\coloneqq (A_{\bm p}^{\bm q} + A_{\bm q}^{\bm p})/2, {\rm Im}[A_{\bm p}^{\bm q}]\coloneqq (A_{\bm p}^{\bm q} - A_{\bm q}^{\bm p})/2i$:
\begin{lemma}
    \label{lem:evaluation_norm}
    For $N$-mode fermionic system with $\eta$ particles~($k \leq \eta \leq N-k$), the 
    observables 
    $2 \mtr{Re}[A^{\bm p}_{\bm q}], 2 \mtr{Im}[A^{\bm p}_{\bm q}], A^{\bm p}_{\bm p} $ 
    satisfy the following inequality,
    \begin{align}
    \label{eq:eval_norm_kRDM}
        \norm{  \sum_{ \norm{\bm q}_1 < \norm{\bm p}_1   } \qty(2 \Pi_\eta \mtr{Re}[A^{\bm p}_{\bm q}] \Pi_\eta)^2 + \qty(2 \Pi_\eta \mtr{Im}[A^{\bm p}_{\bm q}] \Pi_\eta)^2 + \sum_{\bm p} (\Pi_\eta A^{\bm p}_{\bm p} \Pi_\eta )^2   }  
        &\leq 2 \mqty(\eta \\ k) \mqty(N - \eta  + k \\ k  ),
    \end{align}
    where $\Pi_\eta$ denotes the projection on $\eta$-particle subspaces. Here $\norm{\bm p}_1 \coloneq p_1 + p_2 + \ldots +p_k$ for an argument $\bm p = (p_1, p_2, \ldots,p_k)$.
\end{lemma}

\begin{proof}
    We first consider two ordered arguments of distinct elements, \(\bm{p} = (p_1 < p_2 < \dots < p_k)\) and \(\bm{q} = (q_1 < q_2 < \dots < q_k)\), both subsets of the set \(\{0, 1, \dots, N-1\}\). Suppose that these arguments \(\bm{p}\) and \(\bm{q}\) share \(m\) elements, where \(0 \leq m \leq k\). To simplify the counting, we first select the \(k\) elements of argument \(\bm{p}\) from the full set \(\{0,1,\dots,N-1\}\) in \(\binom{N}{k}\) ways. 
    Next, among these \(k\) elements, we choose \(m\) elements to serve as the common elements with argument \(\bm{q}\), which can be done in \(\binom{k}{m}\) ways. Finally, having used \(k\) elements for \(\bm{p}\), there remain \(N-k\) elements, from which we must select the additional \(k-m\) elements unique to \(\bm{q}\), giving \(\binom{N-k}{k-m}\) possibilities. Multiplying these factors yields the total number of ordered pairs \((\bm{p},\bm{q})\) sharing exactly \(m\) elements:
    \begin{align}
        \mqty(N \\ k)\mqty(N- k \\ k - m) \mqty(k \\ m ). 
    \end{align}

For arguments \(\bm{p}\) and \(\bm{q}\) with \(m\) common elements, where \(m \neq k\), the actions of the operators \(\left(2 \mathrm{Re}[A^{\bm{p}}_{\bm{q}}]\right)^2\) or \(\left(2 \mathrm{Im}[A^{\bm{p}}_{\bm{q}}]\right)^2\) are expressed as follows:
\begin{align}
    \left(2 \mathrm{Re}[A^{\bm{p}}_{\bm{q}}]\right)^2 = \left(2 \mathrm{Im}[A^{\bm{p}}_{\bm{q}}]\right)^2 = \Biggl(
  \prod_{p_i \in \bm{p}} (\ketbra{1}{1})_{p_i}
  \prod_{q_j \in (\bm{q} \setminus \bm{p})} (\ketbra{0}{0})_{q_j}
  + \prod_{p_i \in (\bm{p} \setminus \bm{q})} (\ketbra{0}{0})_{p_i}
  \prod_{q_j \in \bm{q}} (\ketbra{1}{1})_{q_j}
\Biggr)
\otimes \bigotimes_{r \notin (\bm{p}  \cup\bm{q})} I_r,
\end{align}
where \((\ketbra{0})_i\) denotes the projection operator on the \({i}\)-th qubit. 
From this expression, we observe that \(\left(2 \mathrm{Re}[A^{\bm{p}}_{\bm{q}}]\right)^2\) and \(\left(2 \mathrm{Im}[A^{\bm{p}}_{\bm{q}}]\right)^2\) can be expressed as a linear combination of rank-1 projectors
    $\prod_{n=0}^{N-1} (\ketbra{b_n})_n,$
where \(b_n \in \{0,1\}\) for all \(n \in \{0, 1, \dots, N-1\}\).

Note that a projection onto a fixed \(\eta\)-particle subspace \(\Pi_\eta\) imposes the following restriction on the rank-1 projection:
\begin{align}
    \Pi_\eta \left( \prod_{n=0}^{N-1} (\ketbra{b_n})_n \right) \Pi_\eta =
    \begin{cases} 
        \prod_{n=0}^{N-1} (\ketbra{b_n})_n &\quad \text{if } \sum_{n=0}^{N-1} b_n = \eta, \\
        0 &\quad \text{otherwise.}
    \end{cases}
\end{align}
This indicates that
the number of rank-1 projection terms in \(\left( 2 \mathrm{Re}[A^{\bm{p}}_{\bm{q}}] \right)^2\) and \(\left( 2 \mathrm{Im}[A^{\bm{p}}_{\bm{q}}] \right)^2\) is reduced under the action of \(\Pi_\eta\).
To quantitatively examine such a reduction, consider a typical rank-1 projection term appearing in \(\left( 2 \mathrm{Re}[A^{\bm p}_{\bm q}] \right)^2\):
\begin{align}
    \prod_{p_i \in \bm{p} } (\ketbra{1}{1})_{p_i}
    \prod_{q_j \in (\bm q \setminus \bm p)} (\ketbra{0}{0})_{q_j} 
    \prod_{r \notin (\bm p \cup \bm q)} (\ketbra{b_r})_r,
\end{align}
where \(b_r \in \{0,1\}\). Under the projection \(\Pi_\eta\), the bit string \(\{b_r\}\) must satisfy
\(
\sum_{r} b_r = \eta - k,
\)
since \(|\bm p| = k\) $(\ketbra{1})_{p_i}$ have already appeared. Given that \(|\bm p \cup \bm q| = 2k - m\), the number of unfixed positions \(r \notin (\bm p \cup \bm q)\) is \(N - (2k - m)\). Therefore, the number of valid bitstrings \(\{b_r\}\) is
\(
\binom{N - (2k - m)}{\eta - k}.
\)
Hence, the number of rank-1 projection terms involved in 
\(\left( 2 \mathrm{Re}[A^{\bm p}_{\bm q}] \right)^2\) or \(\left( 2 \mathrm{Im}[A^{\bm p}_{\bm q}] \right)^2\) for fixed $(\bm p, \bm q)$ and $m$ is \(
2\binom{N - (2k - m)}{\eta - k}.
\)

Combining this with the number of configurations for choosing \(\bm{p}, \bm{q}\), the total number of rank-1 projection terms in \(\left( 2 \mathrm{Re}[A^{\bm{p}}_{\bm{q}}] \right)^2\) and \(\left( 2 \mathrm{Im}[A^{\bm{p}}_{\bm{q}}] \right)^2\) for \(\bm{p} \neq \bm{q}\) is given by
\begin{align}
    \label{eq:combination}
    4 \sum_{m=1}^k \binom{N}{k} \binom{N-k}{k - m} \binom{k}{m} \cdot \binom{N - (2k - m)}{\eta - k}.
\end{align}
Here, we consider the total number of rank-1 projection terms of $\sum_{\norm{\bm q}_1 < \norm{\bm p }_1 } \left( 2 \mathrm{Re}[A^{\bm p}_{\bm q}] \right)^2 + \left( 2 \mathrm{Im}[A^{\bm p}_{\bm q}] \right)^2$, and this number becomes half of Eq.~\eqref{eq:combination}.

In the special case where \(\bm p = \bm q\), corresponding to $m=k$, we consider the square of a operator $A^{\bm p}_{\bm p}$:
\begin{align}
    (A^{\bm p}_{\bm p})^2 = \prod_{p_i \in \bm p} (\ketbra{1}{1})_{p_i} \prod_{r \notin \bm p} (\ketbra{b_r})_r,
\end{align}
where again \(b_r \in \{0,1\}\) and the constraint \(\sum_r b_r = \eta - k\) must be satisfied. Hence, the number of valid configurations is
\(
\binom{N - k}{\eta - k}.
\) 
Hence, the number of rank-1 projection in $\sum_{\bm p} \qty([A^{\bm p}_{\bm p}]^2 )$ is expressed as 
\begin{align}
    \binom{N}{k} \binom{N-2k}{\eta-k}.
\end{align}

Combining both cases, the total number of $1$-rank projection appearing in 
\begin{align}
    \label{eq:target_combination_op}
    \sum_{\norm{\bm q}_1 < \norm{\bm p}_1   } \qty(2 \Pi_\eta \mtr{Re}[A^{\bm p}_{\bm q}] \Pi_\eta)^2 + \qty(2 \Pi_\eta \mtr{Im}[A^{\bm p}_{\bm q}] \Pi_\eta)^2 + \sum_{\bm p} (\Pi_\eta A^{\bm p}_{\bm p} \Pi_\eta )^2 
\end{align}
is expressed as 
\begin{align}
    \binom{N}{k} \binom{N-2k}{\eta-k}  + 2\sum^k_{m=1}  \mqty(N \\ k)\mqty(N- k \\ k - m) \mqty(k \\ m ) \cdot \mqty(N-(2k - m ) \\ \eta-k).
\end{align}
If we consider all pairs \((\bm{p}, \bm{q})\) where \(\|\bm{q}\|_1 < \|\bm{p}\|_1\), each rank-1 projection should appear the same number of times. From this insight, we observe that Eq.~\eqref{eq:target_combination_op} is proportional to the identity operator on the \(\eta\)-particle Hilbert space.
Since the number of choices of \(\{b_n\}\) satisfying \(\sum_{n=0}^{N-1} b_r = \eta\) is \(d_\eta \coloneq \binom{N}{\eta}\), which corresponds to the dimension of the \(\eta\)-particle subspace,
 Eq.~\eqref{eq:target_combination_op} can be expressed as
\begin{align}
    \binom{N}{\eta}^{-1} \left[ \binom{N}{k} \binom{N - 2k}{\eta - k} + 2 \sum_{m=1}^k \binom{N}{k} \binom{N - k}{k - m} \binom{k}{m} \cdot \binom{N - (2k - m)}{\eta - k} \right] \cdot \mathds{1}_\eta,
    \label{eq:operator_particle_number}
\end{align}
where \(\mathds{1}_\eta\) denotes the identity operator on the \(\eta\)-particle Hilbert space.

The spectral norm of Eq.~\eqref{eq:operator_particle_number} is equivalent to its coefficient. For simplicity, we approximate the upper bound of this value as follows:
\begin{align}
    \label{eq:complex_formula}
    2 \binom{N}{\eta}^{-1}\sum_{m=0}^{k} \binom{N}{k} \binom{N - k}{k - m} \binom{k}{m} \cdot \binom{N - (2k - m)}{\eta - k} = 2 \binom{\eta}{k} \binom{N-\eta + k}{k}.
\end{align}
Deriving the right-hand side of this equality is somewhat complex, so we defer the details to Appendix~\ref{appsub:matrix_norm_eval}, which completes the proof of Eq.~\eqref{eq:eval_norm_kRDM}.
\end{proof}

From Theorems \ref{thm:evaluation_symmetry_QGE}, \ref{thm:evaluation_symmetry_parallel_QGE} and Lemma ~\ref{lem:evaluation_norm}, we can state the total query complexities for this problem with our QGE algorithms as follows,

\begin{lemma}
    Let $\varepsilon \in (0,1) $.
    For an oracle $U_\psi$ generating a state supported on $\eta$-electron subspace and $\mqty(N \\ k)^2$ $a$-block-encodings for $k$-RDM elements, 
    Method I yields the estimators $\hat u^{\bm p}_{\bm q}, \hat v^{\bm p}_{\bm q} $ satisfying Eqs.~\eqref{eq:cond1}, \eqref{eq:cond2}, and \eqref{eq:cond3}, requiring 
    \begin{align}
        \mac{O} \qty( \sqrt{ \mqty(\eta \\ k)
        \mqty(N - \eta+k  \\ k) \log( \mqty(N \\ \eta ) ) } \log(\mqty(N \\ k ) )     ) \Big/ \varepsilon = \tilde{\mac{O}} \qty( \sqrt{ \mqty(\eta \\ k) \mqty(N - \eta+k  \\ k) 
        } )\Big/ \varepsilon 
    \end{align}
    queries for $U_\psi$ and its inverse.
\end{lemma}

\begin{lemma}
    Let $\varepsilon \in (0,1) $.
    For an oracle $U_\psi$ generating a state supported on $\eta$-electron subspace and $\mqty(N \\ k)^2$ $a$-block-encodings for $k$-RDM elements, 
    Method II yields the estimators $\hat u^{\bm p}_{\bm q}, \hat v^{\bm p}_{\bm q} $ satisfying Eqs.~\eqref{eq:cond1}, \eqref{eq:cond2}, and \eqref{eq:cond3}, requiring 
    \begin{align}
        \mac{O} \qty( \sqrt{ \mqty(\eta \\ k)
        \mqty(N - \eta+k  \\ k) \log( \mqty(N \\ \eta ) ) \log(\mqty(N \\ k ) ) }   ) \Big/ \varepsilon = \tilde{\mac{O}} \qty( \sqrt{ \mqty(\eta \\ k) \mqty(N - \eta+k  \\ k) 
        } )\Big/ \varepsilon 
    \end{align}
    queries for $U_\psi$ and its inverse.
\end{lemma}
Note that if \(k \le \eta \ll N\) and \(k \ll \eta \sim N\), then
$
\tilde{\mathcal{O}}\!\Bigl(\sqrt{\tbinom{\eta}{k}\,\tbinom{N - \eta + k}{k}}\Bigr) = 
\tilde{\mathcal{O}}(N^{k/2}).
$
Under these conditions, we therefore obtain an additional quadratic speedup in terms of the qubit number \(N\) for $k$-RDM elements estimation.

Next, we focus on numerically evaluating the total query complexities of our proposals for $k$-RDM estimation including the constant factors. The procedures described in Algorithms \ref{alg:query-symmetry-QGE} and \ref{alg:query-symmetry+parallel-QGE} provide detailed outlines of these calculations. As a subroutine, Algorithm~\ref{alg:func_degree_HS} describes a function named $\texttt{HS\_degree}$ which determines the minimum polynomial degree \(Q\) for optimal Hamiltonian simulation. The core principle, following Ref.~\cite{gilyen2019quantum}, is that, for a nonzero parameter \(t \in \mathbb{R}\), the minimal degree \(Q\) of a polynomial \(f(x)\) approximating \(e^{ixt}\) with accuracy \(\varepsilon''\) is given by
\cite{gilyen2019quantum}
\begin{align}
    Q = -1 + \min \left\{ l \in \mathbb{N} : \frac{4t^l}{2^l l!} \leq \frac{\varepsilon''}8 \right\}.
\end{align}
To find the smallest integer \(l \) satisfying $(4t^l)/(2^l l!) \leq \varepsilon''/8 $, we employ a binary search, as shown in Algorithm~\ref{alg:func_degree_HS}. In Algorithms \ref{alg:query-symmetry-QGE} and \ref{alg:query-symmetry+parallel-QGE}, we set \(c = {1}/{(80(1+\pi)^2}) \) to reduce \(q_{\max}\) for a given target precision \(\varepsilon\). By choosing this parameter, we slightly decrease \(q_{\max}\) from \(\lceil\log_2\!\bigl(1/\varepsilon\bigr)\rceil\) to \(\lceil\log_2\!\bigl(1/(\sqrt{2}\,\varepsilon)\bigr)\rceil\). To make this explicit, we refer to the analysis in Eq.~\eqref{eq:the_choice_of_c}.
When we define \(q \coloneq \lceil\log_2\bigl(1/\tilde{\varepsilon}\bigr)\rceil\) and let \(c = {1}/{(80(1+\pi)^2}) \), the MSE of the estimators satisfies
\begin{align}
    \mathbb{E}[(\hat u_j - \langle O_j\rangle)^2] &\leq \frac{1}{2^{2(q_{\max} + 1)}} + (1+\pi)^2 c\, 2^{-2 q_{\max}+1} \\
    &= \frac{11}{40} \cdot \frac{1}{2^{2q_{\max} }}  
    \leq \frac{11\tilde\varepsilon^2}{40}\sim \frac{\tilde\varepsilon^2}{4}.
\end{align}
Hence, by setting \(c\) in this manner, we can choose \(\tilde{\varepsilon} = \sqrt{40/11}\,\varepsilon \sim 2 \varepsilon\), which slightly reduces the overall query complexity.

\begin{algorithm}[H]
    \caption{ The function yielding the degree of the polynomial approximating Hamiltonian simulation \cite{gilyen2019quantum}. }
  \label{alg:func_degree_HS}
\begin{algorithmic}[1]
\Statex \textnormal{Input: time parameter $t$, and target precision parameter $\varepsilon'' \in (0, 1)$.}
\Statex \textnormal{Output: the minimum degree $Q$ of a polynomial $f(x)$ that approximates $e^{ixt}$ with accuracy $\varepsilon''$.}
\Function{\texttt{HS\_degree}}{$t,\varepsilon''$}
    \State $Q_{\mathrm{upp}} \gets t$, $Q_{\mathrm{low}} \gets t$
    \While{true}
        \State $Q_{\mathrm{upp}} \gets 2\cdot Q_{\mathrm{upp}}$
        \If{$\dfrac{4t^{\,Q_{\mathrm{upp}}}}{2^{\,Q_{\mathrm{upp}}}\,Q_{\mathrm{upp}}!} < \dfrac{\varepsilon''}{8}$}
            \State break
        \EndIf
    \EndWhile
    \While{true}
        \State $Q_{\mathrm{mid}} \gets \lceil (Q_{\mathrm{upp}}+Q_{\mathrm{low}})/2 \rceil$
        \If{$\dfrac{4t^{\,Q_{\mathrm{mid}}}}{2^{\,Q_{\mathrm{mid}}}\,Q_{\mathrm{mid}}!} < \dfrac{\varepsilon''}{8}$}
            \State $Q_{\mathrm{low}} \gets Q_{\mathrm{mid}}$
        \Else
            \State $Q_{\mathrm{upp}} \gets Q_{\mathrm{mid}}$
        \EndIf
        \If{$Q_{\mathrm{mid}} = \lceil (Q_{\mathrm{upp}}+Q_{\mathrm{low}})/2 \rceil$}
            \State break
        \EndIf
    \EndWhile
    \State \textbf{return} $Q = \lceil (Q_{\mathrm{upp}}+Q_{\mathrm{low}})/2 \rceil$
\EndFunction
\end{algorithmic}
\end{algorithm}

\begin{algorithm}[H]
    \caption{Numerical evaluation of the total query complexity of Method I for fermionic $k$-RDM estimation
    }
  \label{alg:query-symmetry-QGE}
  \begin{algorithmic}[1]
  \Statex \textbf{Input:} qubit number $N$, particle number $\eta$, parameter $k \in \{1,2,\ldots,N\} $, target precision parameter $\varepsilon \in (0, 1)$. 
  \Statex \textbf{Output:} 
  The total query complexity $L$ for $U_\psi$ and its inverse required to yield an estimator $\hat{u} = (\hat{u}_1, \dots, \hat{u}_M)$ whose $j$-th element estimates $\langle O_j \rangle := \langle \psi | O_j | \psi \rangle$ within MSE $\epsilon^2$ as 
  $$
  \max_{j=1,2,\dots,M} \mathbb{E}[(\hat{u}_j - \langle O_j \rangle)^2] \leq \epsilon^2
  $$
  \State Set $p=3, c = 1/(80(1+\pi)^2)$, $M = \mqty(N \\ k)^2$, and calculate the $\sigma_\Delta$ expressed in Eq.~\eqref{eq:def_sigma_Delta} for $k$-RDM elements. Explicitly,
      \begin{align}
          \sigma_\Delta = \left\lceil  \sqrt{4 v \mqty(\eta \\k ) \mqty(N-\eta +k \\ k) \log(2 \mqty(N \\ \eta ) /\delta')} + \frac43 \log(2 \mqty(N \\ \eta ) /\delta')  \right\rceil ,
      \end{align}
      where $\delta' = 2^{-10} $ and $v = \mathbb{E}[(2X)^2] = 0.1652$ comes from the variance of random variable $X$ with  $\forall j, \mtr{Pr} [X = x ] = \abs{c_{x}^{(\cos) } }^2  $.
  \For{$q = 0, 1, \dots, q_{\max} \coloneq \lceil \log_2(1/(\sqrt {40/11}\epsilon) ) \rceil$}
      \State Set 
      $\delta^{(q)} \leftarrow \frac{c}{8^{q_{\max} - q} }, \quad R^{(q)} \leftarrow \tilde{\Phi}^{-1}(\delta^{(q)} / 2M) ,$
      where $\tilde \Phi (R ) \coloneq 1- F(\lfloor (R+1)/2 \rfloor-1; 0.011+1/12, R)$ and $F(k; \mu, n)$ is the cumulative distribution function of a binomial distribution with success probability $\mu$ and $n$ trials.
      \State $t \leftarrow 2^{p+q+1}\sigma_\Delta$,   $Q \leftarrow \texttt{HS\_DEGREE} (t, \varepsilon'' = 2^{-14})$
      \State $L^{(q+1) } \leftarrow L^{(q)} + 2QR^{(q)} $
  \EndFor
  \State $L \leftarrow L^{(q_{\rm max}+1)}$
  \State \textbf{return} $L$
  \end{algorithmic}
\end{algorithm}

\begin{algorithm}[H]
    \caption{ Numerical evaluation of the total query complexity of Method II for fermionic $k$-RDM  estimation}
   \label{alg:query-symmetry+parallel-QGE}
  \begin{algorithmic}[1]
  \Statex \textbf{Input:} qubit number $N$, particle number $\eta$, parameter $k \in \{1,2,\ldots,N\} $, target precision parameter $\varepsilon \in (0, 1)$. 
  \Statex \textbf{Output:} 
  The total query complexity $L$ for $U_\psi$ and its inverse required to yield an estimator $\hat{u} = (\hat{u}_1, \dots, \hat{u}_M)$ whose $j$-th element estimates $\langle O_j \rangle := \langle \psi | O_j | \psi \rangle$ within MSE $\epsilon^2$ as 
  $$
  \max_{j=1,2,\dots,M} \mathbb{E}[(\hat{u}_j - \langle O_j \rangle)^2] \leq \epsilon^2
  $$
  \State Set $p=3, c = 1/(80(1+\pi)^2)$, and $M = \mqty(N \\ k)^2$.
  \For{$q = 0, 1, \dots, q_{\max} \coloneq \lceil \log_2(1/(\sqrt {40/11}\epsilon) ) \rceil$}
      \State Set 
      $\delta^{(q)} \leftarrow \frac{c}{8^{q_{\max} - q} }, \quad R^{(q)} \leftarrow \tilde{\Phi}^{-1}(\delta^{(q)} / 2M) ,$
      where $\tilde \Phi (R ) \coloneq 1- F(\lfloor (R+1)/2 \rfloor-1; 0.011, R)$ and $F(k; \mu, n)$ is the cumulative distribution function of a binomial distribution with success probability $\mu$ and $n$ trials.
      \State Calculate the $\tilde{\sigma}^{(q)}_\Delta$ for $k$-RDM elements expressed as,
      \begin{align}
          \tilde{\sigma}^{(q)}_\Delta = \left\lceil  \sqrt{4 v R^{(q)} \mqty(\eta \\k ) \mqty(N-\eta +k \\ k) \log(2 \mqty(N \\ \eta ) /\delta')} + \frac43 \log(2 \mqty(N \\ \eta ) /\delta') \right\rceil ,
      \end{align}
      where $\delta' = (\delta^{(q)})^2/80 $ and $v = \mathbb{E}[(2X)^2] = 0.1652$ comes from the variance of random variable $X$ with  $\forall j, \mtr{Pr} [X = x ] = \abs{c_{x}^{(\cos) } }^2  $.
      \State $t \leftarrow 2^{p+q+1}\tilde{\sigma}^{(q)}_\Delta$,  $Q \leftarrow \texttt{HS\_DEGREE} (t, \varepsilon'' = (\delta^{(q)})^2 / 2^{6})$
      \State $L^{(q+1) } \leftarrow L^{(q)} + 2Q $
  \EndFor
  \State $L \leftarrow L^{(q_{\rm max}+1)}$
  \State \textbf{return} $L$
  \end{algorithmic}
\end{algorithm}

\section{Discussion} \label{sec:discussion}
In this work, we have generalized the adaptive scheme of the QGE algorithm and presented two highly practical variants. The first variant, Method I, leverages the inherent symmetries of quantum systems and Method II adopts a framework that partially parallelizes the probe system during adaptive procedure  at the cost of additional qubits. We have shown that our proposals not only maintain the quantum enhancement achieved in the adaptive QGE algorithm, but also significantly reduce the total query complexity. As a result, our method not only delivers better asymptotic performance, but also offers tangible benefits for practical problems such as fermionic partial tomography. Given the central role the fermionic $k$-RDM play in analyzing quantum correlations, our approach holds promise for wide-ranging applications in fields such as condensed matter physics, quantum chemistry, and high-energy physics.

There are many potential avenues for future work. Firstly, it is of practical importance to understand how algorithmic errors and hardware-induced noise influence the algorithm’s performance. In particular, it remains unclear whether phase estimation calibration methods can be directly used to stabilize measurement outcomes~\cite{kimmel2015robust}.
Secondly, since the algorithm encodes observable information into quantum phases through Hamiltonian simulation, a natural theoretical question is whether it can benefit from fast-forwarding techniques—especially those capable of exponentially reducing the computational cost for commuting Hamiltonians~\cite{atia2017fast}.
Thirdly, we envision that it would be crucial for practitioners of quantum computing to perform in-depth resource estimation on early FTQC architecture.

\section*{Acknowledgements}
The authors wish to thank Dominic Berry, Yosuke Mitsuhashi, Takahiro Sagawa, and Kento Tsubouchi for fruitful discussions.
Y. K. is supported by the Program for Leading Graduate Schools (MERIT-WINGS).
K. W. was supported by JSPS KAKENHI Grant Number JP24KJ1963.
W. M. is supported by MEXT Quantum Leap Flagship Program
(MEXTQLEAP) Grant No. JPMXS0120319794,  the
JST COI-NEXT Program Grant No. JPMJPF2014,
the JST ASPIRE Program Grant No. JPMJAP2319,
and the JSPS
Grants-in-Aid for Scientific Research (KAKENHI) Grant
No. JP23H03819.
N.Y. is supported by JST Grant Number JPMJPF2221, JST CREST Grant Number JPMJCR23I4, IBM Quantum, JST ASPIRE Grant Number JPMJAP2316, JST ERATO Grant Number JPMJER2302, and Institute of AI and Beyond of the University of Tokyo.

\bibliography{ref}

\appendix

\section{Quantum arithmetic techniques}
\label{sec:Quantum_arithmetic}
In this section, we state key results related to  QSVT. By leveraging these QSVT techniques, we can explicitly construct novel QGE algorithms and thereby derive the total query complexity for $U_\psi$ and its inverse and other computational costs.

\begin{lemma}[Uniform singular value amplification \cite{gilyen2019quantum, low2017hamiltonian}]
\label{lem:uniform_amplification}
Let \(\gamma > 1\) and let \(\delta, \varepsilon \in (0,1/2)\). Suppose we have an \(a\)-block-encoding \(U\) of \(A\) with \(\|A\| \leq \frac{1-\delta}{\gamma}\). Then, one can implement a \((1, a+1, \varepsilon)\)-block-encoding of \(\gamma A\) using 
$
m = \mathcal{O}\!\left(\frac{\gamma}{\delta} \log\!\Bigl(\frac{\gamma}{\varepsilon}\Bigr)\right)
$
queries to \(U\) or \(U^\dagger\), \(2m\) controlled NOT gates with control on \(a\) qubits, \(\mathcal{O}(m)\) single-qubit gates, and \(\mathcal{O}(\mathrm{poly}(m))\) classical computation to determine the necessary circuit parameters.
\end{lemma}

\begin{lemma}[Optimal block-Hamiltonian simulation \cite{low2019hamiltonian}]
\label{lem:optimal_HS}
Let \(t \in \mathbb{R}\setminus\{0\}\) and \(\varepsilon'' \in (0,1)\), and suppose that \(U\) is a \((1, a, 0)\)-block-encoding of a Hamiltonian \(H\). Then, one can implement a \((1, a+2, \varepsilon'')\)-block-encoding of \(e^{itH}\) using \(4Q\) queries to the controlled version of \(U\) (or its inverse), \(2Q\) NOT gates controlled by \((a+1)\)-qubit registers, \(\mathcal{O}(Q)\) single-qubit or two-qubit gates, and \(\mathcal{O}(\mathrm{poly}(Q))\) classical computation to determine the circuit parameters, where
$
Q = \mathcal{O}\Bigl(|t| + \log\Bigl({1}/{\varepsilon''}\Bigr)\Bigr).
$
\end{lemma}

\section{Classical shadow tomography}
\label{sec:classical-shadow-review}

In this section, we concisely review the 
classical shadow tomography~\cite{Huang:2020tih} and analyze its sample complexity in terms of achieving a desired MSE. The core idea of classical shadow tomography is to perform randomized measurement such that numerous (local and non-commuting) observables can be estimated efficiently than ordinary projective measurement protocols. Among various extensions, we later discuss the performance of classical shadow tomography tailored for fermionic systems in Sec.~\ref{sec:Application}.

Suppose $\rho$ is an $n$-qubit quantum state, and $\{O_1, \ldots, O_M\}$ is a set of $M$ traceless observables whose expectation values $\mathrm{tr}(O_1 \rho), ..., \mathrm{tr}(O_M \rho)$ we aim to estimate. 
The classical shadow tomography relies on a simple measurement primitive that consists of three steps: (i) for each preparation of \(\rho\), apply a unitary transformation \(\rho \mapsto U \rho U^\dagger\) with \(U\) randomly sampled from an ensemble \(\mathcal{U}\); (ii) perform a projective measurement in the computational basis \(\{|z\rangle \,|\, z \in \{0,1\}^n\}\); and (iii) record the measurement outcome $z$ along with the corresponding unitary \(U\).
With three steps combined all together, the process corresponds to the following quantum channel:
\begin{align}
\label{eq:shadow_channel}
    \mathcal{M}_\mathcal{U}(\rho) := \mathbb{E}_{U \sim \mac{U}, z \sim U \rho U^\dagger}[U^\dagger \ketbra{z} U],
\end{align}
where the measurement outcome $z$ follows the probability distribution determined from the Born's rule as 
$p_{z|U} = \bra{z} U \rho U^\dagger \ket{z}$. This is indicated shorthandly by the expression $z\sim U \rho U^\dagger.$


When the channel $\mathcal{M}_\mathcal{U}$ is invertible and efficiently expressed on a classical computer, we can have a classical description of the following state:
\begin{align}\label{eq:shadow_def}
    \hat{\rho}_{U, z} := \mathcal{M}_\mathcal{U}^{-1}(U^\dagger \ketbra{z} U),
\end{align}
where $\hat{\rho}$ is called {\it classical shadow}. From Eqs.~\eqref{eq:shadow_channel} and~\eqref{eq:shadow_def}, it is straightforward to see that mixture of classical shadows based on the Born's rule gives an unbiased estimator of $\rho$ as 
\begin{eqnarray}
    \rho = \mathbb{E}_{U \sim \mac{U}, z \sim U\rho U^\dagger}[\hat{\rho}_{U, z}]
\end{eqnarray}
Accordingly, we can estimate expectation values for any observable $O$:
\begin{align}
    \mathrm{tr}(O \rho) = \mathbb{E}_{U \sim \mathcal{U}, z \sim U \rho U^\dagger}[\mathrm{tr}(O \hat{\rho}_{U,z})]
\end{align}


Here we are interested in evaluating the sample complexity in terms of the MSE since it can bound the worst-case behavior of the algorithm. 
To this end, we require an upper bound on the variances of all estimators for \(\mathrm{tr}(O_j\,\rho)\). This can be derived from the shadow norm \cite{Huang:2020tih}, defined by
\begin{align}
    \|O_j\|_\mathcal{U}^2 
    \;\coloneq\; 
    \max_{\sigma}\,
    \mathbb{E}_{U \sim \mathcal{U},\, z \sim U \sigma U^\dagger}
    \Bigl[\bigl\langle z \bigr| U \,\mathcal{M}_{\mathcal{U}}^{-1}(O_j)\,U^\dagger \bigl| z \bigr\rangle^2\Bigr].
\end{align}
From definition, $\mathbb{E} [ \abs{\tr(O_j \hat{\rho} ) -\tr(O_j \rho)   }^2 ]\leq \norm{O_j}{\mac U}^2 $ holds.
Because a classical shadow yields an unbiased estimator, its MSE coincides with its variance. For this context, we gather \(N\) independent classical-shadow samples 
\(\{\hat{\rho}^{(l)}\}_{l=1}^N\) and define the empirical mean for each observable by
\begin{align}
    \hat{\omega}_j(N) \coloneq \frac{1}{N}\sum_{l=1}^N \mathrm{tr}\!\bigl(O_j\,\hat{\rho}^{(l)}\bigr).
\end{align}
One then finds, for all \(j \in \{1,\dots,M\}\),
\begin{align}
    \mathrm{MSE}\bigl[\hat{\omega}_j(N)\bigr] 
    \coloneq 
    \mathbb{E}\Bigl[\bigl(\hat{\omega}_j(N)\;-\;\mathrm{tr}\!\bigl(O_j\,\rho\bigr)\bigr)^2\Bigr]
    \;\le\; 
    \frac{\|O_j\|_{\mathcal{U}}^2}{N},
\end{align}
Thus, for a given \(\varepsilon \in (0,1)\), choosing $
N = 
{\max_{\,1 \le j \le M}\,\bigl\|O_j\bigr\|_{\mathcal{U}}^2}/{\varepsilon^2}
$
ensures that \(\mathrm{MSE}\bigl[\hat{\omega}_j(N)\bigr] \le \varepsilon^2\) for all \(j\in\{1,\dots,M\}.\)

\section{Subspace QSVT}
\label{sec:subspace_QSVT}
\textit{Quantum singular value transformation} (QSVT) \cite{gilyen2019quantum} is a powerful framework that enables an efficient implementation of matrix polynomials on a quantum computer. Given a block-encoding \( U_A \) of a matrix \( A \), one can systematically construct block-encodings for a variety of matrix polynomials \( f(A) \) by using the circuit expressed in Fig.~\ref{fig:description_QSVT}, even without any prior knowledge of the structure of \( U_A \) or \( A \). Recent studies, however, have demonstrated that focusing on low-energy subspaces can improve the query complexity of QSVT \cite{low2017hamiltonian, zlokapa2024hamiltonian}. Motivated by these insights, we propose the \textit{subspace QSVT} framework, which generalizes the application of QSVT from low-energy subspaces to more general subspaces. This approach has the potential to significantly improve QSVT's efficiency across a broad range of subspaces. Notably, this approach offers remarkable computational benefits for estimating expectation values of observables related to symmetry via the QGE algorithm, as we demonstrate in later sections.

\begin{figure}[t]
    \centering
    \includegraphics[width=0.8\linewidth]{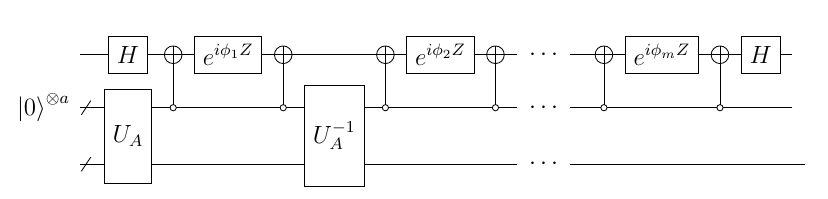}
    \caption{The QSVT circuit \(U_\Phi\) transforms a block-encoding \(U_A\) of \(A\) into a block-encoding of \(f(A)\), where \(f: [-1,1] \to [-1,1]\) is a definite-parity polynomial of degree \(m\). The phase angles \(\{\phi_i\}\) can be classically computed.}
    \label{fig:description_QSVT}
\end{figure}

For simplicity, we first consider the case where the target operator $O \in \mathbb{C}^{d \times d}$ is Hermitian, and later extend the discussion to more general operations. In particular, we consider the situation where \( O\) has the following direct sum decomposition:
\begin{align}
    O & = \bigoplus_{\lambda \in \Lambda} O^{(\lambda)}
    = \begin{pmatrix}
O^{(1)} & \bm O & \cdots & \bm O  \\
\bm O  & O^{(2)}  & \cdots & \bm O  \\
\vdots & \vdots & \ddots & \vdots \\
\bm O  & \bm O  & \cdots & O^{(\abs{\Lambda})}
\end{pmatrix}\,  
\end{align}
where we define \(\Lambda = \{1,2,\ldots,|\Lambda|\}\) as the label set, \(O^{(\lambda)} \in \mathbb{C}^{d_\lambda \times d_\lambda } \) is the summand corresponding to the label \(\lambda\), and \(\bm O\) denotes the zero matrix. With an appropriate choice of basis, we assume that there exists an \(a\)-block-encoding \(U_O\) such that
\begin{align}
    (\1 \otimes \bra{\bm 0}^{\otimes a} ) U_O  (\1 \otimes \ket{\bm 0}^{\otimes a} ) = O^{(1)} \oplus O^{(2)} \oplus \cdots \oplus O^{(\abs{\Lambda})}.
\end{align}
For an operator expressed as a direct sum, we now state the following lemma.
\begin{lemma}[Subspace QSVT]
    \label{lem:subspace-qsvt}
    Let $U_O$ be a $a$-block-encoding of Hermitian operator $O \in \mathbb{C}^{d \times d}$ satisfying $O =  \bigoplus_{\lambda \in \Lambda} O^{(\lambda)} $, where we define $\Lambda = \{1,2,\ldots, \abs{\Lambda} \} $ as a label set and $O^{(\lambda)} \in \mathbb{C}^{d_\lambda \times d_\lambda} $ is a corresponding summand.  We also assume $\Phi = (\phi_1, \ldots, \phi_m) \in \mathbb{R}^{m} $ be the sequence of phase factors such that $\mtr{QSP}(\Phi, x)$ computes the degree-$m$ polynomial $f(x)$. Then, the QSVT circuit $U_\Phi$ described in Fig.~\ref{fig:description_QSVT} for $U_O$ is a $(a+1)$-block encoding of  
    \begin{align}
        \bigoplus_{\lambda \in \Lambda } f(O^{(\lambda)}  )
        = \begin{pmatrix}
f\qty( O^{(\lambda_1)} ) & \bm O & \cdots & \bm O  \\
\bm O  & f\qty( O^{(\lambda_2)}  ) & \cdots & \bm O  \\
\vdots & \vdots & \ddots & \vdots \\
\bm O  & \bm O  & \cdots & f \qty( O^{(\abs{\Lambda})} )
\end{pmatrix}\,
    \end{align}
\end{lemma}

\begin{proof}
From the direct sum decomposition, we can derive the eigenvalue decomposition of each $d_\lambda$-dimensional operator $O^{(\lambda)}$ and denote it as $O^{(\lambda)} = U^{(\lambda)} \Sigma^{(\lambda)} (U^{(\lambda)})^\dagger$ where $U^{(\lambda)}$ is a $d_\lambda \times d_\lambda$ unitary matrix. Since the $O^{(\lambda)}$ is Hermitian, there exists an eigenvalue decomposition $O^{(\lambda)} = U^{(\lambda)} \Sigma^{(\lambda)} (U^{(\lambda)})^\dagger$ where $U^{(\lambda)}$ is a unitary matrix and $\Sigma^{(\lambda)}$ is a diagonal matrix. This eigenvalue decomposition leads to 
\begin{align}
     O 
     = \bigoplus_{\lambda \in \Lambda} U^{(\lambda)}  \Sigma^{(\lambda) \dagger} U^{(\lambda)}.
\end{align}
Here, we define $U \coloneq \bigoplus_{\lambda \in \Lambda} U^{(\lambda)} $, and it is trivial to show that $U$ is unitary. Since $\bigoplus_{\lambda \in \Lambda} \Sigma^{(\lambda)} $ is a diagonal matrix, we can consider the eigenvalue decomposition of $O$ is equivalent to $U \bigoplus_{\lambda \in \Lambda} \Sigma^{(\lambda)}   (U)^\dagger $.

Given $\Phi$ related to the degree-$m$ polynomial $f(x)$ and a $a$-block-encoding $U_O$ of Hermitian operator $O$, QSVT circuit $U_\Phi$ \cite{gilyen2019quantum} provides a method for implementing 
\begin{align}
    f^{\mtr{(SV)}} (O) \coloneq U f \qty(  \bigoplus_{\lambda \in \Lambda } \Sigma^{(\lambda)}   ) U^\dagger.
\end{align}
Since $f$ is a polynomial, we can demonstrate that
\begin{align}
    U f \qty(  \bigoplus_{\lambda \in \Lambda } \Sigma^{(\lambda)}) U^\dagger &= f  \qty( U  \qty( \bigoplus_{\lambda \in \Lambda }  \Sigma^{(\lambda)} )  U^\dagger ) \\
    &= f \qty(\bigoplus_{\lambda \in \Lambda} U^{(\lambda)} \Sigma^{(\lambda)} (U^{(\lambda)})^\dagger )  \\
    &= f \qty( \bigoplus_{\lambda \in \Lambda} O^{(\lambda)} ) = \bigoplus_{\lambda \in \Lambda } f(O^{(\lambda)} ).
\end{align}
Therefore, $( \1 \otimes \bra{\bm 0} ) U_\Phi ( \1 \otimes \ket{\bm 0} )$ is $(a+1)$-block-encoding of $\bigoplus_{\lambda \in \Lambda } f(O^{(\lambda)} )  $.
\end{proof}

\begin{Remark}
    This technique is derived directly from a primitive construction of a QSVT circuit. Nonetheless, we restate this lemma because it reveals that the conditions for eigenvalues to successfully implement QSVT can be eased when our interest is confined to certain subspaces—especially when the target operator is represented as a direct sum. 
    In other words, by using prior knowledge of the operator's structure, such as its symmetry or low-energy domain \cite{low2017hamiltonian, zlokapa2024hamiltonian}, we may employ lower degrees of polynomial to approximate the desired transformation.
    Indeed, the significant advantages of this technique are concretely illustrated as examples in Secs.~\ref{subsec:Symmetry-QGE-algorithm} and \ref{sec:k-RDM-with-QGE}.
\end{Remark}

In the above lemma, we have focused on the case where the target operator is Hermitian. Here, we show a more general version to transform the singular values of block-encoded operators that have a direct sum structure.
Let $U_{A}$ be an $a$-block-encoding of an operator $A \in \mathbb{C}^{d' \times d}$. Here, we assume that 
$A$ has the following direct sum decomposition:
\begin{align}
    A = \bigoplus_{\lambda \in \Lambda} A^{(\lambda)} 
    =
    \begin{pmatrix}A_{1} & \bm O & \cdots & \bm O  \\
\bm O  & A_{2}  & \cdots & \bm O  \\
\vdots & \vdots & \ddots & \vdots \\
\bm O  & \bm O  & \cdots & A_{\abs{\Lambda}}
\end{pmatrix}\, 
 \end{align}
 where we define \(\Lambda = \{1,2,\ldots,|\Lambda|\}\) as the label set, and \(A^{(\lambda)} \in \mathbb{C}^{d_\lambda' \times d_\lambda} \) is the summand corresponding to the label \(\lambda\).   
For each $A^{(\lambda)}$, we write its singular value decomposition as
\begin{equation}
    A^{(\lambda)} = \sum_{k=1}^{\min(d_{\lambda}', d_\lambda)} \varsigma_{k}^{(\lambda)} 
    |\widetilde{\psi}^{(\lambda)}_k\rangle\langle{{\psi}^{(\lambda)}_k}|.
\end{equation}
Using this notation, we have
\begin{equation}
    A=\bigoplus_{\lambda\in \Lambda} A^{(\lambda)} =\bigoplus_{\lambda\in \Lambda} \left[\sum_{k=1}^{\min(d_{\lambda}', d_\lambda)} \varsigma_{k}^{(\lambda)} |\widetilde{\psi}^{(\lambda)}_k\rangle\langle{{\psi}^{(\lambda)}_k}| \right]
\end{equation}
and for $\Pi:=\ketbra{0}^{\otimes a}\otimes \1_d $ and $\Pi':=\ketbra{0}^{\otimes a}\otimes \1_{d'} $,
\begin{equation}
    \Pi' U_{A} \Pi = \ketbra{0}^{\otimes a}\otimes \bigoplus_{\lambda\in \Lambda} \left[\sum_{k=1}^{\min(d_{\lambda}', d_\lambda)} \varsigma_{k}^{(\lambda)}  |\widetilde{\psi}^{(\lambda)}_k\rangle\langle{{\psi}^{(\lambda)}_k}| \right]
\end{equation}
This corresponds to the singular value decomposition of $\Pi' U_{A} \Pi$, and thus the corresponding QSVT circuit $U_{\Phi}$ for a polynomial $f$ satisfies
\begin{align}
    \Pi' U_{\Phi} \Pi &= f^{(\rm SV)}\left(\Pi' U_{A} \Pi\right)=\begin{cases}
        \ketbra{0}^{\otimes a} \otimes \bigoplus_{\lambda\in \Lambda} \left[\sum_{k=1}^{\min( d_{\lambda}', d_\lambda) } f(\varsigma_{k}^{(\lambda)}) |\widetilde{\psi}^{(\lambda)}_k\rangle\langle{{\psi}^{(\lambda)}_k}|\right],~~~\mbox{if}~N~\mbox{is~odd},\\
        \ketbra{0}^{\otimes a} \otimes \bigoplus_{\lambda\in \Lambda}  \left[\sum_{k=1}^{\min( d_{\lambda}', d_\lambda)} f(\varsigma_{k}^{(\lambda)}) |{\psi}^{(\lambda)}_k\rangle\langle{{\psi}^{(\lambda)}_k}\right],~~~\mbox{if}~N~\mbox{is~even}
    \end{cases}
\end{align}
from Theorem~17 in Ref.~\cite{gilyen2019quantum}.
This immediately leads to 
\begin{align}\label{eq:general_symmetric_QSVT}
    \qty(\bra{0}_a \otimes \1 ) U_{\Phi} \qty( \ket{0}_a\otimes \1)
    & =
    \begin{cases}
        \bigoplus_{\lambda\in \Lambda} \left[\sum_{k=1}^{\min( d_{\lambda}', d_\lambda) } f(\varsigma_{k}^{(\lambda)}) |\widetilde{\psi}^{(\lambda)}_k\rangle\langle{{\psi}^{(\lambda)}_k}|\right],~~~\mbox{if}~N~\mbox{is~odd},\\
        \bigoplus_{\lambda\in \Lambda}  \left[\sum_{k=1}^{\min( d_{\lambda}', d_\lambda)} f(\varsigma_{k}^{(\lambda)}) |{\psi}^{(\lambda)}_k\rangle\langle{{\psi}^{(\lambda)}_k}\right],~~~\mbox{if}~N~\mbox{is~even}
    \end{cases}\notag\\
    &=\bigoplus_{\lambda\in \Lambda} \left[f^{(\rm SV)}(A^{(\lambda)}) \right].
\end{align}
Note that if $A^{(\lambda)}$ is Hermitian, $f^{(\rm SV)}(A^{(\lambda)})=f(A^{(\lambda)})$ holds, and thus, Eq.~\eqref{eq:general_symmetric_QSVT} recovers Lemma~\ref{lem:subspace-qsvt}. 


As the application of subspace QSVT and the key technique in the adaptive QGE algorithm enhanced by symmetry condition, the singular value amplification on specific subspaces is used in Section \ref{subsec:Symmetry-QGE-algorithm}. In what follows, we show the details of proof of Lemma \ref{lem:Amplitude_Amp_subspace}.

\noindent
\textit{Proof of Lemma \ref{lem:Amplitude_Amp_subspace}.} The uniform singular value amplification circuit is designed to implement the linear function $P(x) = \gamma \cdot x$ with $\varepsilon$-precision on the domain $[-(-1+\delta)/\gamma, (1-\delta)/\gamma]$ \cite{gilyen2019quantum}. By applying this circuit to $U_O$ and Lemma \ref{lem:subspace-qsvt}, we obtain an $(a+1)$-block-encoding of $\tilde{O}$ satisfying
    \begin{align}
        \tilde{O} = \bigoplus_{\lambda \in \Lambda} P(O^{(\lambda)}).
    \end{align}
    Restricting to the subspace \(\Delta \subseteq \Lambda\) via the projector \(\Pi_\Delta\), we have $\Pi_\Delta \tilde{O} \Pi_\Delta = \bigoplus_{\lambda \in \Delta} P(O^{(\lambda)})$. Thus, if all singular values of \(\bigoplus_{\lambda \in \Delta} O^{(\lambda)}\) lie within the domain \(\left[-\frac{1-\delta}{\gamma}, \frac{1-\delta}{\gamma}\right]\) (i.e., if \(\|\Pi_\Delta O \Pi_\Delta\| \leq \frac{1-\delta}{\gamma}\)), then the amplified operator \(\Pi_\Delta \tilde{O} \Pi_\Delta\) approximates \(\Pi_\Delta (\gamma O) \Pi_\Delta\) with \(\varepsilon\)-precision. 
\hfill \hfill \qed

\section{Technical lemmas}
In this section, we provide several technical but less central facts. Section~\ref{appsub:proof_matrix_ineq_subgaussian} presents an explicit proof of the matrix Bernstein inequality. Section~\ref{appsub:eval_tail_prob} derives a tighter probability bound for the median. Section~\ref{appsub:construction-kRDM} describes the implementation of block-encodings for observables associated with $k$-RDM elements. Finally, Section~\ref{appsub:matrix_norm_eval} provides a brief proof of the equality in Eq.~\eqref{eq:complex_formula}.

\subsection{The proof of Eq.\eqref{eq:matrix_inequality_for_subgaussian} }
\label{appsub:proof_matrix_ineq_subgaussian}

In this section, we show a detailed proof of the Bernstein inequality for a sum of Hermitian operators with independent, symmetrically distributed random variable coefficients based on Theorem 6.6.1 in Ref.~\cite{tropp2015introduction} and Theorem 35 in Ref.~\cite{apeldoorn2023quantum}. Here, for $d \times d$ matrix $A$ and $B$, we denote $B \geq A$ when $B - A$ is positive semidefinite.  

To prove the Bernstein inequality, we state some results from Ref.~\cite{tropp2015introduction},
\begin{lemma}[Transfer Rule, Proposition 2.1.4 in Ref.~\cite{tropp2015introduction}]
    \label{lem:Transfer_Rule}
    Let $f$ and $g$ be real-valued functions defined on $\mathbb{R} $, and let $A$ be a Hermitian matrix whose eigenvalues are contained in $I \subseteq \mathbb{R}$. Then, $\forall a \in I, f(a) \leq g(a) \Rightarrow f(A) \leq g(A)$.
\end{lemma}

\begin{lemma}[Master Bound for a Sum of Independent Random Matrices, Theorem 3.6.1 in Ref.~\cite{tropp2015introduction}]
    \label{lem:Master_Bound}
   Let $H_1, \ldots, H_M$ be $d \times d$ Hermitian matrices with $\norm{H_j} \leq 1$ for $j \in \{1,\ldots,M\}$. Let $\{\lambda_j\}_{j=1}^M$ be a sequence of independent, symmetrically distributed random variables supported on $[-1,1]$. Then, for all $t \in \mathbb{R}$ 
   \begin{align}
       \mtr{Pr}_{\lambda_1, \ldots, \lambda_M} \qty[ \sigma_{\max} (Y) \geq t] 
        &\leq \inf_{\theta > 0 } e^{-\theta t} tr\qty( \exp(\sum_j \log (\mathbb{E} [e^{\theta \lambda_j H_j}]))), \\ 
        \mtr{Pr}_{\lambda_1, \ldots, \lambda_M} \qty[ \sigma_{\min} (Y) \leq t] 
        &\leq \inf_{\theta < 0 } e^{-\theta t} \tr\qty( \exp(\sum_j \log (\mathbb{E} [e^{\theta \lambda_j H_j}]))).
   \end{align}
   where $\sigma_{\max}(Y) , \sigma_{\min}(Y) $ represent the maximum and minimum eigenvalue for $Y$, respectively.
\end{lemma}

\begin{lemma}[Matrix Bernstein]
    Let $H_1, \ldots, H_M$ be $d \times d$ Hermitian matrices with $\norm{H_j} \leq 1$ for $j \in \{1,\ldots,M\}$. Let $\{X_j\}_{j=1}^M$ be a sequence of independent, symmetrically distributed random variables supported on $[-1,1]$ and let $Y = \sum^M_{j=1} X_j H_j$. Then, $v(Y) \leq \mathbb{E}[X_1^2] \norm{\sum_j H_j^2}$ and 
    \begin{align}
        \label{eq:Matrix_Bernstein}
        \mtr{Pr}_{X_1, \ldots, X_M } \qty[ \norm{Y} \geq t] \leq 2d \exp(-\frac{t^2/2}{v(Y) + t/3 }).
    \end{align}
\end{lemma}

\begin{proof}
   Let $\theta$ be a parameter, $X$ a random variable supported on $[-1,1]$, and $H$ a $d \times d$ Hermitian matrix. We begin by expanding the exponential $e^{\theta X H}$ as
\begin{align}
    e^{\theta X H} = \1 + \theta X H + \left(e^{\theta X H} - \1 - \theta X H\right) = \1 + \theta X H + \lambda^2 H f(X H) H,
\end{align}
where the function $f: \mathbb{R} \to \mathbb{R}$ is defined by $f(x) \coloneq \frac{e^{\theta x} - \theta x - 1}{x^2}$ for $x \neq 0$, and $f(0) \coloneq \theta^2 / 2$. Since $f(x)$ is monotonically increasing, $x \leq L$ implies $f(x) \leq f(L)$. From Lemma \ref{lem:Transfer_Rule}, $\norm{X H} \leq 1$ leads to $f(X H) \leq f(1)\cdot \1$, yielding
\begin{align}
    e^{\theta X H} \leq \1 + \theta X H + f(1)\cdot (X H)^2.
\end{align}
We further bound $f(1)$ as
\begin{align}
    f(1) = e^{\theta} - \theta - 1 = \sum_{k=2}^{ \frac{\theta^k}{k!}} \leq \frac{\theta^2}{2} \sum_{k=2} \left(\frac{\theta}{3}\right)^{k-2} = \frac{\theta^2 / 2}{1 - \theta / 3}.
\end{align}
Thus, we obtain the bound
\begin{align}
    e^{\theta X H} \leq \1 + \theta X H + \frac{\theta^2 / 2}{1 - \theta / 3} (X H)^2.
\end{align}
Taking the expectation over $X$, we can obtain $\mathbb{E}_X[\theta X H] = 0$ because $X$ is symmetrically distributed. Hence, we calculate
\begin{align}
    \mathbb{E}_X \left[e^{\theta X H}\right] \leq \1 + \frac{\theta^2 / 2}{1 - \theta / 3} \mathbb{E}_X[X^2] H^2 \leq \exp\left(\frac{\theta^2 / 2}{1 - \theta / 3} \mathbb{E}_X[X^2] H^2\right).
\end{align}
Since the logarithm is operator monotone for positive definite matrices, it follows that
\begin{align}
    \log \left( \mathbb{E}_X \left[e^{\theta X H}\right] \right) \leq \frac{\theta^2 / 2}{1 - \theta / 3} \mathbb{E}_X[X^2] H^2.
\end{align}
For a sequence of independent, symmetrically distributed random variables $\{X_j\}_{j=1}^M$, we sum over $j$ to obtain
\begin{align}
    \sum_j \log \left( \mathbb{E}_{X_j} \left[e^{\theta X_j H_j}\right] \right) \leq \frac{\theta^2 / 2}{1 - \theta / 3} \sum_j \mathbb{E}[X_j^2] H_j^2.
\end{align}
By applying the operator trace inequality $A \leq B \Rightarrow \tr\left(e^A\right) \leq \tr\left(e^B\right)$, we derive the tail bound for the maximum eigenvalue of $Y \coloneq \sum_j X_j H_j$:
\begin{align}
    \mtr{Pr}_{X_1, \ldots, X_M} \left[ \sigma_{\max}(Y) \geq t \right]
    &\leq \inf_{\theta > 0} \tr \left( \exp\left( \sum_j \log \mathbb{E} \left[ e^{\theta X_j H_j} \right] \right) \right) e^{-\theta t} \\
    &\leq \inf_{\theta > 0} d e^{-\theta t} \exp\left( \frac{\theta^2 / 2}{1 - \theta / 3} \sum_j \mathbb{E}[X_j^2] \norm{H_j}^2 \right).
\end{align}
Choosing $\theta = \frac{t}{v(Y) + t / 3}$ with $v(Y) \coloneq \sum_j \mathbb{E}[X_j^2] \norm{H_j}^2$,
we obtain the following inequality:
\begin{align}
    \mtr{Pr}_{\lambda_1, \ldots, \lambda_M} \left[ \sigma_{\max}(Y) \geq t \right] \leq d \exp\left( -\frac{t^2 / 2}{v(Y) + t / 3} \right).
\end{align}


The same argument applies to the eigenvalue of $Y$. We consider the following probability
\begin{align}
    \mtr{Pr}[\sigma_{\min}(Y) \leq -t] = \mtr{Pr}[\sigma_{\max} (-Y) \geq t].
\end{align}
Since $X_j$ in $Y$ is symmetrically distributed, 
\begin{align}
    \inf_{\theta > 0} \tr \left( \exp\left( \sum_j \log \mathbb{E} \left[ e^{-\theta X_j H_j} \right] \right) \right) e^{-\theta t} 
    = \inf_{\theta > 0} \tr \left( \exp\left( \sum_j \log \mathbb{E} \left[ e^{\theta X_j H_j} \right] \right) \right) e^{-\theta t} .
\end{align}
Hence, we can obtain the same form inequality, 
\begin{align}
    \mtr{Pr}_{\lambda_1, \ldots, \lambda_M} \left[ \sigma_{\min}(Y) \leq -t \right] \leq d \exp\left( -\frac{t^2 / 2}{v(Y) + t / 3} \right).
\end{align}
Combining these facts, we obtain the matrix Bernstein inequality \eqref{eq:Matrix_Bernstein}. 

\end{proof}





\subsection{Tighter probability bound for median estimate}
\label{appsub:eval_tail_prob}

In this section, we present a tighter probability bound for the median of independent samples, improving upon Hoeffding's inequality. Notably, our target bound involves the cumulative distribution function (CDF), which is generally intractable to compute analytically. However, the CDF can be evaluated numerically, leading to a reduced constant factor in the total query complexity of our proposal in numerical calculations.

\begin{lemma}
Let $R$ be a positive integer. Suppose that $k^{(r)}$ for $r \in \{1, 2, \ldots, R\}$ are independent random variables satisfying $\mtr{Pr}\qty[\abs{k^{(r)} - k^*} > \varepsilon] \leq \mu$ for $\varepsilon > 0$ and $0 < \mu < \frac{1}{2}$. Then, the probability that the median deviates from the true value exceeds $\varepsilon$ is bounded as
\begin{align}
    \mtr{Pr}\qty[\abs{k^{(\mathrm{med})} - k^*} > \varepsilon] \leq 1 - F\qty(\left\lceil \frac{R}{2} \right\rceil - 1; \mu, R),
\end{align}
where $k^{(\mathrm{med})}$ denotes the median of $\{k^{(r)}\}_{r=1}^R$, and $F(n; \mu, R)$ is the cumulative distribution function of the binomial distribution with success probability $\mu$ and $R$ trials.
    
\end{lemma}
\begin{proof} 
By definition of the median, the event $\{\abs{k^{(\text{med})} - k^*} > \varepsilon\}$ implies that at least $R/2$ of the samples $\{k^{(r)}\}_{r=1}^R$ deviate from $k^*$ by more than $\varepsilon$, yielding
\begin{align}
    \left\{ \abs{k^{(\text{med})} - k^*} > \varepsilon \right\} \subseteq \left\{ \sum_{r=1}^R \chi\qty[\abs{k^{(r)} - k^*} > \varepsilon] \geq \frac{R}{2} \right\},
\end{align}
where $\chi[\cdot]$ denotes the indicator function.
Defining $Z_r \coloneq \chi\qty[\abs{k^{(r)} - k^*} > \varepsilon]$, we have independent Bernoulli random variables with its mean satisfying $\mathbb{E}[Z_r] \leq \mu$. It follows that
\begin{align}
    \mtr{Pr}\qty[\abs{k^{(\text{med})} - k^*} > \varepsilon]
    &\leq \mtr{Pr}\qty[\sum_{r=1}^R Z_r \geq \frac{R}{2}] \\
    &= 1 - \mtr{Pr}\qty[\sum_{r=1}^R Z_r < \frac{R}{2}] \\
    &= 1 - \mtr{Pr}\qty[\sum_{r=1}^R Z_r \leq \left\lfloor \frac{R+1}{2} \right\rfloor - 1].
\end{align}
For a non-negative integer $n \leq R$, we introduce the CDF of the binomial distribution:
\begin{align}
    F(n; \mu, R) \coloneq  \sum_{m=0}^n \binom{R}{m} \mu^m (1 - \mu)^{R - m}.
\end{align}
Using this notation, the failure probability is compactly expressed as
\begin{align}
    \mtr{Pr}\qty[\abs{k^{(\text{med})} - k^*} > \varepsilon] \leq 1 - F\qty(\left\lfloor \frac{R+1}{2} \right\rfloor - 1; \mu, R).
\end{align}
\end{proof}

Note that $\Phi(R) \coloneq 1 - F\qty(\left\lfloor \frac{R+1}{2} \right\rfloor - 1; \mu, R)  $ is upper bounded by $\exp(-2 R(1/2 -\mu)^2 )$ by using Hoeffding's inequality. 

\subsection{Explict implementation of block-encoding for fermionic $k$-RDM estimation}
\label{appsub:construction-kRDM}

In this section, we examine the explicit implementation of block-encodings for operators used to evaluate fermionic $k$-RDMs. In order to estimate the value of fermionic $k$-RDM elements with QAE algorithm and our proposals,
we must develop concrete implementations of the block-encodings of target observables. Since the $k$-RDM elements are not always real numbers, we decompose the target values into real and imaginary parts. Let us consider a $k$-body operator as 
\begin{align}
    A^{\bm p}_{\bm q} =  a_{p_1}^\dagger \cdots a_{p_k}^\dagger a_{q_1} \cdots a_{q_k}.
\end{align}
Then, the real and imaginary parts of expectation values can be obtained by measuring the following operators
\begin{align}
    \mtr{Re} [A^{\bm p}_{\bm q} ] \coloneq \frac{ A^{\bm p}_{\bm q} +  A^{\bm q}_{\bm p}}2, \quad 
    \mtr{Im} [A^{\bm p}_{\bm q} ] \coloneq \frac{ A^{\bm p}_{\bm q} -  A^{\bm q}_{\bm p}}{2i}. 
\end{align}
Note that if $\bm p = \bm q $, we only consider ${A^{\bm p}_{\bm p}}$ and that $ \mtr{Re} [A^{\bm p}_{\bm q} ] =  \mtr{Re} [A^{\bm q}_{\bm p} ], \mtr{Im} [A^{\bm p}_{\bm q} ] =  -\mtr{Im} [A^{\bm q}_{\bm p} ]$, so we can generally take $p < q$. 

We first describe the block encoding for the case of 1-RDM.
By Jordan-Wigner transformation, the operators are expressed as 
\begin{gather}
    2\mtr{Re}[A^p_q] = \frac{ X_p \vec{Z} X_q + Y_p \vec{Z} Y_q }2, (p < q) \\
    2\mtr{Im}[A^p_q] = \frac{ X_p \vec{Z} X_q -  Y_p \vec{Z} Y_q }{2}, (p < q) \\
    {A^p_p} = \frac{\1 - Z_p}2
\end{gather}
where the notation $P_p \vec{Z} P_q$ denotes the operator $A_p Z_{p+1} \cdots Z_{q-1} A_q $ for $ p < q $.
From these expressions, we can easily construct block-encoding of these observables with one ancilla qubit: 
\begin{align}
    U_{p,q}^{(1)} &\coloneq \ketbra{0} \otimes X_p \vec{Z} X_q + \ketbra{1} \otimes Y_p \vec{Z} Y_q, \quad \text{for $p < q$.} \\
    U_{p,q}^{(1, \mtr{Re})} &\coloneq (H \otimes \1)U_{p,q}^{(1)}  (H \otimes \1) \Rightarrow  (\bra{0} \otimes \1)U_{p,q}^{(1, \mtr{Re} )} (\ket{0} \otimes \1) = 2\mtr{Re}[A^p_q] \\
     U_{p,q}^{(1, \mtr{Im})} &\coloneq (H \otimes \1)U_{p,q}^{(1)}  (ZH \otimes \1) \Rightarrow  (\bra{0} \otimes \1)U_{p,q}^{(1, \mtr{Im} )} (\ket{0} \otimes \1) = 2\mtr{Im}[A^p_q] \\
     U_{p,p}^{(1)} &\coloneq  (H \otimes \1) [ \ketbra{0} \otimes \1 
     + \ketbra{1} \otimes Z_p ] (ZH \otimes \1).  
\end{align}

For 2-RDMs, 
it is beneficial to provide the following identites:
\begin{align}
\label{eq:identity-1}
\text{Re}(A^{p_1}_{q_1} A^{p_2}_{q_2}) = \text{Re}(A^{p_1}_{q_1}) \cdot \text{Re}(A^{p_2}_{q_2}) - \text{Im}(A^{p_1}_{q_1}) \cdot \text{Im}(A^{p_2}_{q_2}), \\
\label{eq:identity-2}
\text{Im}(A^{p_1}_{q_1} A^{p_2}_{q_2}) = \text{Re}(A^{p_1}_{q_1}) \cdot \text{Im}(A^{p_2}_{q_2}) + \text{Im}(A^{p_1}_{q_1}) \cdot \text{Re}(A^{p_2}_{q_2}).
\end{align}
From the anti-commutative property of fermion operators as $a_{p_1}^\dagger a_{q_1} a_{p_2}^\dagger a_{q_2} = -a_{p_1}^\dagger a_{p_2}^\dagger a_{q_1} a_{q_2}$, we further derive 
\begin{align}
    \mtr{Re} [A^{\bm p}_{\bm q}]  = -  \text{Re}(A^{p_1}_{q_1}) \cdot \text{Re}(A^{p_2}_{q_2}) + \text{Im}(A^{p_1}_{q_1}) \cdot \text{Im}(A^{p_2}_{q_2}), \\
    \mtr{Im} [A^{\bm p}_{\bm q}]  = -  \text{Re}(A^{p_1}_{q_1}) \cdot \text{Im}(A^{p_2}_{q_2}) - \text{Im}(A^{p_1}_{q_1}) \cdot \text{Re}(A^{p_2}_{q_2}).
\end{align}
From this expression, when $p_1 < q_1 <  p_2 < q_2 $, then 
block-encoding is expressed as 
\begin{align}
    U_{\bm p, \bm q}^{(2, \mtr{Re})} &\coloneq (H \otimes \1 ) \qty[\ketbra{0} \otimes U_{p_1,q_1}^{(1,\mtr{Im } )} U_{p_2,q_2}^{(1,\mtr{Im })} + 
    \ketbra{1} \otimes U_{p_1,q_1}^{(1,\mtr{Re})} U_{p_2,q_2}^{(1,\mtr{Re })} ](ZH \otimes \1 ) \\ 
    U_{\bm p, \bm q}^{(2, \mtr{Im})} &\coloneq -(H \otimes \1 ) \qty[\ketbra{0} \otimes U_{p_1,q_1}^{(1,\mtr{Re } )} U_{p_2,q_2}^{(1,\mtr{Im })} + 
    \ketbra{1} \otimes U_{p_1,q_1}^{(1,\mtr{Im})} U_{p_2,q_2}^{(1,\mtr{Re })} ] (H \otimes \1 )
\end{align}
These implementations yield block-encodings of $2\mtr{Re}[A^{\bm p}_{\bm q}]$ and $2\mtr{Im}[A^{\bm p}_{\bm q}]$, respectively, in the case where  $p_1 < q_1 <  p_2 < q_2 $. For cases in which some indices coincide, we disregard the contribution from some imaginary parts. This construction ensures that the norm of the embedded observables remains $1$, avoiding the shrink in the norm of observables.
We remark that, even in the case where the arguments are not ordered as such, we may construct the block encoding by modifying the structure of the Z string.

For higher-order RDM elements, $A^{\bm p}_{\bm q}$ can be decomposed into a linear combination of products of $1$-RDM terms, as in Eqs.~\eqref{eq:identity-1} and \eqref{eq:identity-2}. Following the same procedure of $2$-RDM elements, block-encodings for all higher-order RDM elements can be systematically constructed.

\subsection{The evaluation of Eq.\eqref{eq:complex_formula} }
\label{appsub:matrix_norm_eval}

\begin{lemma}
    \label{lem:eval_combination_series}
    For integers $k,\eta$, if $0 \leq k \leq \eta$ and $\eta+k \leq N$ holds, then,  
    \begin{align}
        \sum^{k}_{m=0}  \mqty(N \\ k)\mqty(N- k \\ k - m) \mqty(k \\ m )  \cdot \mqty(N-(2k - m ) \\ \eta - k ) / \mqty(N \\ \eta)
        = \mqty(\eta \\ k) \mqty(N - \eta  + k \\ k  )
    \end{align}
\end{lemma}

\begin{proof}
    Firstly we rewrite the LHS of the target equation,
    \begin{align}
    \label{eq:original_form}
        \sum^{k}_{m=0}  \mqty(N \\ k)\mqty(N- k \\ k - m) \mqty(k \\ m )  \cdot \mqty(N-(2k - m ) \\ \eta - k ) / \mqty(N \\ \eta)
        = \frac{\mqty(N \\ k)}{\mqty(N \\ \eta)} \sum^k_{m=0} \mqty(N-k \\ k-m) \mqty(k \\ m ) \mqty(N-(2k-m) \\ \eta - k).
    \end{align}
    For the part of sum, we introduce a new variable $j = k-m$ and obtain,
    \begin{align}
        \label{eq:Alice}
         \sum^k_{m=0} \mqty(N-k \\ k-m) \mqty(k \\ m ) \mqty(N-(2k-m) \\ \eta - k)= \sum^k_{j=0} \mqty(N-k \\ j) \mqty(k \\ k-j) \mqty(N - (k+j) \\ \eta - k) .
    \end{align}
    Regarding the last term, we can demonstrate the following identity,
    \begin{align}
        \mqty(N - (k+j)  \\ \eta - k ) = \mqty(N - k \\ \eta -k) \cdot 
        \frac{ (-(N-\eta) )_j }{(-(N-k))_j},
    \end{align}
    where $(a )_j $ is a Pochhammer symbol defined for $j \geq 0$ by
    \begin{align}
        (a)_j \coloneq \begin{cases} 
        1, &\quad (j=0) \\
        a (a+1) \cdots (a+j-1), &\quad \text{otherwise.}
        \end{cases}
    \end{align}
    From the fact $\mqty(c \\ j) = \frac{(-1)^j}{j !} (-c)_j$ for $c \geq j$, we can express Eq.~\eqref{eq:Alice} as follows,
    \begin{align}
        \sum^k_{j=0} \mqty(N-k \\ j) \mqty(k \\ k-j) \mqty(N - (k+j) \\ \eta - k) 
        &= \mqty( N -k \\ \eta -k ) \sum^k_{j =0} \frac{(-1)^j}{j!} \qty(-(N-k))_j \cdot \frac{ (-(N-\eta))_j }{ (-(N-k))_j } \mqty(k \\ k-j) \\
        &= \mqty( N -k \\ \eta -k ) \sum^k_{j =0} \mqty(N-\eta \\ j) \mqty(k \\k- j) \\
        &= \mqty( N -k \\ \eta -k )  \mqty(N -\eta + k \\ k ).
    \end{align}
    The last equation comes from Vandermonde's convolution. Inserting this result to Eq.\eqref{eq:original_form}, we have
    \begin{align}
        \frac{\mqty(N \\ k) \cdot \mqty(N-k \\ \eta- k )}{\mqty(N \\ \eta)} \mqty(N -\eta + k \\ k ) = \mqty(\eta \\ k) \mqty(N - \eta  + k \\ k  ).
    \end{align}
\end{proof}

\end{document}